\documentclass[review]{elsarticle}
\usepackage[letterpaper,top=2.5cm,bottom=2.5cm,left=2.6cm,right=2.6cm,marginparwidth=1.75cm]{geometry}
\usepackage{lineno}
\usepackage[dvipsnames]{xcolor}
\usepackage{tikz}
\usetikzlibrary{backgrounds}
\usetikzlibrary{arrows,shapes}
\usetikzlibrary{tikzmark}
\usetikzlibrary{calc}
\usepackage{amsmath}
\usepackage{amsthm}
\usepackage{amssymb}
\usepackage{mathtools, nccmath}
\usepackage{wrapfig}
\usepackage{comment}
\usepackage{makecell}
\usepackage{ulem}
\usepackage{bm}
\usepackage{cancel}
\usepackage{lscape}
\usepackage{url}

\newcolumntype{Y}{>{\centering\arraybackslash}X} 

\usepackage[colorlinks=true, linkcolor=red, urlcolor=red]{hyperref}
\usepackage{xspace}

\usepackage{array}
\usepackage{ragged2e}
\newcolumntype{P}[1]{>{\RaggedRight\hspace{0pt}}p{#1}}
\newcolumntype{X}[1]{>{\RaggedRight\hspace*{0pt}}p{#1}}

\usepackage{tcolorbox}

\usepackage{tikz}
\usetikzlibrary{arrows,shapes,positioning,shadows,trees,mindmap}
\usepackage[edges]{forest}
\usetikzlibrary{arrows.meta}
\colorlet{linecol}{black!75}
\usepackage{xkcdcolors} 
\usepackage{breakurl}

\usepackage{tikz}
\usetikzlibrary{backgrounds}
\usetikzlibrary{arrows,shapes}
\usetikzlibrary{tikzmark}
\usetikzlibrary{calc}

\newcommand{\thicktoprule}{\Xhline{4\arrayrulewidth}}
\newcommand{\thickbottomrule}{\Xhline{4\arrayrulewidth}}
\newcolumntype{Y}{>{\centering\arraybackslash}X}
\newcolumntype{Z}{>{\raggedright\arraybackslash}X}
\colorlet{mhpurple}{Plum!80}

\usepackage{lineno,amsmath,amsfonts,amssymb,amsthm,mathrsfs}
\hypersetup{colorlinks=true,linkcolor=blue}
\usepackage[linesnumbered,ruled,vlined,onelanguage]{algorithm2e}
\usepackage{booktabs}
\usepackage[skip=0pt]{caption,subcaption}
\usepackage{float}
\usepackage{longtable}
\usepackage{setspace}
\usepackage{tabularx}
\usepackage{booktabs}
\usepackage{tablefootnote}

\DeclareMathOperator*{\argmax}{arg\,max}
\usepackage{todonotes}
\journal{Transportation Research Part C: Emerging Technologies}
\biboptions{authoryear}

\newtheorem{lemma}{Lemma}

\newtheorem{definition}{Definition}
\begin{document}

\begin{frontmatter}

\title{Unraveling Stochastic Fundamental Diagrams with Empirical Knowledge: Modeling, Limitations, and Future Directions}

\author[1]{Yuan-Zheng Lei}
\author[1]{Yaobang Gong}
\author[1]{Xianfeng Terry Yang*} 
\ead{xtyang@umd.edu}

\address[1]{Department of Civil and Environmental Engineering, University of Maryland, 1173 Glenn L.Martin Hall, College Park, MD 20742, United States}
\begin{abstract}
\par Traffic flow modeling relies heavily on fundamental diagrams. However, deterministic fundamental diagrams, such as single or multi-regime models, cannot capture the underlying uncertainty in traffic flow. To address this limitation, this study proposes a non-parametric Gaussian process model to formulate the stochastic fundamental diagram. Unlike parametric models, the non-parametric approach is insensitive to parameters, flexible, and widely applicable. The computational complexity and high memory requirements of Gaussian process regression are also mitigated by introducing sparse Gaussian process regression. This study also examines the impact of incorporating empirical knowledge into the prior of the stochastic fundamental diagram model and assesses whether such knowledge can enhance the model's robustness and accuracy. By using several well-known single-regime fundamental diagram models as priors and testing the model's performance with different sampling methods on real-world data, this study finds that empirical knowledge benefits the model only when small inducing samples are used with a relatively clean dataset. In other cases, a purely data-driven approach is sufficient to estimate and describe the density-speed relationship pattern.
\end{abstract}
\end{frontmatter}

\section{Introduction} \label{1}
\par The traffic fundamental diagram is a crucial concept that provides a quantitative framework for understanding and estimating traffic behavior in different scenarios. Its application is widespread, including traffic flow theory, highway capacity analysis, and the design of intelligent transportation systems. Therefore, it plays a vital role in the planning and operating of efficient transportation networks.
\par The fundamental diagram is used to describe the following relationship:
\begin{equation}
    q = \rho \times v \label{eq:1}
\end{equation}
where $q$, $\rho$, and $v$ represent flow, density, and space mean speed. Based on \cite{ni2015traffic}, it can also be regarded as an equilibrium traffic flow model since the fundamental diagram represents a steady-state or equilibrium relationship between flow, density, and speed. In this paper, the density-speed relationship is the primary focus. Given this relationship, the other two relationships can be determined using Eq \ref{eq:1}. Traditional density-speed fundamental diagram models can be classified into single-regime and multi-regime models, as shown in Figure \ref{figure:3}. Single-regime models usually describe the relationship between traffic density and speed using a linear or nonlinear function. However, single-regime models may not always fit well in all density ranges. This is why multi-regime models use piecewise functions to represent the density-speed relationship across the entire range of densities.
\par Almost all traditional single and multi-regime fundamental density-speed diagram models are shown in a deterministic form, which can be briefly represented by:
\begin{equation}
    v = f(\rho) \label{eq:2}
\end{equation}
\par In real-life traffic situations, it is difficult to find a steady-state condition. Therefore, when a traffic density $\rho$ is given, the actual observed speed is more likely to fall within a specific range instead of a fixed value. Adopting the idea presented in \cite{ni2015traffic}, the stochastic fundamental diagram model operates on the principle that the density-speed relationship is not a function of density but a function of density and its distribution. This relationship is expressed mathematically in equation \ref{eq:3}, where $w(\rho)$ indicates the density distribution.
\begin{equation}
    v = f(\rho, w(\rho)) \label{eq:3}
\end{equation}
\par Figure \ref{figure:4} shows the stochastic fundamental diagram, where a certain speed distribution corresponds to a given density $\rho$. In numerous studies on network-level control, such as \cite{gupta2023simple} and \cite{li2024beyond}, the fundamental diagram is utilized to depict the traffic patterns of a route, section, or zone. Many of these studies focus on maximizing flow through signal control or route guidance strategies. In a deterministic fundamental diagram model, the target density after control is a constant value since the deterministic fundamental diagram models a curve. In the stochastic fundamental diagram model, the maximum mean flow will associate with a distribution indicating all possibilities, and it may have a wide confidence interval, indicating that it could lead to a very high maximum flow and a very low minimum flow. Therefore, the stochastic fundamental diagram model provides insight into the robust network-level control, suggesting that it is necessary to simultaneously consideration of both the mean and variance of the flow, rather than focusing solely on the mean, and it is better to steer the network toward a traffic density with a relatively high mean flow but a narrow confidence interval.
\begin{figure}[htbp]
    \centering
    \begin{subfigure}{.49\textwidth}
        \centering
        \includegraphics[width=\linewidth]{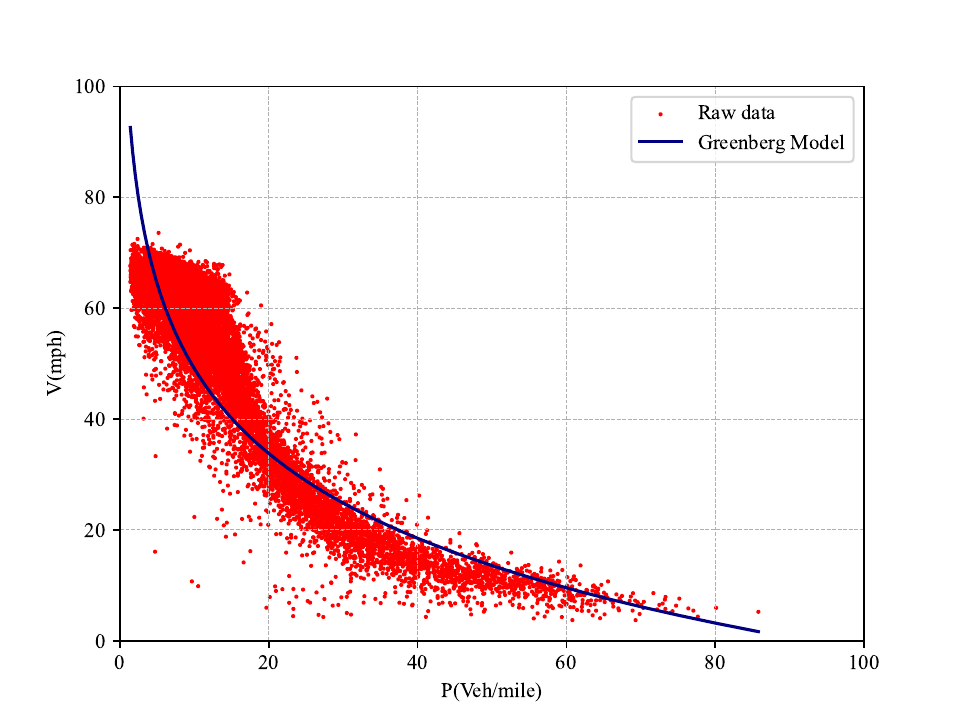}
        \caption{Greenberg Model}
        \label{figure:1}
    \end{subfigure}
    \begin{subfigure}{.49\textwidth}
        \centering
        \includegraphics[width=1.0\textwidth]{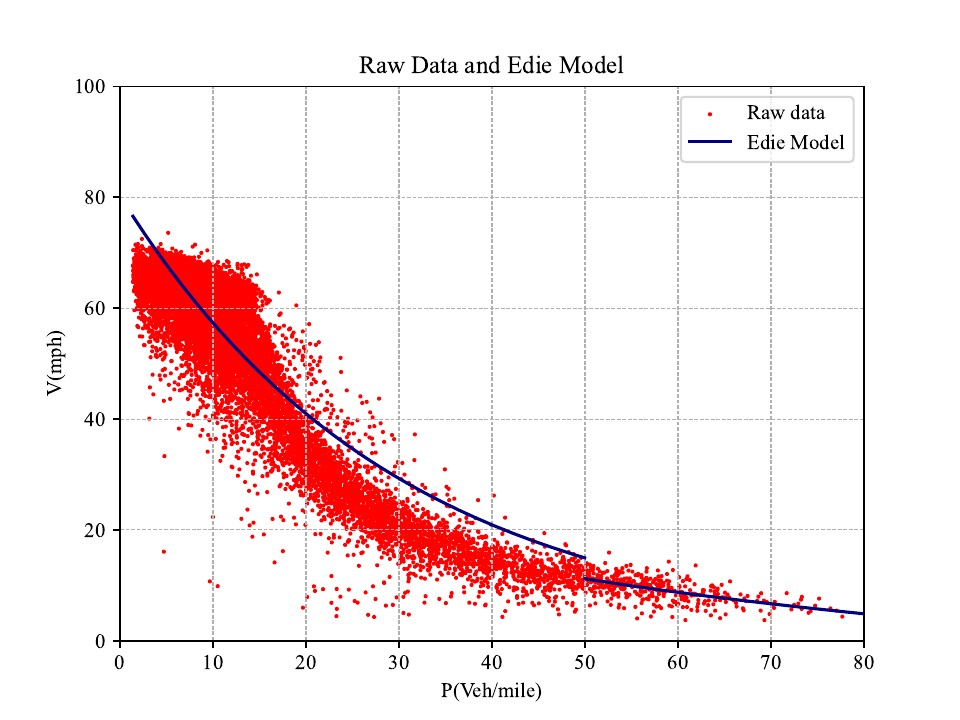}
        \caption{Edie Model}
        \label{figure:2}
    \end{subfigure}
    \caption{Single-regime and multi-regime models}\footnotemark
    \label{figure:3}
\end{figure}
\footnotetext{Both models are calibrated based on a dataset collected on the Georgia State Route 400 and weighted least square method proposed in \cite{qu2015fundamental}}
\begin{figure}
    \centering   \includegraphics[width=0.5\textwidth]{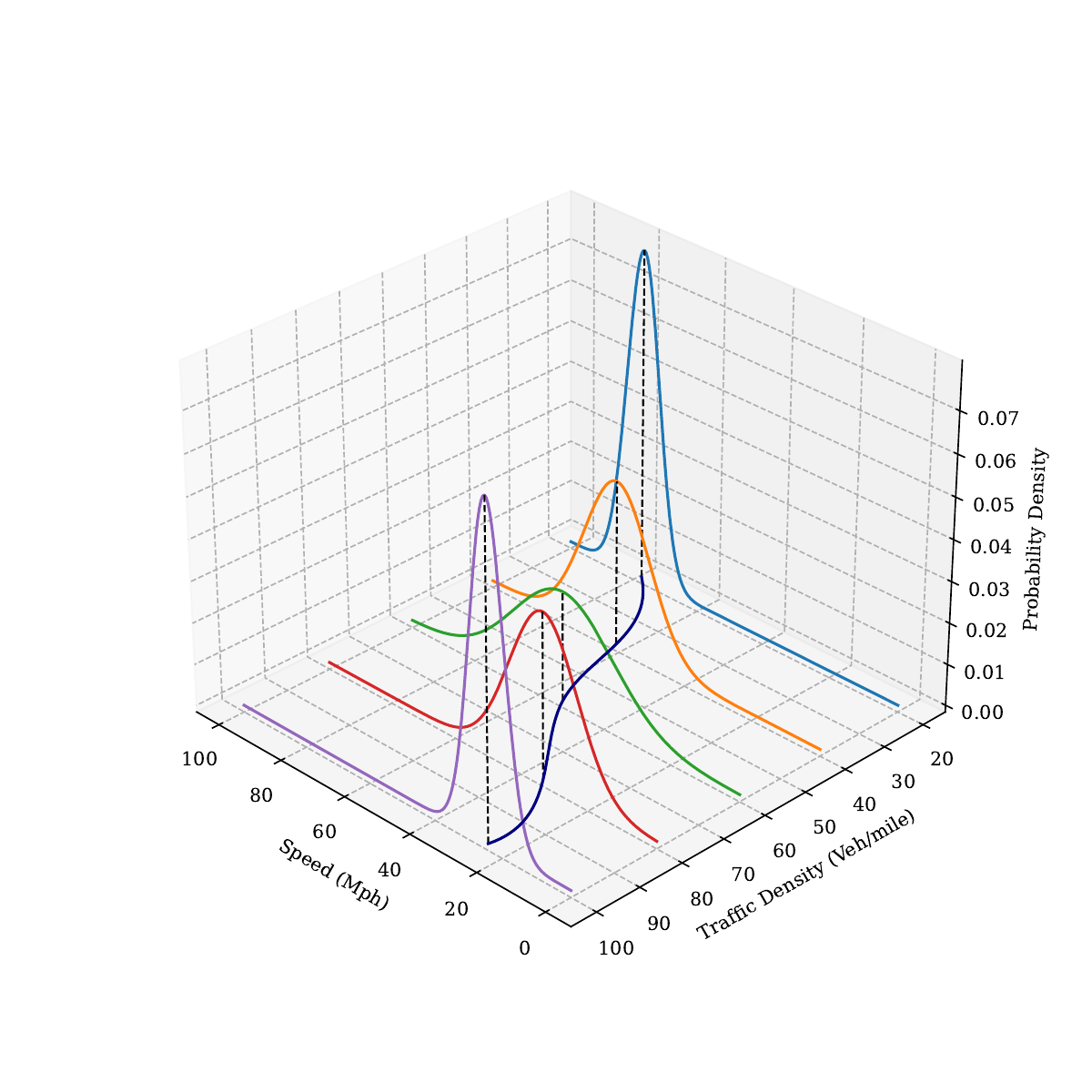}
    \caption{Representation of the stochastic fundamental diagram}
    \label{figure:4}
\end{figure}
\par Stochastic and deterministic fundamental diagram models are essentially regression models. Assuming some data has been collected, which can be expressed as $\mathcal{D} = \{(x_{1},y_{1}),...,(x_{n},y_{n})\}$, a regression model is a model that used to describe the relationship between $\mathbf{y}$ and $\mathbf{x}$ that can be written as:
\begin{equation}
    \mathbf{y} = f(\mathbf{x},\bm{\theta}) + \epsilon \label{eq:4}
\end{equation}
where $\bm{\theta}$ is a vector of $k$ parameters, $\epsilon$ denotes an error term whose distribution may or may not be Gaussian, and $f(\cdot)$ is a function that delineates the relationship between $\mathbf{y}$ and $\mathbf{x}$.
\par Most existing stochastic or deterministic fundamental diagram models are parametric regression models. For a parametric regression model, the $\bm{\theta}$ in Equation \ref{eq:4} is a vector in finite-dimensional space, and $f(\mathbf{x},\bm{\theta})$ need to be known functions and distribution of $\epsilon$ is known prior for the inference (\cite{mahmoud2019parametric}). However, for a non-parametric regression model, the $\bm{\theta}$ in Equation \ref{eq:4} is a subset of an infinite-dimensional vector of vector spaces, and $f(\mathbf{x},\bm{\theta})$ is an unknown function.
\par The flexibility of parametric models is limited by their fixed functional form and assumptions on error distributions, making it challenging to capture complex relationships. Non-parametric models, however, are better suited for modeling complex, nonlinear relationships. Gaussian process regression (GPR) is a powerful non-parametric model for the stochastic density-speed relationship (\cite{yuan2021traffic,yuan2021macroscopic}). Nonetheless, handling extensive datasets presents a challenge for Gaussian Process Regression (GPR), as it demands significant computational resources. Specifically, for a dataset comprising $n$ samples, GPR's computational complexity scales cubically, with a time complexity of $\mathcal{O}(n^3)$. Meanwhile, its space complexity, which indicates the amount of memory needed, grows quadratically, marked by $\mathcal{O}(n^2)$. This makes GPR less feasible for very large datasets without employing variational techniques or approximation methods. To overcome this issue, a sparse Gaussian process regression (SGPR) is used to approximate the full GPR using a subset of data points, for example, $m (m << n)$, and the time complexity will fall into $\mathcal{O}(mn^2)$. Although the accuracy of SGPR is generally considered lower than GPR, it can still be a viable solution for deterministic fundamental diagram models. In the majority of cases, deterministic models are parametric in nature. Consequently, employing the complete dataset for their calibration does not pose significant challenges because its time complexity is generally considered as $\mathcal{O}(c^2n)$ (\cite{li1996new}), where $c$ is the total number of features, $c << m$ in practical. This observation leads to an insightful proposition: calibrating deterministic models initially and subsequently utilizing them as priors over the unknown function may offer a viable strategy. Such an approach potentially facilitates the incorporation of empirical knowledge derived from the entire dataset without incurring additional time and space complexities associated with Sparse Gaussian Process Regression (SGPR).
\par The objective of this paper is twofold. Firstly, inspired by \cite{wang2013stochastic}, \cite{qu2017stochastic}, and \cite{cheng2021s}, it aims to develop a non-parametric model for the density-speed relationship. This model will reflect the equilibrium density-speed relationship based on the mean function curve (the navy curve on the z = 0 plane in Figure \ref{figure:4}) and traffic patterns based on different speed distributions. Secondly, the paper will explore the possibility of bridging the gap between previous deterministic and stochastic models. Specifically, it will investigate the feasibility of incorporating empirical models as a prior or direct empirical model term into the loss function of the stochastic model. The paper will answer the following questions: \textbf{\textit{Is it possible to merge empirical knowledge from deterministic models into the stochastic model?}} \textbf{\textit{If so, how can it be done? And, is this empirical knowledge really helpful for the stochastic model?}}
\par In this paper, our contribution can be concluded as follows:
\begin{itemize}
    \item A sparse non-parametric regression framework is proposed to model the stochastic density-speed relationship. The mean function of the generated Gaussian process can be regarded as the density-speed relationship under the equilibrium state, and the different Gaussian distributions distributed at every single density point represent the overall traffic pattern.
    \item This paper describes how to model empirical knowledge into a stochastic model but also highlights the limitations caused by deterministic model parameters, which require clean data for calibration.
    \item In the first time, to the best of our knowledge, this paper shows that even taking an extremely well-calibrated empirical model as the prior, when the size of inducing samples is over a certain range, the advantages embedded in it may not exist.
\end{itemize}
 \par The structure can be organized as follows: section \ref{2} will review deterministic fundamental diagram models and stochastic fundamental diagram models; section \ref{3} will introduce how deterministic models are calibrated and demonstrate the non-parametric technique that will be used for sparse regression and briefly discuss the sampling algorithms that will be used to select inducing variables for making the approximation; section \ref{4} will discuss the potential influence of empirical models.
\section{Literature review} \label{2}
\subsection{Deterministic fundamental diagram models}
\par Most well-known single and multi-regime fundamental diagram models are deterministic models. The density-speed relationship is represented by a linear or nonlinear function, which can be considered to reflect the density-speed relationship under equilibrium. Given a certain density $\rho$, the corresponding speed is fixed. Based on the modeling methods of the density-speed relationship corresponding to different traffic conditions, the deterministic fundamental diagram models can be classified into single-regime models and multi-regime models.
\subsubsection{Single-regime models}
\par In 1935, \cite{greenshields1935study} proposed the first fundamental traffic flow diagram. This diagram established a linear relationship between speed and density. Since then, a lot of research has been done to enhance both the accuracy and interpretability of this relationship, such as Greenberg's model (\cite{greenberg1959analysis}), Underwood's model (\cite{underwood1961speed}), Newell's model (\cite{newell1961nonlinear}), Drake's model (\cite{drake1965statistical}) and Drew's model (\cite{drew1968traffic}). 
\par There are different types of fundamental diagrams that are classified based on the functional forms of the speed-density relationship. These diagrams can be categorized into three main groups - exponential function-related models, logarithmic function-related models, and polynomial function-related models along with their variants. Exponential function-related models include the models proposed by \cite{newell1961nonlinear}, \cite{drake1965statistical}, \cite{kerner1994structure}, \cite{del1995functional-part-1}, and \cite{wang2011logistic}. On the other hand, \cite{ni2016vehicle} proposed logarithmic function-related models. Polynomial function-related models and their variants are proposed by \cite{pipes1966car}, \cite{drew1968traffic}, \cite{jayakrishnan1995dynamic}, \cite{macnicholas2011simple}, and \cite{cheng2021s}.
\par According to the (\cite{ni2015traffic}), it has been observed that none of the early single-regime models are capable of fitting the empirical data reasonably well over the entire density range. Each model is only suitable for a specific density range. Due to this limitation, researchers started exploring the idea of fitting the data in a piecewise manner using multiple equations, which resulted in the development of multi-regime models. 
\begin{table}[hthp]
  \centering
  \caption{Some single-regime traffic flow model}
  \label{Table:1}
  \begin{tabular}{>{\centering\arraybackslash}p{0.3\linewidth} >{\centering\arraybackslash}p{0.45\linewidth} >{\centering\arraybackslash}p{0.25\linewidth}}
    \thicktoprule
    \textbf{Model} & \textbf{Relationship Equation} & \textbf{Parameters}\\
    \midrule
    \cite{greenshields1935study} & 
    $\begin{aligned}
      v &= v_{f}\left(1 - \frac{\rho}{\rho_{j}}\right)
    \end{aligned}$ & $v_{f},\rho_{j}$ \\
    \cite{greenberg1959analysis} & $\begin{aligned}
      v &= v_{\text{critical}}\log\left(\frac{\rho_{j}}{\rho}\right)
    \end{aligned}$ & $v_{critical},\rho_{j}$ \\
    \cite{underwood1961speed} & $\begin{aligned}
      v &= v_{f}\exp\left(-\frac{\rho}{\rho_{\text{critical}}}\right)
    \end{aligned}$ & $v_{f},\rho_{critical}$\\
    \cite{newell1961nonlinear} & $\begin{aligned}
      v &= v_{f}\left\{1 -\exp\left[-\frac{\lambda}{v_{f}}(\frac{1}{\rho} - \frac{1}{\rho_{j}})\right]\right\}
    \end{aligned}$ & $v_{f},\rho_{j},\lambda$\\
        \cite{drake1965statistical} & $\begin{aligned}
      v &= v_{f}\exp\left(-(\frac{\rho}{\rho_{{critical}}})^{2}\right)
    \end{aligned}$ & $v_{f},\rho_{critical}$ \\
        \cite{pipes1966car} & $\begin{aligned}
      v &= v_{f}\left(1 - \frac{\rho}{\rho_{j}}\right)^{n}
    \end{aligned}$ & $v_{f},\rho_{j},n$ \\
        \cite{drew1968traffic} & $\begin{aligned}
      v &= v_{f}\left[1 - \left(\frac{\rho}{\rho_{j}}\right)^{m_{1}}\right]^{m_{2}}
    \end{aligned}$ & $v_{f},\rho_{j}, m_{1}, m_{2}$\\
        \cite{papageorgiou1989macroscopic} & $\begin{aligned}
      v &= v_{f}\exp\left[- \frac{1}{\alpha} \left(\frac{\rho}{\rho_{j}}\right)^{\alpha}\right]
    \end{aligned}$ & $v_{f},\rho_{j}, \alpha$\\
        \cite{kerner1994structure} & $\begin{aligned}
      v &= v_{f}\left[\frac{1}{1 + \exp \left(\frac{\frac{\rho}{\rho_{citical}} - 0.25}{0.06}\right)- 372 \times 10 ^{-8}}\right]
    \end{aligned}$ & $v_{f},\rho_{critical}$ \\
       \cite{del1995functional-part-1} & $\begin{aligned}
      v &= v_{f}\left\{1 - \exp[\frac{v_{j}}{v_{f}}(1 - \frac{\rho_{j}}{\rho})]\right\}
    \end{aligned}$  & $v_{f},\rho_{j}, v_{j}$ \\
       \cite{jayakrishnan1995dynamic} & $\begin{aligned}
      v &= v_{min} + (v_{f} - v_{min})(1 - \frac{\rho}{\rho_{j}})
    \end{aligned}$   & $v_{f},v_{min},\rho_{j}$ \\
       \cite{ardekani2008modified} & $\begin{aligned}
      v &= v_{critical}\log(\frac{\rho_{j} + \rho_{min}}{\rho + \rho_{min}})
    \end{aligned}$ & $v_{critical},\rho_{j}, \rho_{min}$ \\
       \cite{macnicholas2011simple} & $\begin{aligned}
      v &= v_{f}\left(\frac{\rho_{j}^{n} - \rho^{n}}{\rho_{j}^{n} + m\rho^{n}} \right)
    \end{aligned}$ & $v_{f},\rho_{j}, n, m$ \\
       \cite{wang2011logistic} & $\begin{aligned}
      v &=  v_{critical} + \frac{v_{f} - v_{critical}}{\left\{1 + \exp[\frac{\rho - \rho_{critical}}{\theta_{1}}] \right\}^{\theta_{2}}}
    \end{aligned}$ & $v_{f},v_{critical},\rho_{critical},\theta_{1},\theta_{2}$ \\
       \cite{cheng2021s} & $\begin{aligned}
      v &=  \frac{v_{f}}{[1 + (\frac{\rho}{\rho_{critical}})^{m}]^{\frac{2}{m}}},1\leq m \leq 8.53
    \end{aligned}$ & $v_{f},\rho_{critical},m$ \\
    \thickbottomrule
  \end{tabular}
\end{table}
\subsubsection{Multi-regime models}
\par Multi-regime models incorporate a series of piecewise functions. Common multi-regime models (\cite{edie1961car}, \cite{drake1965statistical}, \cite{may1990traffic} and \cite{sun2005development}) include two-regime or three-regime models.Recent studies have adopted non-parametric regression frameworks. For instance, \cite{sun2014data} adopted a B-spline non-parametric regression structure. An incomplete list of multi-regime models is shown in Table \ref{Table:2}. The basic idea of multi-regime models is to divide density into two or three groups and then use different functions to fit corresponding empirical data located in different density ranges. Drawing on the findings of \cite{ni2015traffic}, it is evident that the multi-regime model exhibited limited efficacy in accurately fitting the empirical data. Moreover, these models often present issues of discontinuity and lack of differentiability.
\par There has been significant progress in improving the accuracy of both single and multi-regime models. With various attempts and exploration of both parametric and non-parametric regression structures, recent deterministic models have already shown a good fit for empirical data in all density ranges. Therefore, accuracy is no longer the greatest weakness of these models. The deterministic approach has an inherent deficiency in reflecting the density-speed relationship under different states rather than a steady state.
\begin{table}[hthp]
  \centering
  \caption{Some multi-regime traffic flow model\textsuperscript{*}}
  \label{Table:2}
  \begin{tabular}{>{\centering\arraybackslash}p{0.25\linewidth} >{\centering\arraybackslash}p{0.25\linewidth} >
  {\centering\arraybackslash}p{0.25\linewidth} >{\centering\arraybackslash}p{0.25\linewidth}}
    \thicktoprule
    \textbf{Model} & \textbf{Free flow} & \textbf{Transitional} & \textbf{Congested}\\
    \midrule
    Edie model & 
    $\begin{aligned}
      v = 54.9e^{-\frac{\rho}{163.9}}, \rho \leq 50
    \end{aligned}$ & - & $\begin{aligned}
      v = 26.8\ln{\frac{162.5}{\rho}}, \rho > 50
    \end{aligned}$  \\
    Two-regime model & 
    $\begin{aligned}
      v = 60.9 -0.515\rho, \rho \leq 65
    \end{aligned}$ & - & $\begin{aligned}
      v = 40 - 0.265\rho, \rho > 65
    \end{aligned}$  \\
        Modified Greenberg model & 
    $\begin{aligned}
      v = 48, \rho \leq 35
    \end{aligned}$ & - & $\begin{aligned}
      v = 32\ln{\frac{145.5}{\rho}}, \rho > 35
    \end{aligned}$  \\
        Three-regime model & 
    $\begin{aligned}
      v = 50 - 0.098\rho, \rho \leq 40
    \end{aligned}$ & $\begin{aligned}
      v = 81.4 - 0.913\rho, 40 < \rho \leq 65
    \end{aligned}$ & $\begin{aligned}
      v = 40 -0.265\rho, \rho > 65
    \end{aligned}$  \\
    \thickbottomrule
  \end{tabular}
  \footnotesize\textsuperscript{*} The speed of vehicles is measured in miles per hour, while traffic density is quantified as vehicle per mile.
\end{table}
\subsection{Stochastic fundamental diagram models}
\par 
In the 1970s, \cite{soyster1973stochastic} proposed an SFD model that utilized a probabilistic density distribution to describe specific flow rates based on a birth-and-death stochastic process. This model, however, relied on highly restrictive assumptions, including the notion that vehicle arrivals followed a Poisson process. Building upon this, \cite{kerner1998experimental} expanded the framework by offering speed ranges for various densities through three-phase traffic flow models. Further advancements came from \cite{alperovich2008stochastic}, who developed a multi-lane stochastic traffic flow model grounded in microscopic dynamics. Nevertheless, this model did not specify an explicit SFD functional form, omitting a crucial noise or error term. \cite{wang2013stochastic} later introduced a macroscopic stochastic method to depict the speed-flow equilibrium relationship, employing Karhunen-Loeve expansion for estimating variance and mean, yet it lacked explicit stochastic speed-density relationships.
\par Subsequent research efforts, including those by \cite{macnicholas2011simple}, \cite{fan2013data}, \cite{jabari2014probabilistic}, \cite{qu2017stochastic}, and \cite{wang2021model}, succeeded in generating distributions for speed or flow as functions of density, achieving relations based on percentiles. However, the models proposed by the first three studies were limited to specific traffic flow scenarios, which restricted their general applicability in broader traffic flow stochasticity contexts. Specifically, \cite{macnicholas2011simple} introduced a probabilistic graphical approach to the fundamental traffic diagram based on the triangular deterministic model. \cite{fan2013data} incorporated a novel variable, the empty-road velocity, to capture randomness within the ARZ model. \cite{jabari2014probabilistic} crafted a probabilistic stationary speed-density relation using Newell's simplified car-following model.
\par Further, \cite{qu2017stochastic} proposed a percentile-based SFD calibration method, delving into the continuity, differentiability, and convexity of their model. \cite{wang2021model} extended this percentile-based SFD, proposing a comprehensive model capable of coordinating a family of speed–density curves. This model transformed a convex but nondifferentiable model into a linear programming model that was efficiently solvable with advanced solvers. Exploring other aspects, \cite{laval2014distribution} utilized a stochastic triangular fundamental diagram to examine kinematic wave models, while \cite{siqueira2016effect} analyzed the SFD model from a mesoscopic perspective, focusing on vehicle speed transitions. Building on the triangular fundamental diagram concept, \cite{zhang2018empirical} established an SFD model that primarily captured the distributions of shockwave speeds.
\par More recent studies, such as \cite{ni2018modeling}, investigated the SFD through a stochastic process inspired by the Maxwell–Boltzmann distribution, aiming to encapsulate the scattering features in the fundamental diagram. \cite{zhou2020modeling} applied a Gaussian mixture model to depict a triangular fundamental diagram, introducing stochasticity primarily through stochastic headways in a mixed human-driven and automated environment. \cite{ahmed2021fundamental} developed a multimodal SFD model for heterogeneous and undisciplined traffic flows, leveraging unmanned aerial vehicle data from Pakistan. Addressing the uncertainty of speed heterogeneity and rainfall intensity, \cite{bai2021calibration} proposed an SFD model to uncover both observable and unobservable heterogeneities in traffic flows. Most recently, \cite{cheng2024analytical} tackled the limitations of deterministic models in highway traffic flow theory by introducing two stochastic fundamental diagram (SFD) models that employ lognormal and skew-normal distributions. These models aim to accurately represent traffic state variations through stochasticity in free-flow speeds and critical density speeds. Validated with real-world loop detector data, these innovative models effectively replicate empirical data scatterings, demonstrating significant potential in analyzing freeway capacities and identifying sources of stochasticity. However, a stochastic model holds both a stochastic form to reflect the overall traffic pattern and a deterministic form to represent the equilibrium state is needed.
\par Contrary to the majority of macroscopic traffic models, which often use first or second-order approaches and are typically formulated through partial differential equations (PDEs) — much like numerous models in physics, the deterministic models explored in this study do not adhere to a PDE-based framework. Moreover, these models are functions with physics and non-physics parameters, introducing a significant challenge: the imperative for high-quality data to calibrate these empirical models accurately. In the absence of such data, the effectiveness of these models is at risk. Although it is possible to incorporate empirical models into the loss function, doing so will fall into a modeling dilemma and will not necessarily enhance the model's robustness. A robust model can provide accurate estimations and predictions even when the input data is noisy. However, to obtain a robust model, clean data are needed to establish the entire loss function in our case. Once clean data has been identified and processed, using noisy data as input is pointless. The successful enhancement of model robustness through this integration is contingent upon the availability of pristine data. This research proposes the empirical models into the prior as the mean function in sparse Gaussian process regression (SGPR), with the aim of improving the model's precision rather than its robustness. Most deterministic fundamental models can be directly utilized as this mean function. Further into the text, an examination is conducted on whether this inclusion effectively elevates the performance of the model and offsets the limitations of sparse approximation. The primary objective is to facilitate non-parametric models in extracting more valuable insights from the empirical prior — derived from the entire dataset, even when they are trained on only a subset of this data due to significant computational demands. This strategy seeks to enhance the accuracy of non-parametric regression models by making full use of the empirical data, which cannot be done if only rely on non-parametric regression. The primary goal of utilizing empirical models as the mean function is to enhance the precision of non-parametric models when fitting the empirical data. While this approach shares similarities with methods employed in physics-informed machine learning, it is important to note that this research does not classify as physics-informed machine learning. 
\section{Methodology} \label{3}
\par In this section, we will discuss how to generate the non-parametric stochastic fundamental diagram model. The basic methodology consists of three main parts: calibration of the empirical deterministic model, non-parametric stochastic modeling, and integration of empirical prior.
\subsection{Empirical deterministic model}
\par In our overall modeling framework, the empirical deterministic model plays a very important role in bridging the traditional deterministic models and the statistics model. As the prior of the statistics model, we need to ensure that all empirical deterministic models must be well-calibrated. The least square method (LSM) has been widely used to calibrate regression-based models in transportation research areas, such as \cite{sun2003use} and \cite{meng2013bus}. When calibrating the single-regime fundamental diagram model, the LSM may encounter many drawbacks. The real-world data used to calibrate the single-regime fundamental diagram model is unbalanced distributed, and most density-speed pairs are located in the low-traffic density area. As a result, the fitted model based on LSM will be well fitted in high-density areas but skewed in low-density areas and may not be able to reflect the accurate real pattern of the density-speed relations. Therefore, following \cite{qu2015fundamental}, a weighted least square (WLSM) method is adopted to calibrate all the single-regime and multi-regime fundamental diagram models in this paper. Considering a density-speed pairs data set $\mathcal{D}$, $\mathcal{D} = \left\{(\rho_{1},v_{1}),(\rho_{2},v_{2}),...,(\rho_{m},v_{m}) \right\}$ and a function $v = f(\rho,\theta)$, the LSM is aiming to minimize the sum of squared errors as shown in Eq \ref{eq:5}, due to the continuously differentiable assumption of this unconstrained optimization problem, its first-order optimality condition can be expressed as Eq \ref{eq:6}. Unlike LSM, WLSM minimizes Eq \ref{eq:7}, whose first-order optimality condition can now be written as Eq \ref{eq:8}.
\begin{equation}
  \min \mathcal{L}(\theta) = \sum_{i=1}^{n} (v_i - f(\rho_i, \theta))^2   \label{eq:5}
\end{equation}
\begin{equation}
  \frac{\partial \mathcal{L}(\theta)}{\partial \theta_{j}} = -2\sum_{i = 1}^{n} [v_{i} - f(\rho_{i},\theta)] \frac{\partial f(\rho_{i},\theta)}{\partial \theta_{j}} = 0, \quad j = 1, \dots, n  \label{eq:6}
\end{equation}
\begin{equation}
  \min \mathcal{L}(\theta) = \sum_{i=1}^{n} \varpi_{i} (y_{i} - f(x_{i}, \theta))^2   \label{eq:7}
\end{equation}
\begin{equation}
  \frac{\partial \mathcal{L}(\theta)}{\partial \theta_{j}} = -2\sum_{i=1}^{n} \varpi_{i} [v_{i} - f(\rho_{i}, \theta)] \frac{\partial f(\rho_{i}, \theta)}{\partial \theta_{j}} = 0, \quad j = 1, \ldots, n  \label{eq:8}
\end{equation}
\par In our paper, all the single-regime fundamental diagram models will be calibrated through WLSM proposed by \cite{qu2015fundamental}, and the calibration\footnote{The calibration code is shown in
\href{URL}{https://github.com/YuanzhengLei/Weighted-the-least-square-method-for-single-regime-fundamental-diagram-models}} will be done by an optimizer in \cite{virtanen2020scipy}.
\par Several classic and recent single-regime traffic flow models have been listed in Table \ref{Table:1}. The calibrated parameters of each model are presented in Table \ref{Table:3}. However, the jam density of MacNicholas's Model \cite{macnicholas2011simple}) has lost its physical meaning. Therefore, only a portion of the calibrated models are shown in Figures \ref{figure:5} to \ref{figure:16}.
\begin{table}[hthp]
  \centering
  \caption{Parameters of single-regime traffic flow model after calibration}
  \label{Table:3}
  \begin{tabular}{>{\centering\arraybackslash}p{0.35\linewidth} >{\centering\arraybackslash}p{0.65\linewidth}}
    \thicktoprule
    \textbf{Model} &  \textbf{Parameters\textsuperscript{*}}\\
    \midrule
    \cite{greenshields1935study} & $v_{f} = 52.12\ \rho_{j} = 76.68$ \\
    \cite{greenberg1959analysis} & $v_{critical} = 22.06\ \rho_{j} = 92.49$ \\
    \cite{newell1961nonlinear} & $v_{f} = 69.69\ \rho_{j} = 25.000\ \lambda = 1209.02$\\
        \cite{drake1965statistical} & $v_{f} = 80.50\ \rho_{critical} = 50.01$ \\
        \cite{underwood1961speed} & $v_{f} = 80.51\ \rho_{critical} = 92.49$ \\
        \cite{papageorgiou1989macroscopic} & $v_{f} = 79.49\ \rho_{j} = 24.83\ \alpha = 1.02$\\
        \cite{kerner1994structure} & $v_{f} = 60.17\ \rho_{critical} = 106.27$ \\
       \cite{del1995functional-part-1}   & $v_{f} = 69.69\ \rho_{j} = 108.41\ v_{j} = 11.15$ \\
       \cite{jayakrishnan1995dynamic}  & $v_{f} = 52.1198\ v_{min} = 35.0052\ \rho_{j} = 25.1779$ \\
       \cite{ardekani2008modified} & $v_{critical} = 40.41\ \rho_{j} = 56.84\ \rho_{min} = 0.01$ \\
       \cite{macnicholas2011simple} & \sout{$v_{f} = 70.17\ \rho_{j} = 2410.54\ n = 2.00\ m = 13730.07$} \\
       \cite{pipes1966car} & $v_{f} = 76.05\ \rho_{j} = 51.00\ n = 1.22$ \\
       \cite{wang2011logistic} & $v_{f} = 65.23\ v_{critical} = 6.02\ \rho_{critical} = 9.73\ \theta_{1} = 1.53\ \theta_{2} = 0.10$ \\
       \cite{cheng2021s} & $v_{f} = 68.70\ \rho_{critical} = 20.02\ m = 2.21$ \\
    \thickbottomrule
  \end{tabular}
  \footnotesize\textsuperscript{*} The speed of vehicles is measured in miles per hour (mph), while traffic density is quantified as vehicle per mile.
\end{table}
\subsection{Non-parametric stochastic modeling}
\par In the section, we will introduce the non-parametric method that we used to generate our stochastic model. Sparse Gaussian process regression (SGPR), as a powerful non-parametric regression tool, has been selected to archive stochastic modeling.
\par Before formally introducing the SGPR, let us quickly review the Gaussian process regression. Following \cite{williams2006gaussian}, the Gaussian process (GP) is employed to depict a distribution over functions, allowing GP to serve as a powerful analysis method from a function-space perspective. Formally:
\begin{definition}[Gaussian process]
A Gaussian process is a collection of random variables, any finite number of which have a Gaussian distribution.
\end{definition}
\par For any Gaussian process, it can be specified by its mean function $m(\mathbf{x})$ and covariance function $k(\mathbf{x},\mathbf{x'})$. This function is commonly called the kernel of the Gaussian process (\cite{jakel2007tutorial}), and a prevalent choice of a kernel function is the radial basis function kernel, which is defined as $k(\mathbf{x},\mathbf{x'}) = \sigma^{2}e^{-\frac{\|\mathbf{x} - \mathbf{x'}\|^{2}}{2\lambda^{2}}}$, $\sigma$ and $\lambda$ are hyper-parameters, which will decide the shape of the kernel function. The mean function and the covariance function of a real process $f(\mathbf{x})$ can be defined as:
\begin{equation}
    m(\mathbf{x}) = \mathbb{E}[f(\mathbf{x})] \label{eq:9}
\end{equation}
\begin{equation}
    k(\mathbf{x},\mathbf{x'}) = \mathbb{E}[(f(\mathbf{x}) - m(\mathbf{x}))(f(\mathbf{x'}) - m(\mathbf{x'}))] \label{eq:10}
\end{equation}
and the Gaussian process can be expressed as:
\begin{equation}
    f(\mathbf{x}) \sim \mathcal{GP}(m(\mathbf{x}),k(\mathbf{x},\mathbf{x'})) \label{eq:11}
\end{equation}
\par Typically, the mean function is assumed to be zero. Consequently, in subsequent sections of this study, for a pure Gaussian process regression (GPR) model, its mean function $m(\mathbf{\rho}) = 0$, for other empirical prior Gaussian process regression (EPGPR) models, because they are based some deterministic fundamental diagram model $f(\rho)$, their mean function $m(\mathbf{\rho}) = f(\rho)$.
\par Based on \cite{wang2023intuitive}, if points $x_{i}$ and $x_{j}$ are considered similar by the covariance function, their corresponding function outputs, $f(x_{i})$ and $f(x_{j})$, are expected to be similar too. The Gaussian process regression is illustrated as: given observed data $(\mathbf{x}_{i},f_{i}|i = 1,2,3,...,n)$ and a mean function $\mathbf{f}$ estimated from the observed data, to predict at new points $\mathbf{x}_{*}$ as $\mathbf{f}_{*}$. The joint distribution of $\mathbf{f}$ and $\mathbf{f}_{*}$ is actually a multivariate Gaussian/normal distribution, which can be expressed as:
\begin{equation}
\begin{bmatrix}
\mathbf{f} \\
\mathbf{f}_{*}
\end{bmatrix} \sim \mathcal{N} \left(\begin{bmatrix}
m(\mathbf{x})\\
m(\mathbf{x}_{*})
\end{bmatrix}, \begin{bmatrix}
\mathbf{K}_{nn} & \mathbf{K}_{n*} \\
\mathbf{K}_{*n} & \mathbf{K}_{**}
\end{bmatrix} \right)\label{eq:12}
\end{equation}
where $\mathbf{K}_{nn} = \mathbf{K}(\mathbf{x},\mathbf{x})$\footnote{If there are $n$ training points, $\mathbf{K}(\mathbf{x},\mathbf{x}) = \begin{bmatrix}
k(x_{1},x_{1}) & k(x_{1},x_{2}) & k(x_{1},x_{3}) & \cdots & k(x_{1},x_{n}) \\
k(x_{2},x_{1}) & k(x_{2},x_{2}) & k(x_{2},x_{3}) & \cdots & k(x_{2},x_{n}) \\
\cdots & \cdots & \cdots & \cdots & \cdots \\
k(x_{n},x_{1}) & k(x_{n},x_{2}) & k(x_{n},x_{3}) & \cdots & k(x_{n},x_{n})
\end{bmatrix}$}, $\mathbf{K}_{n*} = \mathbf{K}(\mathbf{x},\mathbf{x^{*}})$\footnote{If there are $n$ testing points, $\mathbf{K}(\mathbf{x},\mathbf{x^{*}}) = \begin{bmatrix}
k(x_{1},x_{1}^{*}) & k(x_{1},x_{2}^{*}) & k(x_{1},x_{3}^{*}) & \cdots & k(x_{1},x_{n}^{*}) \\
k(x_{2},x_{1}^{*}) & k(x_{2},x_{2}^{*}) & k(x_{2},x_{3}^{*}) & \cdots & k(x_{2},x_{n}^{*}) \\
\cdots & \cdots & \cdots & \cdots &\cdots \\
k(x_{n},x_{1}^{*}) & k(x_{n},x_{2}^{*}) & k(x_{n},x_{3}^{*}) & \cdots & k(x_{n},x_{n}^{*})
\end{bmatrix}$}, $\mathbf{K}_{**} = \mathbf{K}(\mathbf{x^{*}},\mathbf{x^{*}})$\footnote{If there are $n$ training and testing points, $\mathbf{K}(\mathbf{x^{*}},\mathbf{x^{*}}) = \begin{bmatrix}
k(x_{1}^{*},x_{1}^{*}) & k(x_{1}^{*},x_{2}^{*}) & k(x_{1}^{*},x_{3}^{*}) & \cdots & k(x_{1}^{*},x_{n}^{*}) \\
k(x_{2}^{*},x_{1}^{*}) & k(x_{2}^{*},x_{2}^{*}) & k(x_{2}^{*},x_{3}^{*}) & \cdots & k(x_{2}^{*},x_{n}^{*}) \\
\cdots & \cdots & \cdots & \cdots &\cdots \\
k(x_{n}^{*},x_{1}^{*}) & k(x_{n}^{*},x_{2}^{*}) & k(x_{n}^{*},x_{3}^{*}) & \cdots & k(x_{n}^{*},x_{n}^{*})
\end{bmatrix}$}, Equation \ref{eq:13} gives the joint probability distribution of $\mathbf{f},\mathbf{f}_{*}|\mathbf{x},\mathbf{x}_{*}$ over $\mathbf{f}$ and $\mathbf{f}_{*}$. However, in a regression problem, what we really need is the conditional distribution $\mathbf{f}_{*}|\mathbf{f},\mathbf{x},\mathbf{x}_{*}$ over $\mathbf{f}_{*}$, which can be written as\footnote{Details of how to obtain $\mathbf{f}_{*}|\mathbf{f},\mathbf{x},\mathbf{x}_{*}$ can be found in \cite{bishop2006pattern} and \cite{williams2006gaussian}}:
\begin{equation}
    \mathbf{f}_{*}|\mathbf{f},\mathbf{x},\mathbf{x}_{*} \sim \mathcal{N}(\mathbf{K}_{*n}\mathbf{K}_{nn}^{-1}\mathbf{f},\mathbf{K}_{**} - \mathbf{K}_{*n}\mathbf{K}_{nn}^{-1}\mathbf{K}_{*n}) \label{eq:13}
\end{equation}
\par However, in a more realistic modeling scenarios, $\mathbf{f}$ itself is intractable. The tractable one is the noisy version expressed as $\mathbf{y} = f(\mathbf{x}) + \epsilon$. Assuming $\epsilon$ is an additive independent and identically distributed Gaussian noise, $\epsilon \sim \mathcal{N}(0, \sigma_{\epsilon}^{2})$, then the prior on the noisy observations becomes:
 
and Eq \ref{eq:12} can be rewritten as:
\begin{equation}
\begin{bmatrix}
\mathbf{y} \\
\mathbf{f}_{*}
\end{bmatrix} \sim \mathcal{N} \left(\begin{bmatrix}
m(\mathbf{x})\\
m(\mathbf{x}_{*})
\end{bmatrix}, \begin{bmatrix}
\mathbf{K}_{nn} + \sigma_{\epsilon}^{2}\mathbf{I} & \mathbf{K}_{n*} \\
\mathbf{K}_{*n} & \mathbf{K}_{**}
\end{bmatrix} \right)\label{eq:14}
\end{equation}
\par By deriving the conditional distribution, the predictive equations for Gaussian process regression will be expressed as:
\begin{equation}
    \mathbf{f}_{*}|\mathbf{y},\mathbf{x},\mathbf{x}_{*} \sim \mathcal{N}(\overline{\mathbf{f}}_{*},cov(\mathbf{f}_{*})) \label{eq:15}
\end{equation}
where $\overline{\mathbf{f}}_{*}$ and $cov(\mathbf{f}_{*})$ can be further expressed as:
\begin{equation}
    \overline{\mathbf{f}}_{*} = \mathbf{K}_{*n}[\mathbf{K}_{nn} + \sigma_{\epsilon}^{2}\mathbf{I}]^{-1}\mathbf{y} \label{eq:16}
\end{equation}
\begin{equation}
    cov(\mathbf{\mathbf{f}}_{*}) = \mathbf{K}_{**} - \mathbf{K}_{*n}[\mathbf{K}_{nn} + \sigma_{\epsilon}^{2}\mathbf{I}]^{-1}\mathbf{K}_{n*} \label{eq:17}
\end{equation}
\par Similarly, the prediction of the output $\mathbf{y}_{*},\mathbf{y}_{*} = \mathbf{f}_{*} + \epsilon$ is described by (\cite{titsias2009variational}):
\begin{equation}
    \mathbf{y}_{*}|\mathbf{y},\mathbf{x},\mathbf{x}_{*} \sim \mathcal{N}(\mathbf{y}_{*}|\mathbf{K}_{*n}[\mathbf{K}_{nn} + \sigma_{\epsilon}^{2}\mathbf{I}]^{-1}\mathbf{y},\mathbf{K}_{**} - \mathbf{K}_{*n}[\mathbf{K}_{nn} + \sigma_{\epsilon}^{2}\mathbf{I}]^{-1}\mathbf{K}_{n*})) \label{eq:18}
\end{equation}
\par The posterior GP also depends on the values of the hyper-parameters $\Theta = (\sigma, \lambda)$ (Shown in kernel function), which are determined by maximizing the log marginal likelihood as follows:
\begin{equation}
    \Theta_{*} = \argmax_{\Theta}\log{(\mathbf{y}|\mathbf{x},\Theta)} = \argmax_{\Theta}\log{[\mathcal{N}(\mathbf{y}|m(\mathbf{x}),\mathbf{K}_{nn} + \sigma^{2}\mathbf{I}))]}\label{eq:19}
\end{equation}
\par However, although GP is a compelling and popular non-parametric method and has been shown to perform well on various tasks (\cite{walder2008sparse}), the application of GP models in large datasets is still tricky (\cite{titsias2009variational}) due to the time complexity as $\mathcal{O}(n^{3})$ and the storage as $\mathcal{O}(n^{2})$, where $n$ is the size of samples. Many approximate or sparse methods (\cite{williams2000using}, \cite{smola2000sparse}, \cite{csato2002sparse}, \cite{lawrence2002fast}, \cite{seeger2003fast}, \cite{schwaighofer2002transductive}, \cite{snelson2005sparse}, \cite{quinonero2005unifying} and \cite{titsias2009variational}) have been proposed to solve this dilemma. The core idea of those studies is to employ a small set with $m$ inducing variables to form a reduction of the time complexity from $\mathcal{O}(n^{3}) \dashrightarrow \mathcal{O}(m^{2}n)$.
\par In this paper, we adopted the variational framework proposed in \cite{titsias2009variational}. Suppose $m$ inducing variables are selected to conduct the sparse approximation, which can be expressed as $\mathbf{f}_{m} = \{f_{1}^{m},f_{2}^{m},...,f_{m}^{m}\}$ that are latent function values evaluated at input points $\mathbf{x}_{m}, \mathbf{x}_{m} = (x_{1}^{m}, x_{2}^{m}, ..., x_{m}^{m})$. For the posterior GP, it can be written by:
\begin{equation}
    p(\mathbf{f}_{*}|\mathbf{y}) = \int p(\mathbf{f}_{*}|\mathbf{f})p(\mathbf{f}|\mathbf{y})d\mathbf{f} =  \int p(\mathbf{f}_{*}|\mathbf{f}_{m},\mathbf{f})p(\mathbf{f}|\mathbf{f}_{m},\mathbf{y})p(\mathbf{f}_{m}|\mathbf{y})d\mathbf{f}d\mathbf{f}_{m} \label{eq:20}
\end{equation}
where $\mathbf{f}_{*}$ is the predict function values at new inputs $\mathbf{x}_{*}$. Referring to \cite{tiao2020svgp} and \cite{krasser2020gaussian}, then $p(\mathbf{f}_{*}|\mathbf{y})$ can be further written as an approximate posterior (For detailed derivatives, please refer to Lemma \ref{lemma:1}):
\begin{equation}
        q(\mathbf{f}_{*}) = \int p(\mathbf{f}_{*}|\mathbf{f}_{m})\phi(\mathbf{f}_{m})d\mathbf{f}_{m}\label{eq:21}
\end{equation}
where $\phi(\mathbf{f}_{m})$ is an approximation to $p(\mathbf{f}_{m}|\mathbf{y})$, which is intractable and can be further written as:
\begin{equation}
    \phi(\mathbf{f}_{m}) = \mathcal{N}(\mathbf{f}_{m}|\boldsymbol{\mu}_{m},\mathbf{A}_{m}) \label{eq:22}
\end{equation}
The mean and covariance matrix of the Gaussian approximate posterior $q(\mathbf{f}_{*})$ can be further defined as:
\begin{equation}
    q(\mathbf{f}_{*}) = \mathcal{N}(\mathbf{f}_{*}|\boldsymbol{\mu}_{*}^{q},\mathbf{\Sigma}_{*}^{q}) \label{eq:23}
\end{equation}
\begin{equation}
    \boldsymbol{\mu}_{*}^{q} =  \mathbf{K}_{*m}\mathbf{K}_{mm}^{-1}\boldsymbol{\mu}_{m}\label{eq:24}
\end{equation}
\begin{equation}
    \mathbf{\Sigma}_{*}^{q} = \mathbf{K}_{**} - \mathbf{K}_{*m}\mathbf{K}_{mm}^{-1}\mathbf{K}_{m*} + \mathbf{K}_{*m}\mathbf{K}_{mm}^{-1}\mathbf{A}_{m}\mathbf{K}_{mm}^{-1}\mathbf{K}_{m*} \label{eq:25}
\end{equation}
Then, by minimizing the  Kullback-Leibler (KL) divergence (\cite{kullback1951information}) between the approximate $q(\mathbf{f})$ and the exact posterior $p(\mathbf{f}|\mathbf{y})$, which is relevant to minimize the KL divergence between the augmented variational posterior $q(\mathbf{f},\mathbf{f}_{m}|\mathbf{y})$ and the augmented true posterior $p(\mathbf{f},\mathbf{f}_{m}|\mathbf{y})$. Based on \cite{titsias2009variational} and \cite{titsias2009variational-1}, the minimization is equivalently expressed as the maximization of the following variational lower bound of the true log marginal likelihood (Please check the detailed derivative in Lemma \ref{lemma:2}):
\begin{equation}
F_{V}(\mathbf{x}_{m},\phi) = \int p(\mathbf{f}|\mathbf{f}_{m})\phi(\mathbf{f}_{m})\log\frac{p(\mathbf{f}|\mathbf{y})p(\mathbf{f}_{m})}{\phi(\mathbf{f}_{m})}d\mathbf{f}d\mathbf{f}_{m} \label{eq:26}
\end{equation}
And we can finally have the optimal values of $\boldsymbol{\mu}_{m}$ and $\mathbf{A}_{m}$ as ((Please check the detailed derivative in Lemma \ref{lemma:3} and \ref{lemma:4}):
\begin{equation}
  \boldsymbol{\mu}_{m} = \frac{1}{\sigma^{2}}\mathbf{K}_{mm}\mathbf{\Sigma}\mathbf{K}_{mn}\mathbf{y}  \label{eq:27}
\end{equation}
\begin{equation}
  \mathbf{A}_{m} = \mathbf{K}_{mm}\mathbf{\Sigma}\mathbf{K}_{mm}  \label{eq:28}
\end{equation}
where $\mathbf{\Sigma} = (\mathbf{K}_{mm} + \sigma^{-2}\mathbf{K}_{mn}\mathbf{K}_{nm})^{-1}$, and after substituted $\boldsymbol{\mu}_{m}$ and $\mathbf{A}_{m}$ into equations \ref{eq:24} and \ref{eq:25}, we can compute the approximate posterior $ q(\mathbf{f}_{*})$ at new inputs $\mathbf{x}_{*}$.
\par As mentioned in the preceding section, sparse Gaussian Process Regression (SGPR) requires $m$ inducing variables to approximate the exact GP model closely. Thus, in the paper, four sampling methods, including simple random sampling, systematic sampling, cluster sampling, and weighted random sampling, are used to choose inducing variables from the entire dataset to conduct comparison studies, and details of sampling methods can be found in the appendix \ref{8}.
\subsection{Empirical prior integration}
\par In this section, the methodology for integrating an empirical prior into a non-parametric model is explored. As previously delineated in Section \ref{3}, a Gaussian Process (GP) is defined by its mean function \(m(\mathbf{x})\) and covariance function \(k(\mathbf{x}, \mathbf{x'})\). Given a regression model \(\mathbf{y} = f(\mathbf{x}) + \epsilon\), the mean function \(m(\mathbf{x})\) encapsulates prior knowledge regarding the unknown function \(f(\mathbf{x})\). In most cases, the prior is assumed to be zero, and integrating an empirical prior into a non-parametric model essentially replaces the prior with empirical functions about the speed-density relationship. After the change in the assumption on the prior, the variation distribution $\phi(\mathbf{f}_{m})$ shown in equation \ref{eq:22} becomes (Please check the detailed derivative in Lemma \ref{lemma:5}):
\begin{equation}
        \phi(\mathbf{f}_{m}) = \mathcal{N}(\mathbf{f}_{m}|\boldsymbol{\mu}_{*},\boldsymbol{\Sigma}_{*}) \label{eq:29}
\end{equation}
\begin{equation}
        \boldsymbol{\Sigma}_{*} =  \left(\mathbf{K}_{mm}^{-1}+\frac{1}{\sigma^{2}}\mathbf{K}_{mm}^{-1}\mathbf{K}_{mn}\mathbf{K}_{nm}\mathbf{K}_{mm}\right)^{-1}  \label{eq:30}
\end{equation}
\begin{equation}
        \boldsymbol{\mu}_{*} = m(\mathbf{x}_m) + \mathbf{\Sigma}_{*}\left(\frac{1}{\sigma^{2}}\mathbf{K}_{mm}^{-1}\mathbf{K}_{nm}(\mathbf{y} - m(\mathbf{x}))\right) \label{eq:31}
\end{equation}
\par The $q(\mathbf{f}_{*})$ now becomes(Please check the detailed derivative in Lemma \ref{lemma:6}):
\begin{equation}
    q(\mathbf{f}_{*}) = \mathcal{N}(m(\mathbf{x}_{*}) + \mathbf{K}_{*m}\mathbf{K}_{mm}^{-1}(\boldsymbol{\mu}_{*}-m(\mathbf{x}_{m}),\mathbf{K}_{**}-\mathbf{K}_{*m}\mathbf{K}_{mm}^{-1}\mathbf{K}_{m*} + \mathbf{K}_{*m}\mathbf{K}_{mm}^{-1}\boldsymbol{\mu}_{*}\mathbf{K}_{mm}^{-1}\mathbf{K}_{m*}  \label{eq:32}
\end{equation}
\par Predominantly, deterministic models are parametric, devoid of partial derivatives, and describe the equilibrium speed-density relationship. Consequently, these models can be seamlessly employed as the mean function within SGPR. Namely, the mean function \(m(\mathbf{v}) = 0\) will be replaced by \(m(\mathbf{v}) = f(\rho)\), $f(\rho)$ is a deterministic fundamental diagram model. For example, if Greenshield's model (\cite{greenshields1935study}), Greenberg's model (\cite{greenberg1959analysis}), and Newell's model (\cite{newell1961nonlinear}) will be used as the prior, then the corresponding mean function of the corresponding SGPR will be set as:
\begin{equation}
    m_{1}(v) = v_{f}\left(1 - \frac{\rho}{\rho_{j}}\right) \label{eq:33}
\end{equation}
\begin{equation}
    m_{2}(v) = v_{\text{critical}}\log\left(\frac{\rho_{j}}{\rho}\right) \label{eq:34}
\end{equation}
\begin{equation}
    m_{3}(v) = v_{f}\left\{1 -\exp\left[-\frac{\lambda}{v_{f}}(\frac{1}{\rho} - \frac{1}{\rho_{j}})\right]\right\} \label{eq:35}
\end{equation}
\par The diagram in figure \ref{figure: flow chart} summarizes the process of integrating an empirical prior into a non-parametric model. Modules highlighted with green backgrounds utilize the entire dataset, whereas those with blue backgrounds only employ a subset of the dataset. Owing to the computationally intensive nature of Gaussian processes, a subset of the dataset is used for sparse modeling to manage computational demands. In contrast, the calibration of deterministic models exploits the full dataset. By employing calibrated empirical models as the prior, we aim to mitigate the accuracy loss inherent in the variational approximation.
\begin{figure}
    \centering   
    \includegraphics[width=1.0\textwidth]{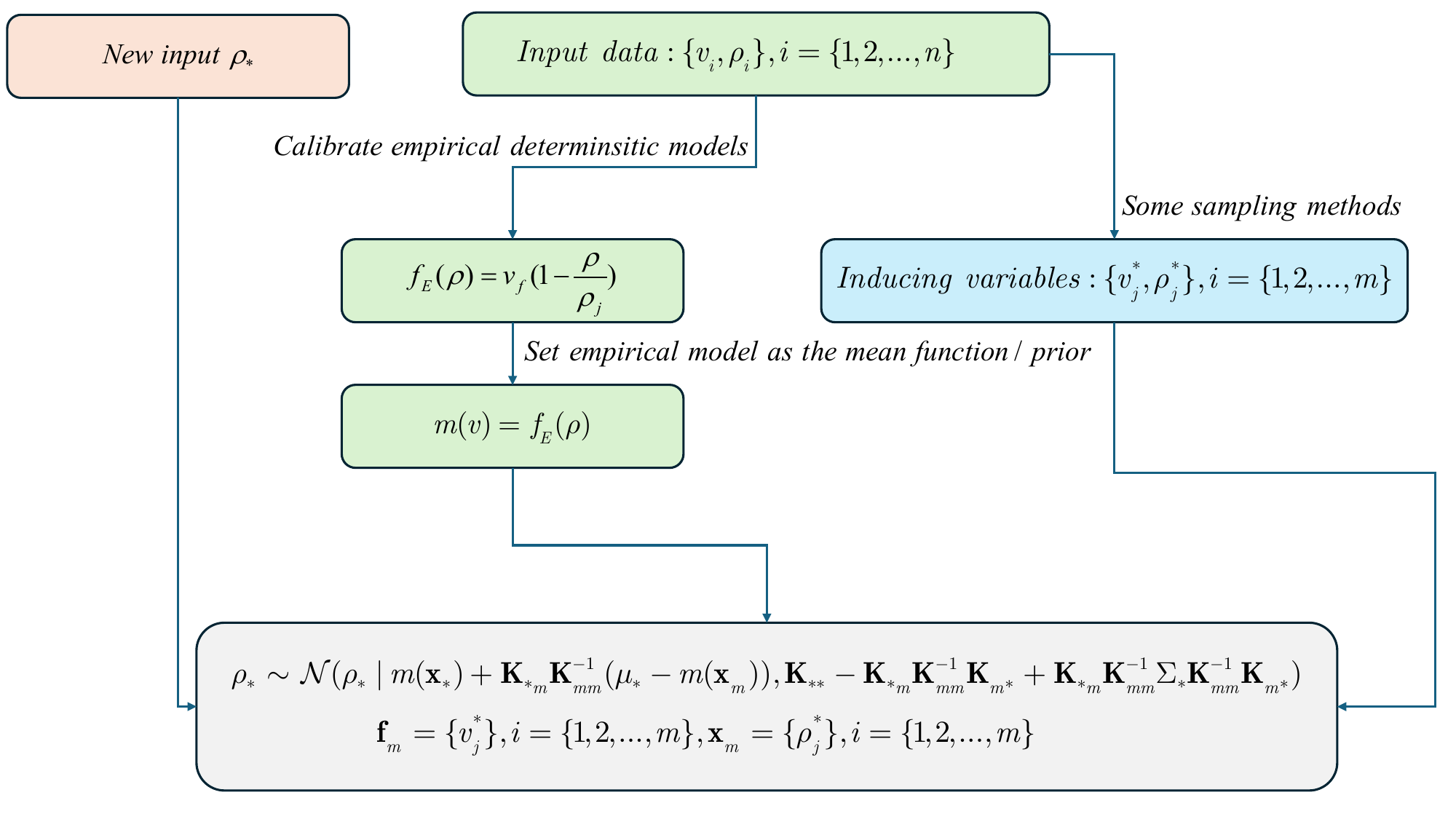}
    \caption{Empirical prior integration}
    \label{figure: flow chart}
\end{figure}
\par Furthermore, kernel functions play a crucial role in SGPR. After evaluating various kernel functions such as the Squared Exponential Kernel, Rational Quadratic Kernel, Mat$\acute{e}$rn kernels, and the Exponential Kernel, it was determined that the Exponential Kernel demonstrated the highest accuracy across all sampling methodologies. Therefore, this kernel function has been selected for implementation, as shown below:
\begin{equation}
    k(\mathbf{x}, \mathbf{x'}) =  \sigma^{2}e^{-\frac{\|\mathbf{x} - \mathbf{x'}\|}{2}} \label{eq:36}
\end{equation}
\begin{figure*}[tp]
    \centering
    \begin{subfigure}{0.32\textwidth}
        \includegraphics[width=\linewidth]{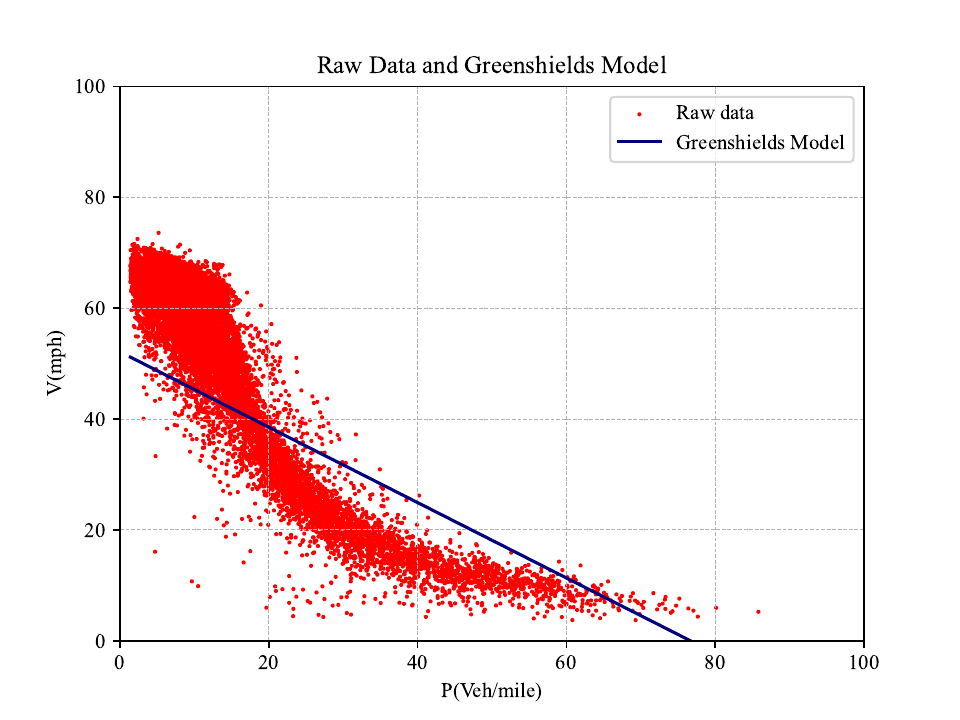}
        \caption{Greenshields model}
        \label{figure:5}
    \end{subfigure}
    \begin{subfigure}{0.32\textwidth}
        \includegraphics[width=\linewidth]{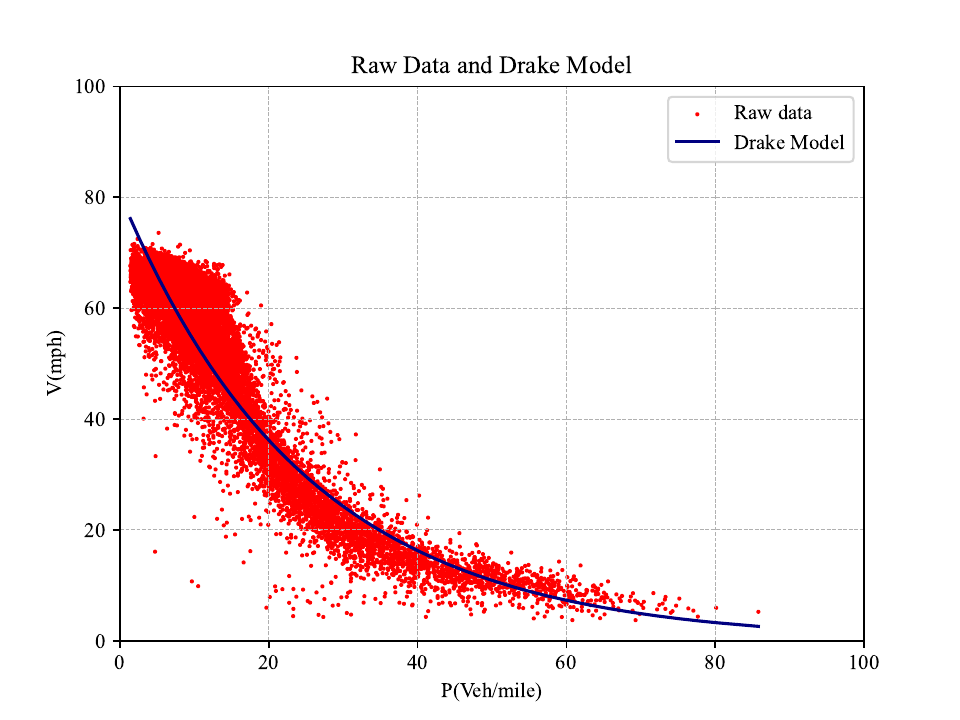}
        \caption{Drake model}
        \label{figure:6}
    \end{subfigure}
    \begin{subfigure}{0.32\textwidth}
        \includegraphics[width=\linewidth]{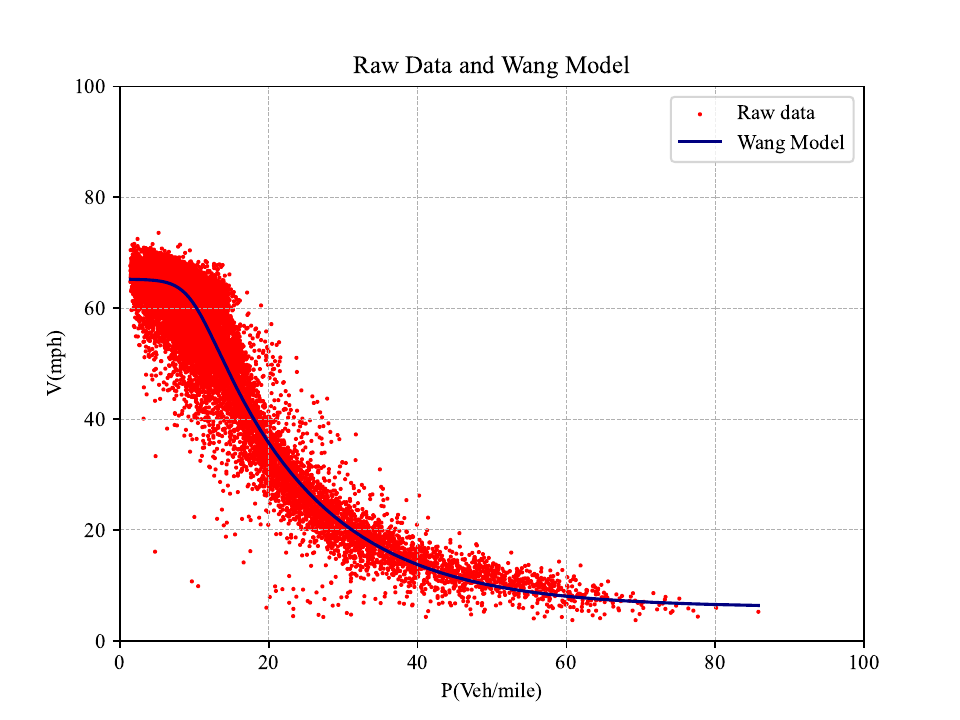}
        \caption{Wang model}
        \label{figure:7}
    \end{subfigure}
    
    \begin{subfigure}{0.32\textwidth}
        \includegraphics[width=\linewidth]{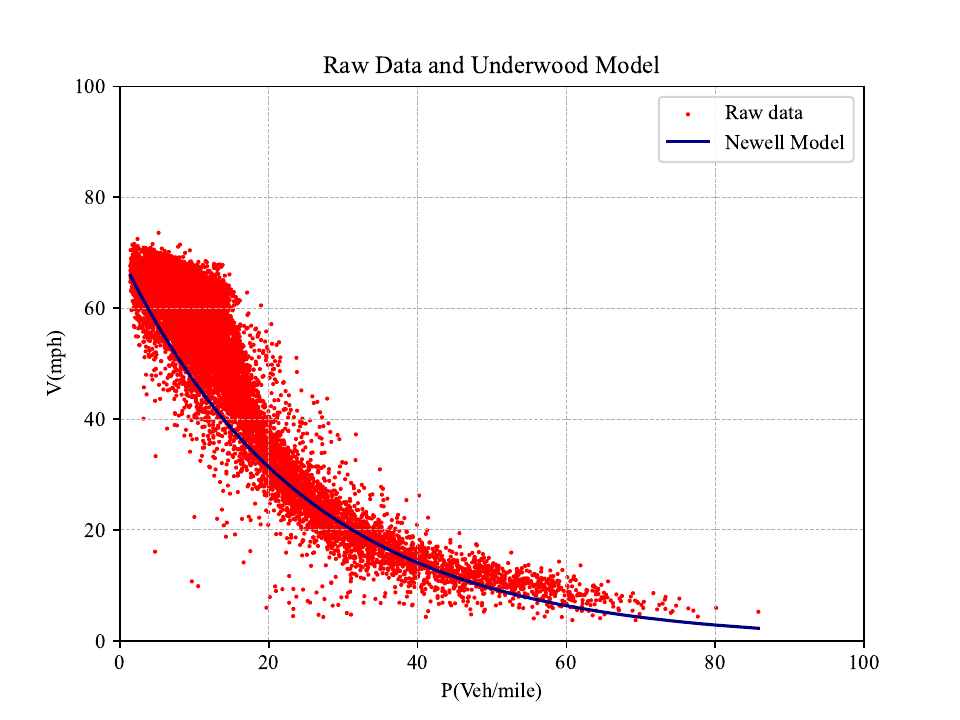}
        \caption{Newell model}
         \label{figure:8}
    \end{subfigure}
    \begin{subfigure}{0.32\textwidth}
        \includegraphics[width=\linewidth]{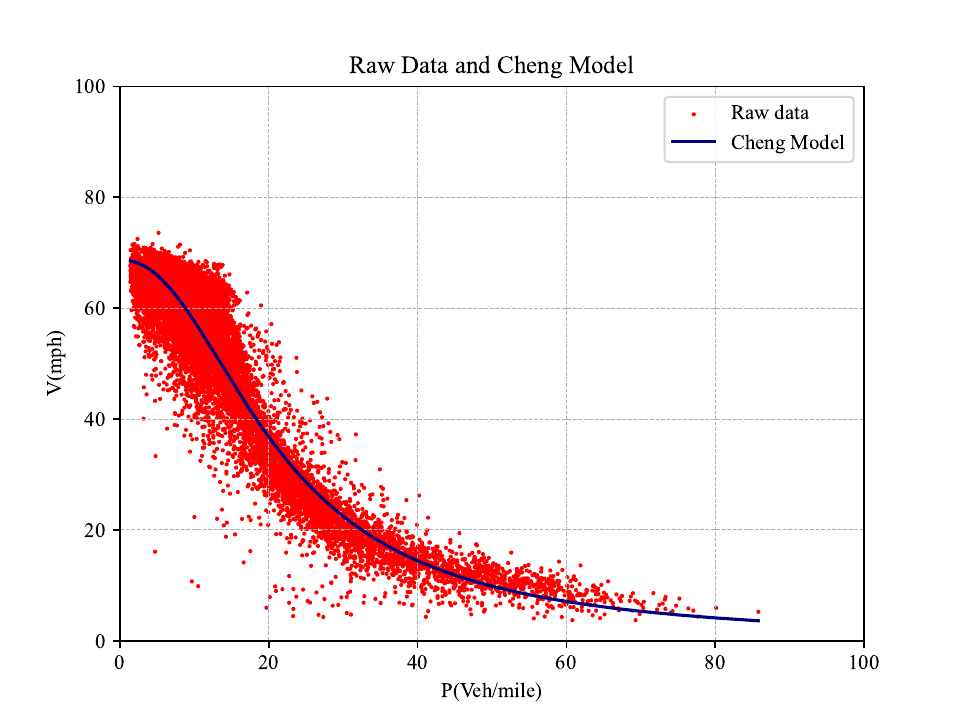}
        \caption{Cheng model}
         \label{figure:9}
    \end{subfigure}
    \begin{subfigure}{0.32\textwidth}
        \includegraphics[width=\linewidth]{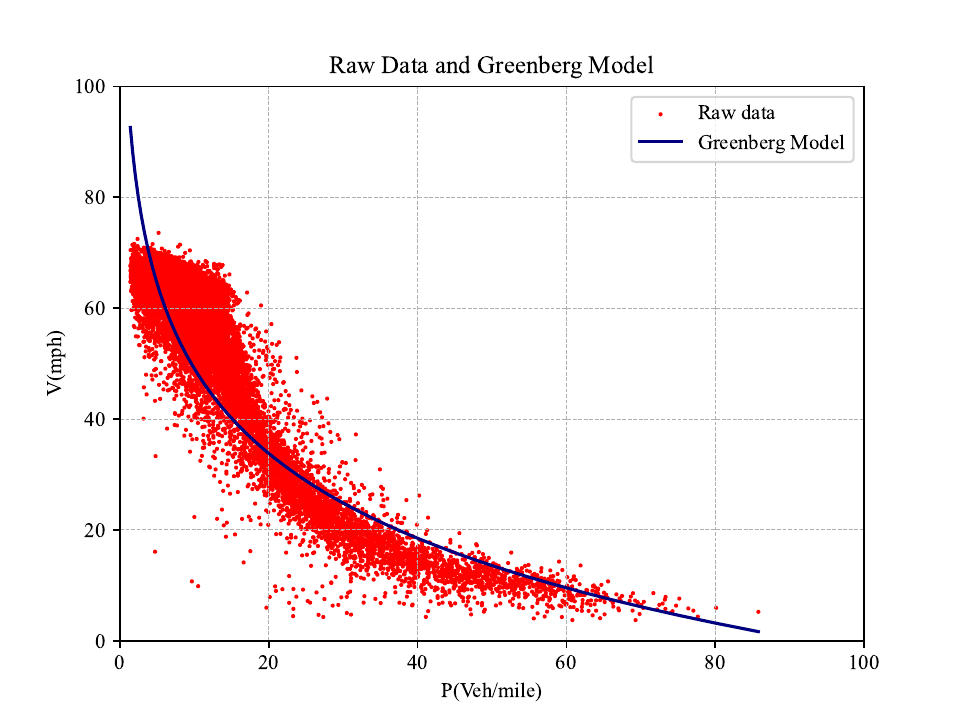}
        \caption{Greenberg model}
        \label{figure:10}
    \end{subfigure}

    \begin{subfigure}{0.32\textwidth}
        \includegraphics[width=\linewidth]{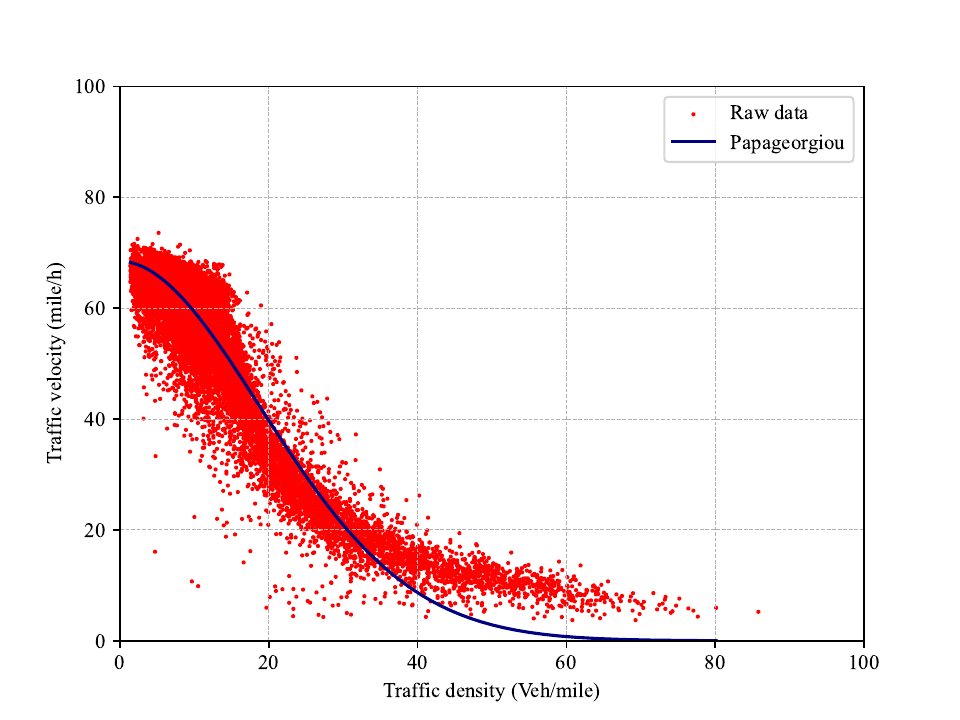}
        \caption{Papageorgiou model}
        \label{figure:11}
    \end{subfigure}
    \begin{subfigure}{0.32\textwidth}
        \includegraphics[width=\linewidth]{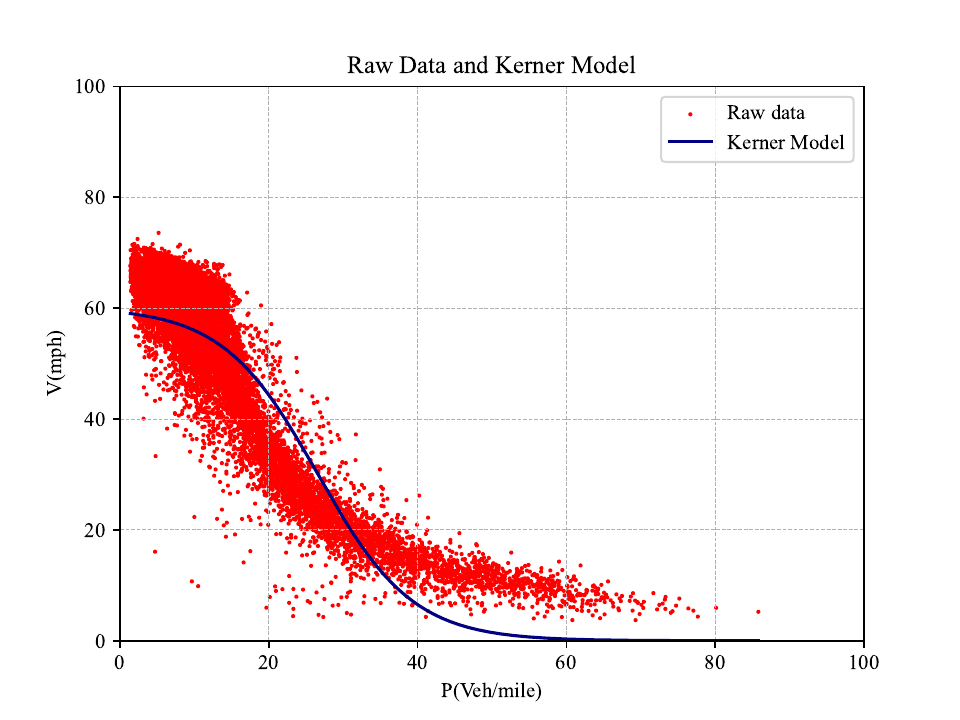}
        \caption{Kerner model}
        \label{figure:12}
    \end{subfigure}
    \begin{subfigure}{0.32\textwidth}
        \includegraphics[width=\linewidth]{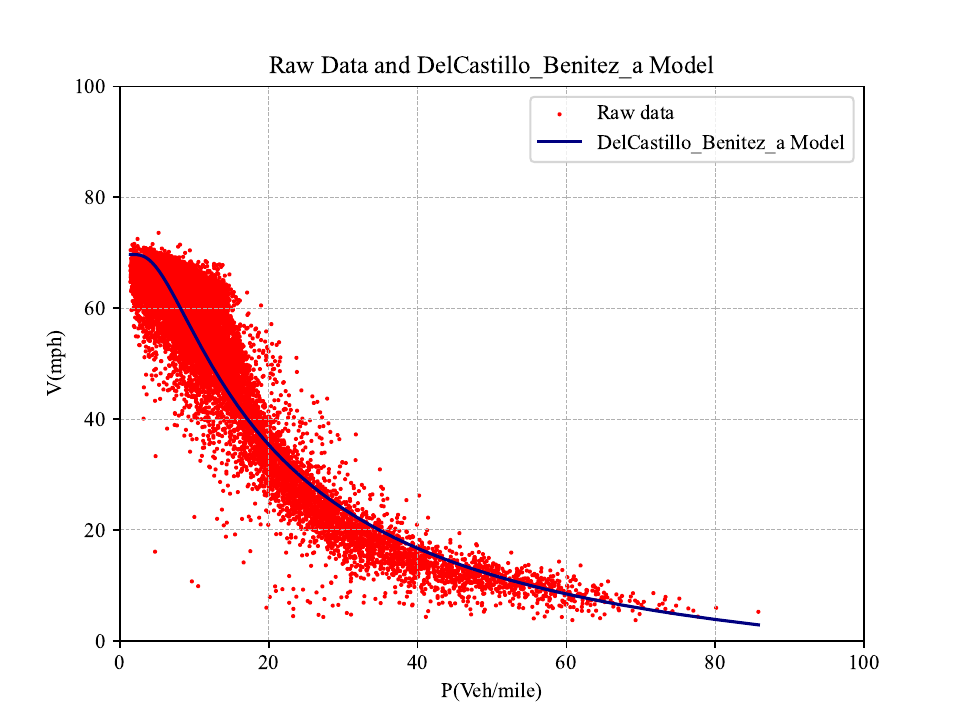}
        \caption{Del Castillo and Benitez model}
        \label{figure:13}
    \end{subfigure}
    
    \begin{subfigure}{0.32\textwidth}
        \includegraphics[width=\linewidth]{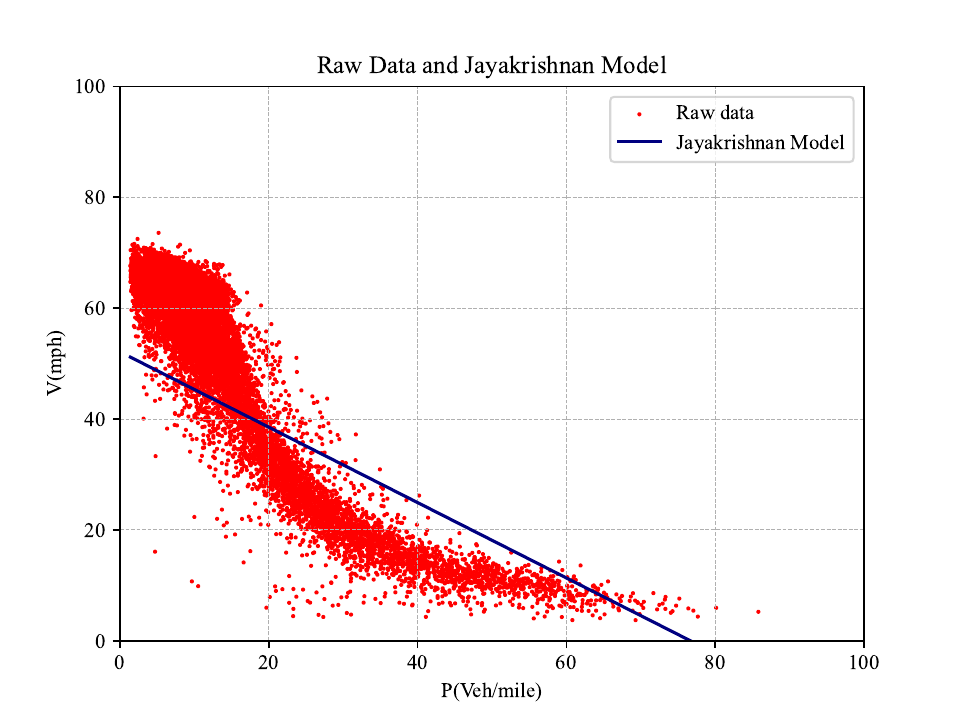}
        \caption{Jayakrishnan model}
        \label{figure:14}
    \end{subfigure}
    \begin{subfigure}{0.32\textwidth}
        \includegraphics[width=\linewidth]{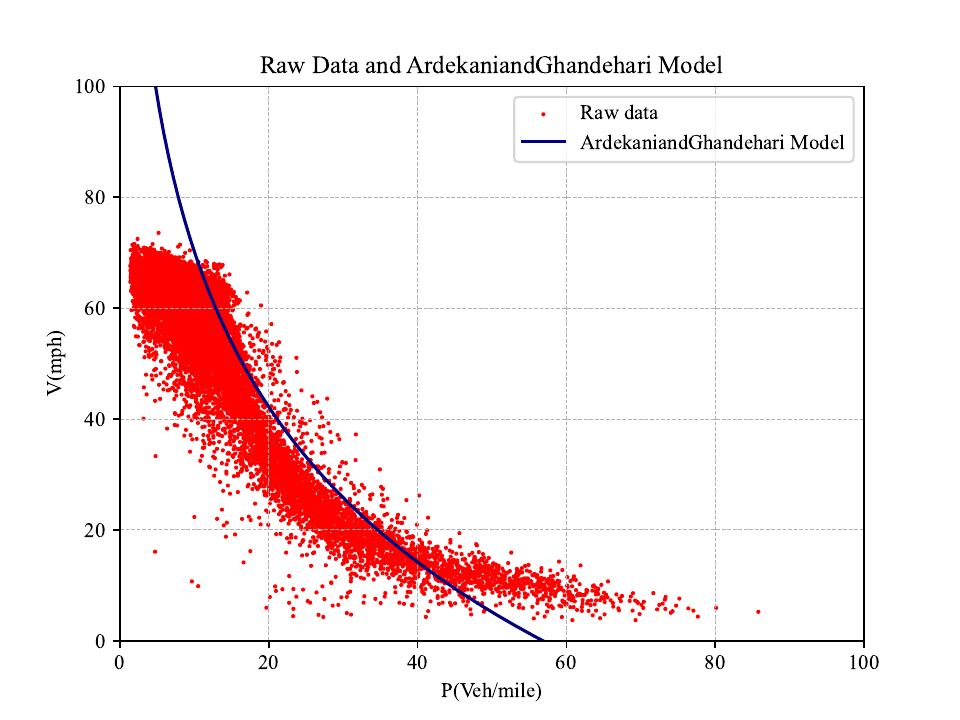}
        \caption{Ardekani-Ghandehar model}
        \label{figure:15}
    \end{subfigure}
    \begin{subfigure}{0.32\textwidth}
        \includegraphics[width=\linewidth]{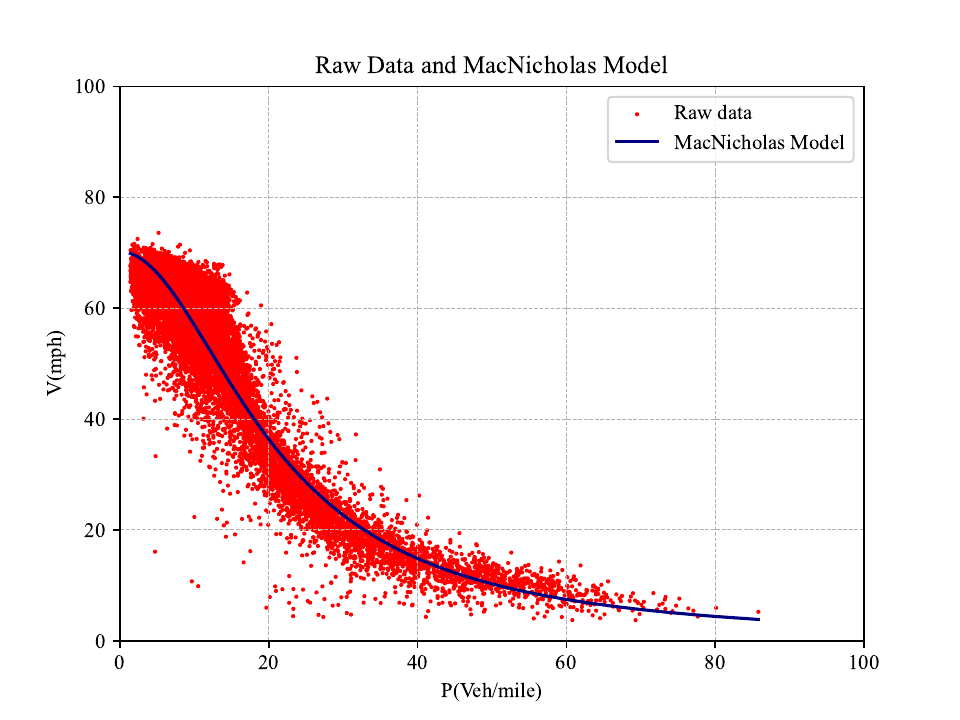}
        \caption{Mac Nicholas model}
        \label{figure:16}
    \end{subfigure}
    \caption{Single-regime traffic flow models after calibration}
    \label{figure:17}
\end{figure*}
\section{Case study} \label{4}
\subsection{Data description}
\par In this case study, data collected from loop detectors across 76 stations on the Georgia State Route 400 serve as the empirical dataset (Denotes as GA400). This data has been extensively utilized within the research community for developing fundamental diagram models, as evident by studies such as \cite{wang2011logistic}, \cite{wang2013stochastic},\cite{qu2015fundamental}, and \cite{qu2017stochastic}. The GA400 dataset comprises 44,787 pairs of density-speed observations, indicating a dataset size of \( n = 44,787 \). The choice of \( m = 288 \), where \( 288 \ll 44,787 \), is justified by the data collection frequency: with a pair recorded every 5 minutes, a single day yields \( \frac{24 \times 60}{5} = 288 \) data pairs. This parameterization comprehensively represents daily traffic patterns within a manageable computational framework. 
\begin{figure}[htbp]
    \centering
    \begin{subfigure}{.28\textwidth}
        \centering
        \includegraphics[width=\linewidth]{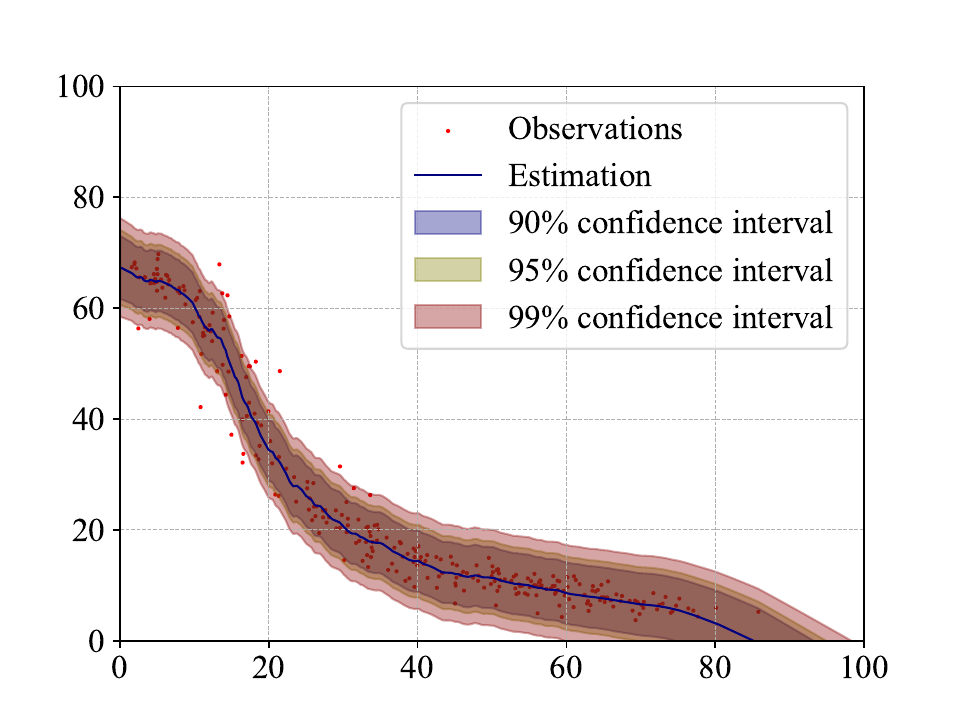}
        \caption{Greensheilds prior}
        \label{figure:18}
    \end{subfigure}
    \begin{subfigure}{.28\textwidth}
        \centering
        \includegraphics[width=1.0\textwidth]{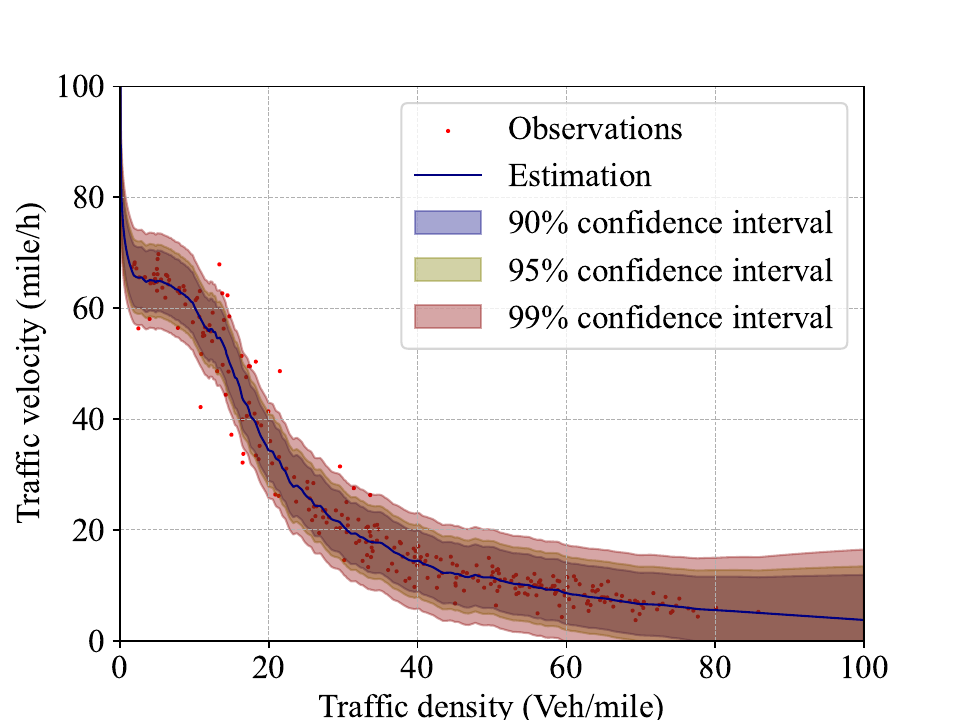}
        \caption{Greenberg prior}
        \label{figure:19}
    \end{subfigure}
    \begin{subfigure}{.28\textwidth}
        \centering
        \includegraphics[width=1.0\textwidth]{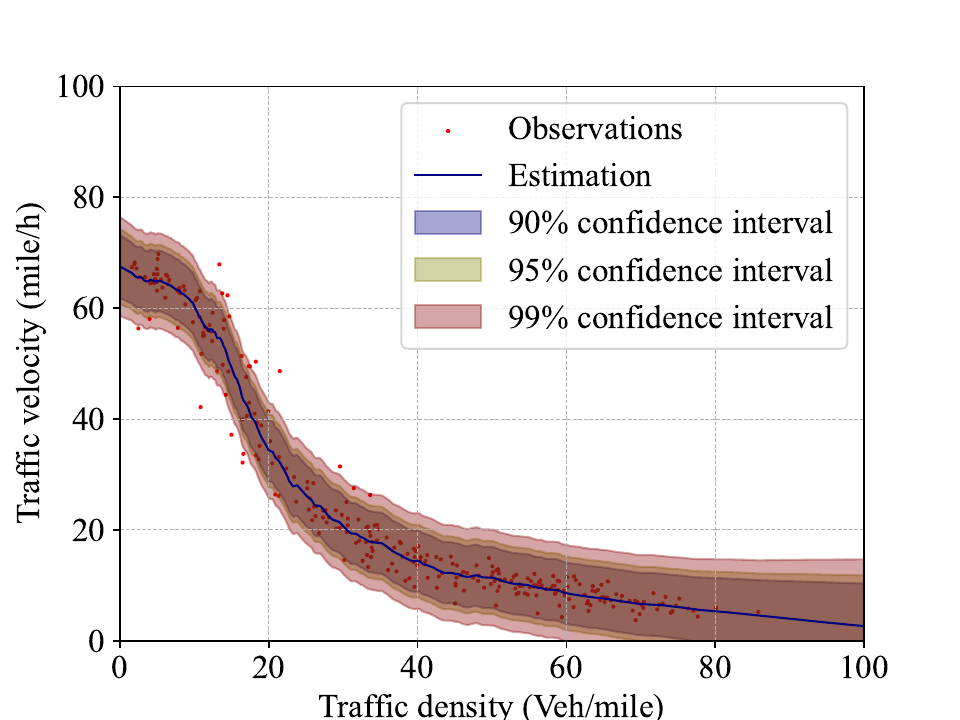}
         \caption{Underwood prior}
        \label{figure:20}
        \label{fig:newell}
    \end{subfigure}

    \begin{subfigure}{.28\textwidth}
        \centering
        \includegraphics[width=1.0\textwidth]{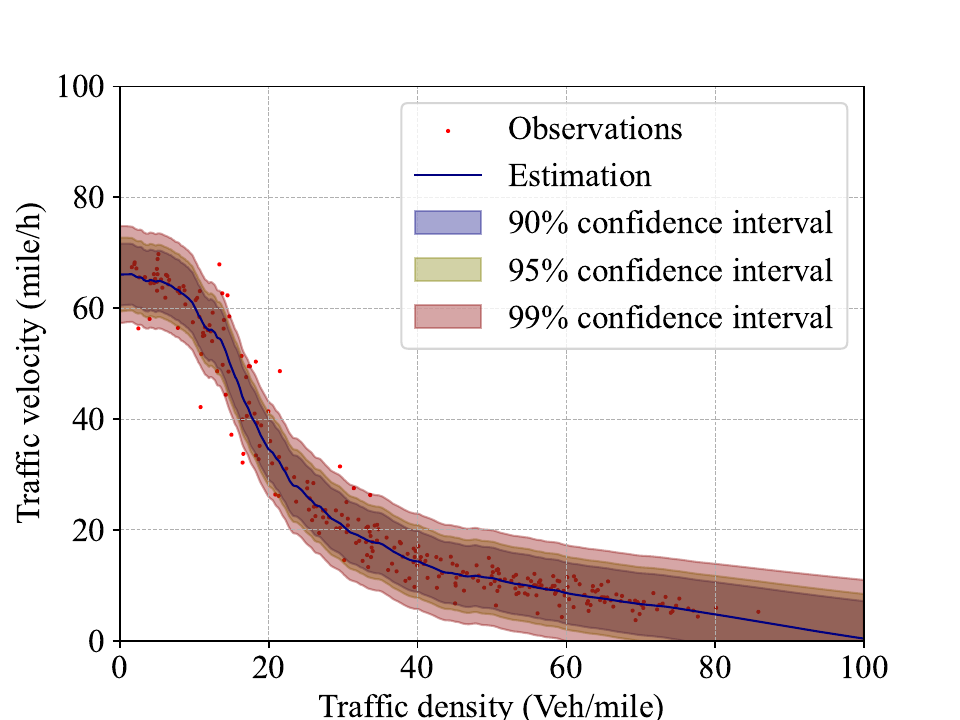}
        \caption{Newell prior}
        \label{figure:21}
    \end{subfigure}
    \begin{subfigure}{.28\textwidth}
        \centering
        \includegraphics[width=1.0\textwidth]{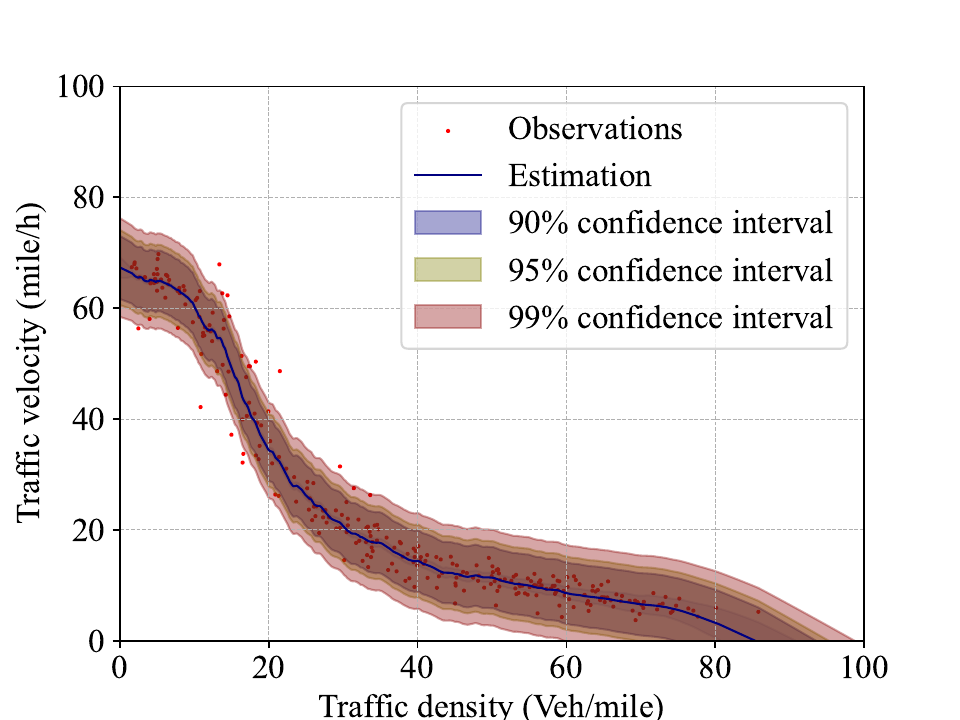}
        \caption{Jayakrishnan prior}
        \label{figure:22}
    \end{subfigure}
    \begin{subfigure}{.28\textwidth}
        \centering
        \includegraphics[width=1.0\textwidth]{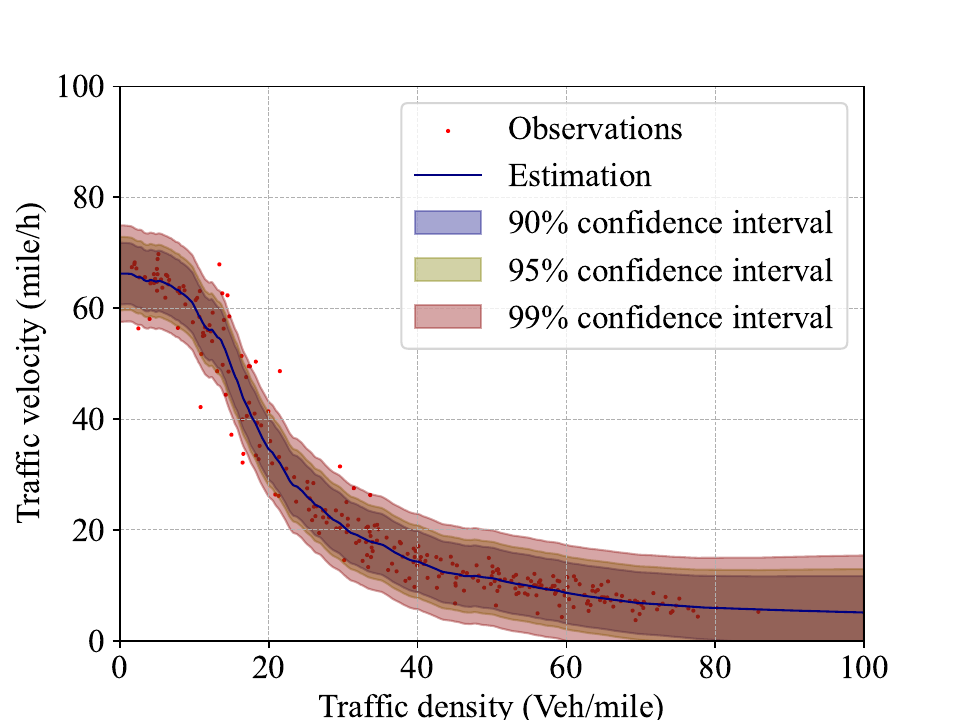}
        \caption{Drew prior}
        \label{figure:23}
    \end{subfigure}

    \begin{subfigure}{.28\textwidth}
        \centering
        \includegraphics[width=1.0\textwidth]{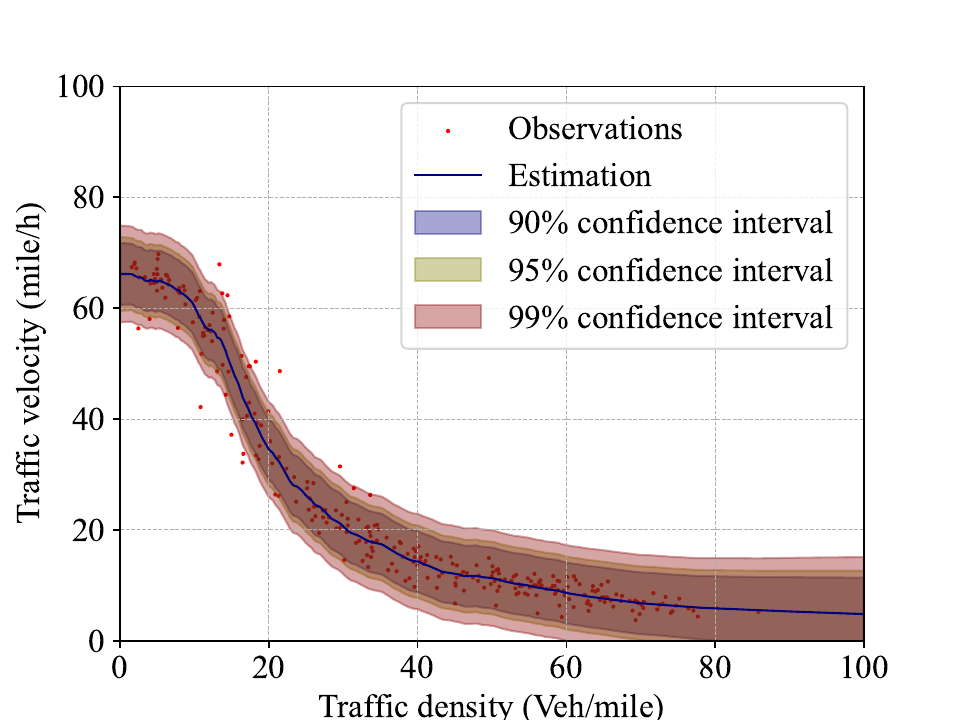}
        \caption{Papageorgiou prior}
        \label{figure:24}
    \end{subfigure}
    \begin{subfigure}{.28\textwidth}
        \centering
        \includegraphics[width=1.0\textwidth]{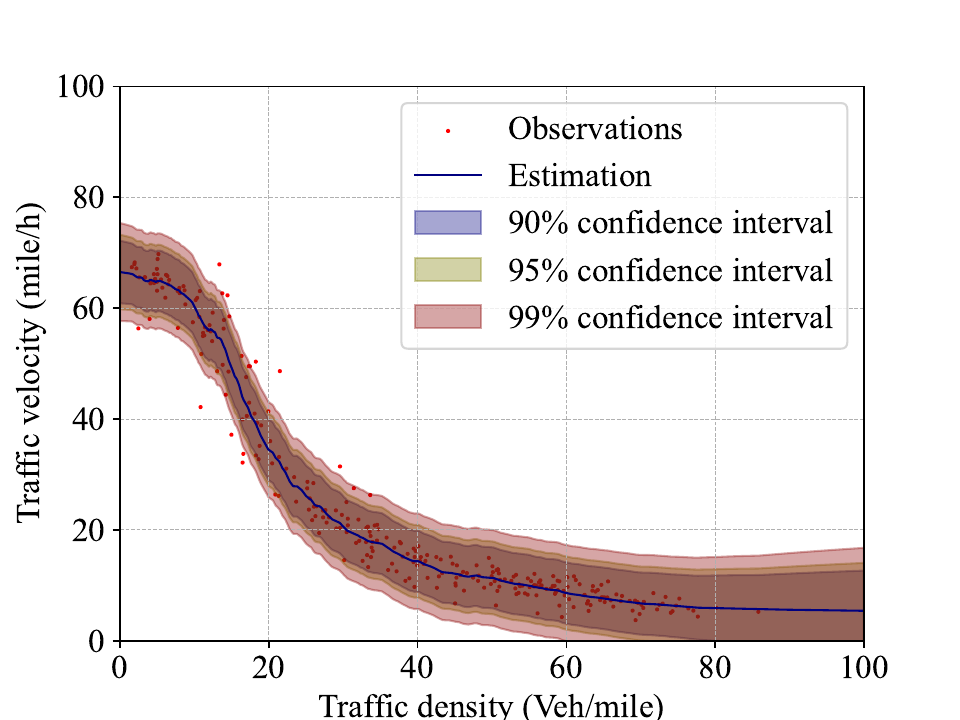}
        \caption{Kerner-Konhauser prior}
        \label{figure:25}
    \end{subfigure}
    \begin{subfigure}{.28\textwidth}
        \centering
        \includegraphics[width=1.0\textwidth]{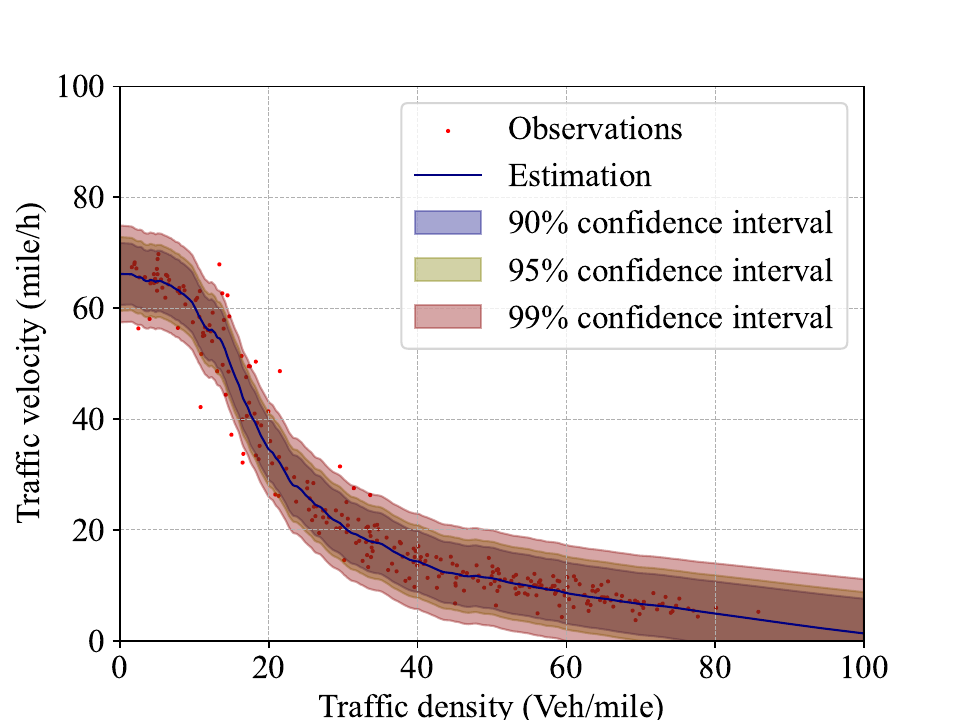}
    \caption{Del-Castillo-Benitez prior}
    \label{figure:26}
    \end{subfigure}
    \begin{subfigure}{.28\textwidth}
        \centering
        \includegraphics[width=1.0\textwidth]{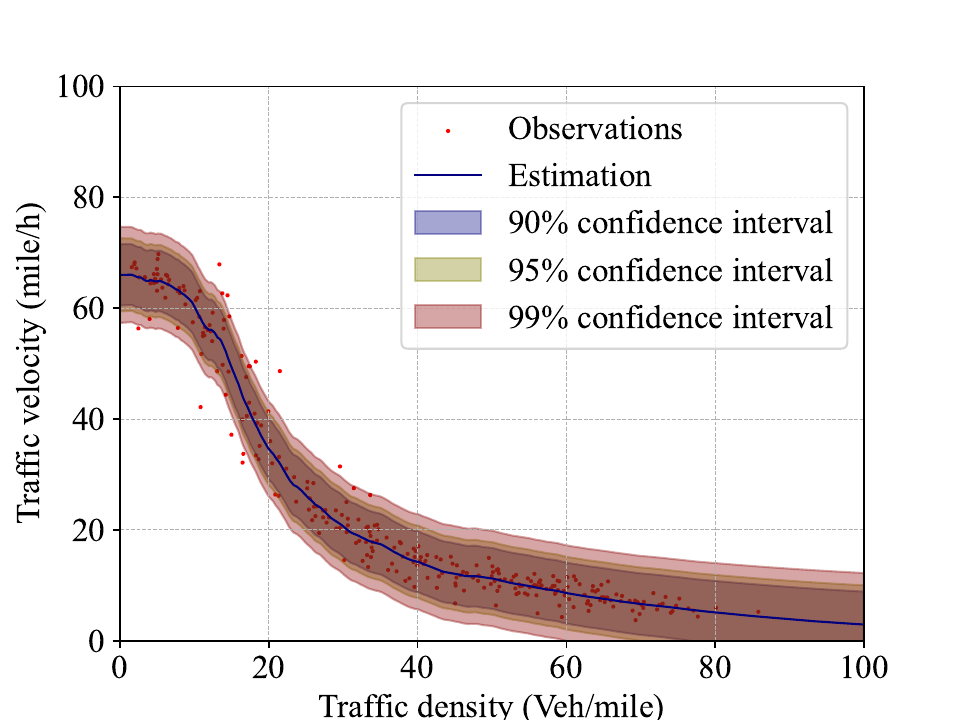}
        \caption{Cheng prior}
        \label{figure:27}
    \end{subfigure}
    \begin{subfigure}{.28\textwidth}
        \centering
        \includegraphics[width=1.0\textwidth]{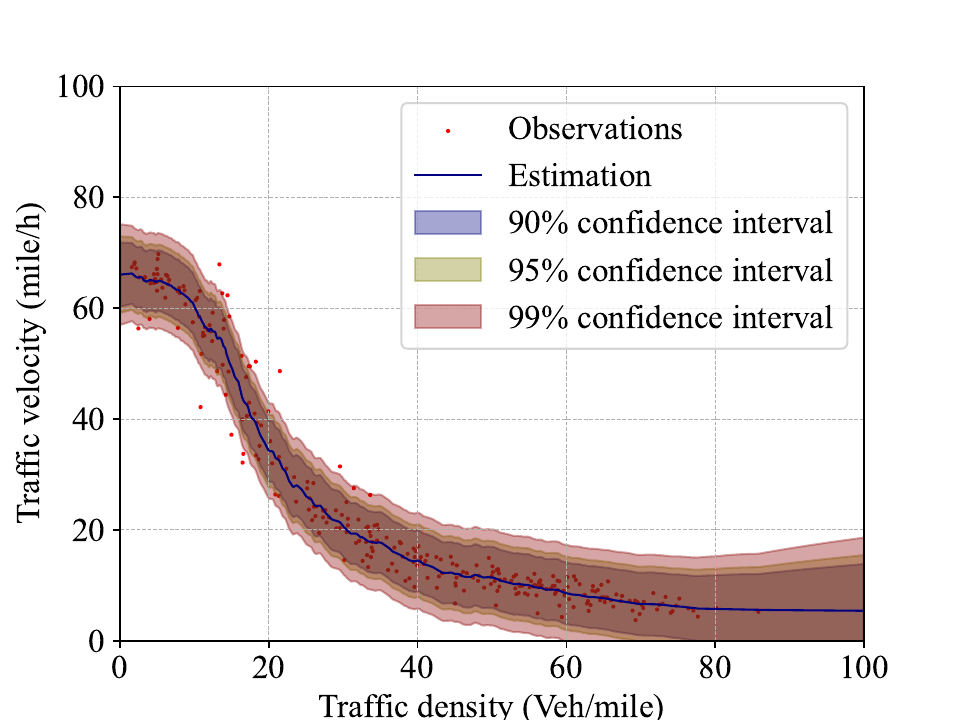}
        \caption{Mac-Nicholas prior}
        \label{figure:28}
    \end{subfigure}
    \begin{subfigure}{.28\textwidth}
        \centering
        \includegraphics[width=1.0\textwidth]{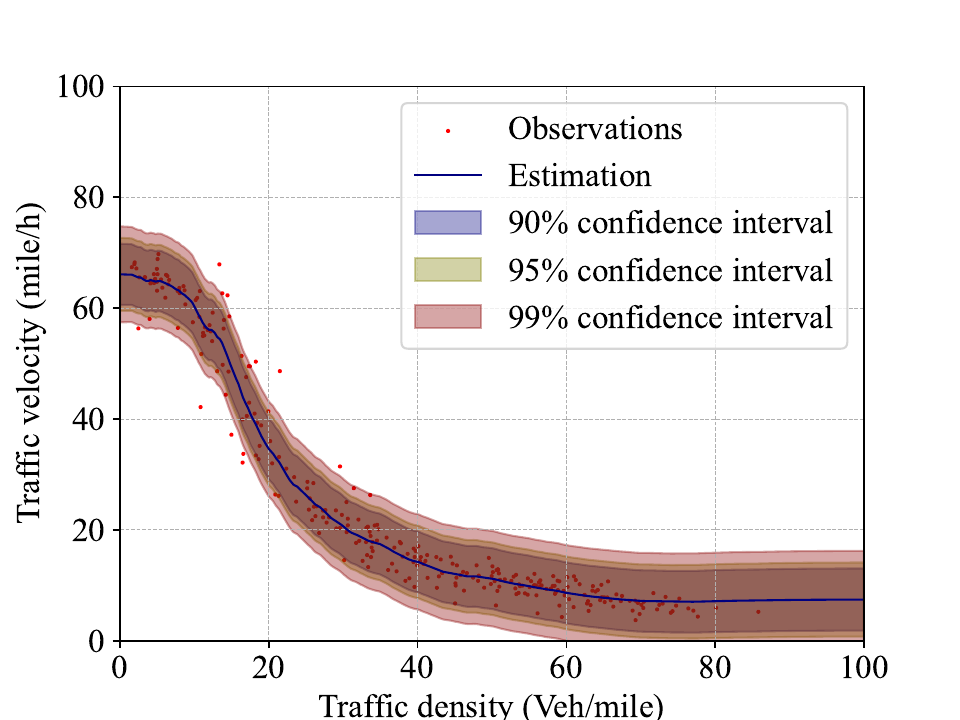}
    \caption{Wang model prior}
    \label{figure:29}
    \end{subfigure}
    \begin{subfigure}{.28\textwidth}
        \centering
        \includegraphics[width=1.0\textwidth]{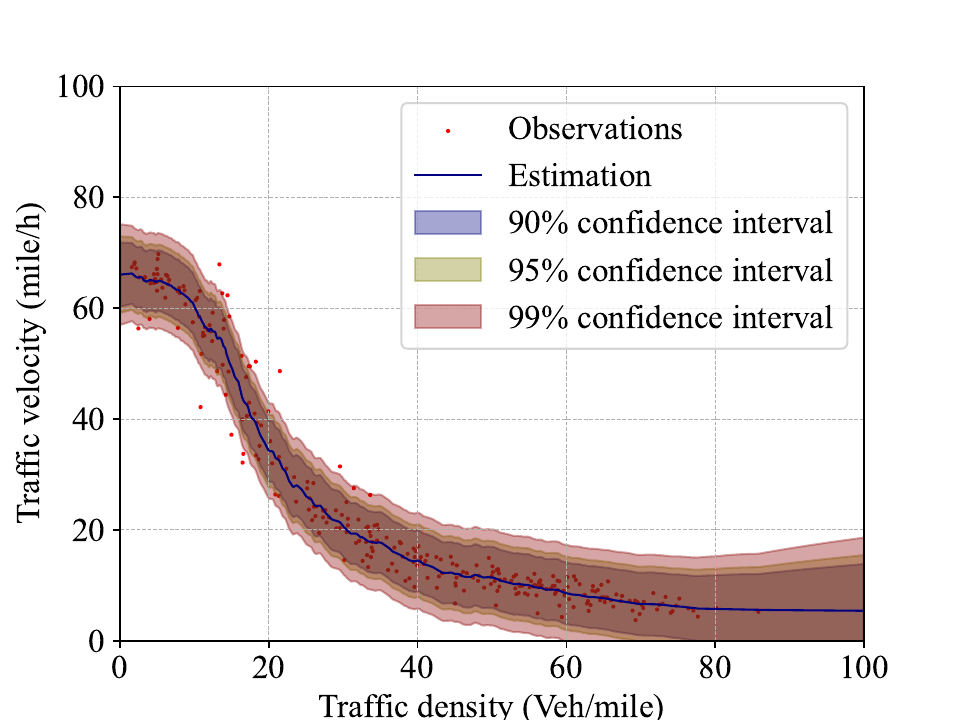}
        \caption{Pure GP regression}
        \label{figure:30}
    \end{subfigure}
    \caption{2D Visualization of speed-density relationships based on different priors}
    \label{figure:31}
\end{figure}

\begin{figure}[htbp]
    \centering
    \begin{subfigure}{.27\textwidth}
        \centering
        \includegraphics[trim=5cm 1cm 4cm 3cm, clip, width=1.0\textwidth]{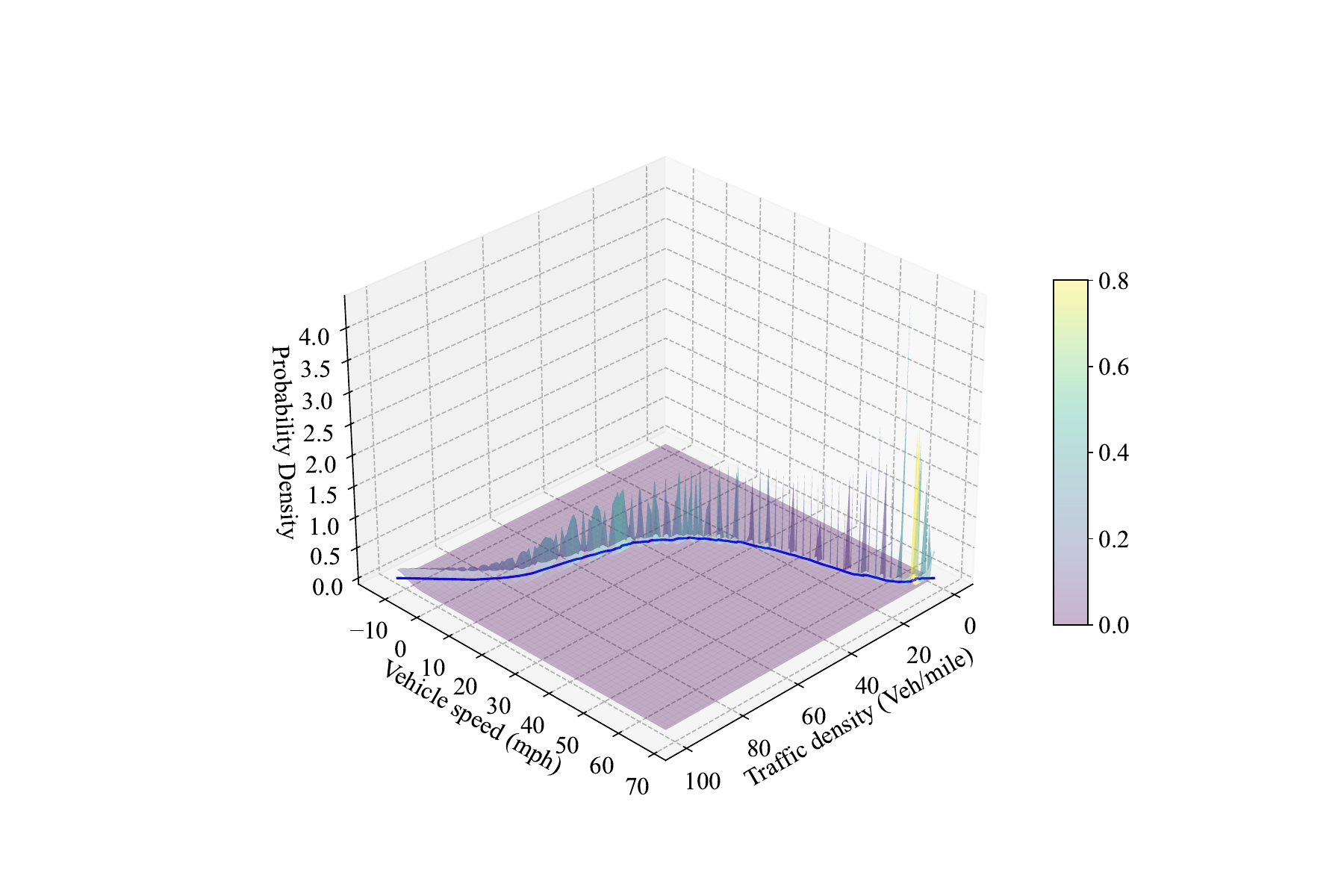}
        \caption{Greensheilds prior}
        \label{figure:32}
    \end{subfigure}
    \begin{subfigure}{.27\textwidth}
        \centering
        \includegraphics[trim=5cm 1cm 4cm 3cm, clip, width=1.0\textwidth]{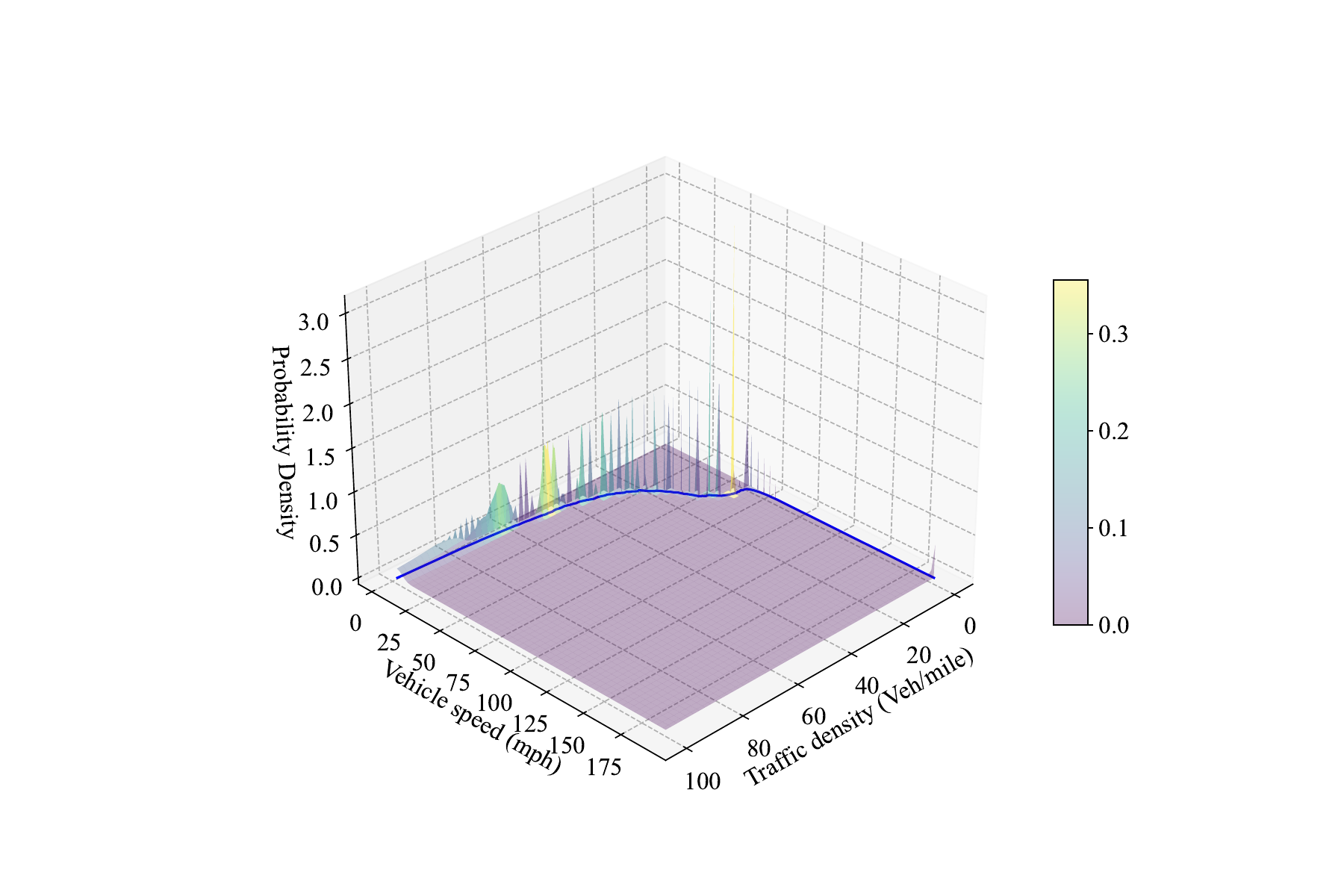}
        \caption{Greenberg prior}
        \label{figure:33}
    \end{subfigure}
    \begin{subfigure}{.27\textwidth}
        \centering
        \includegraphics[trim=5cm 1cm 4cm 3cm, clip, width=1.0\textwidth]{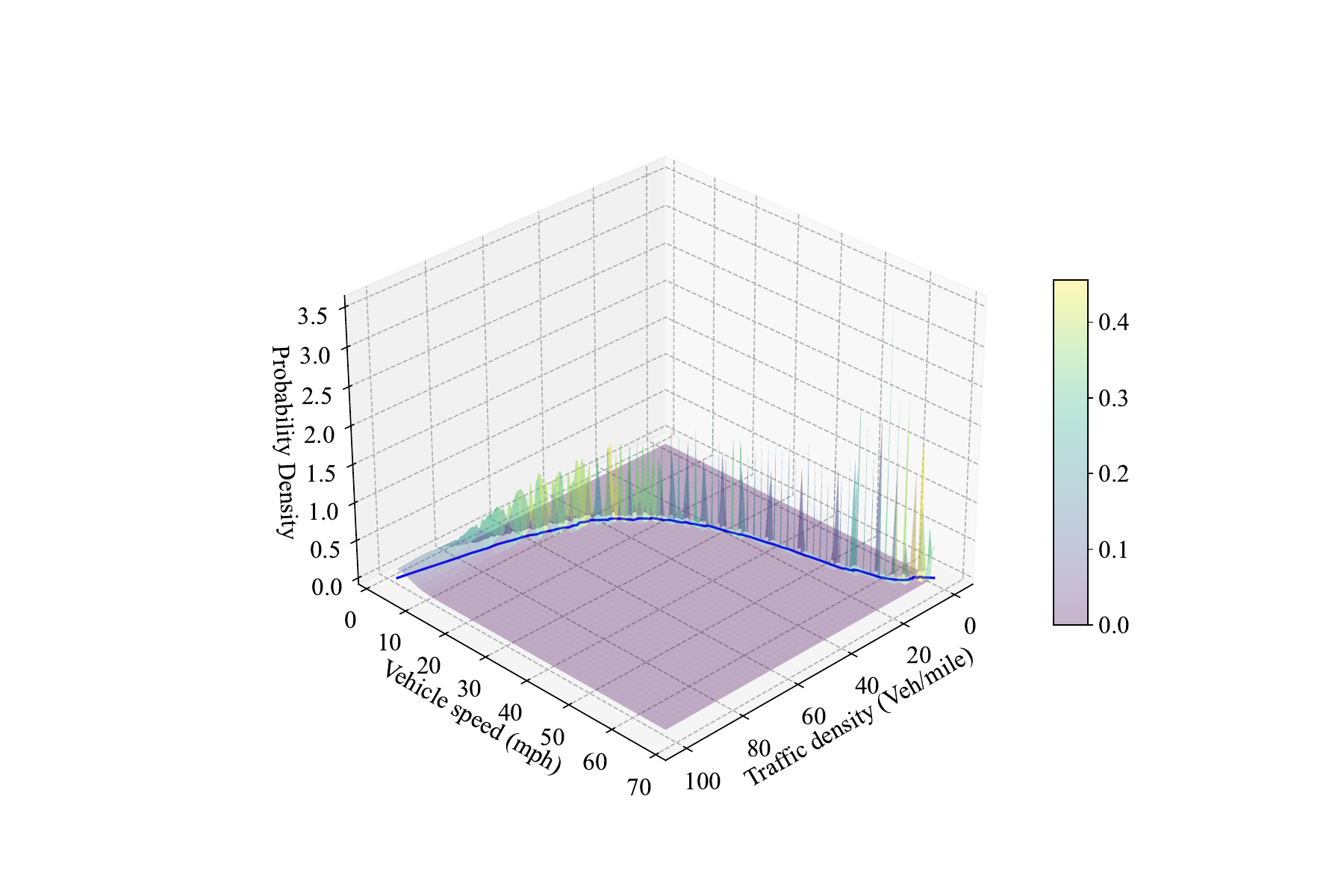}
         \caption{Underwood prior}
        \label{figure:34}
        \label{fig:newell}
    \end{subfigure}

    \begin{subfigure}{.27\textwidth}
        \centering
        \includegraphics[trim=5cm 1cm 4cm 3cm, clip, width=1.0\textwidth]{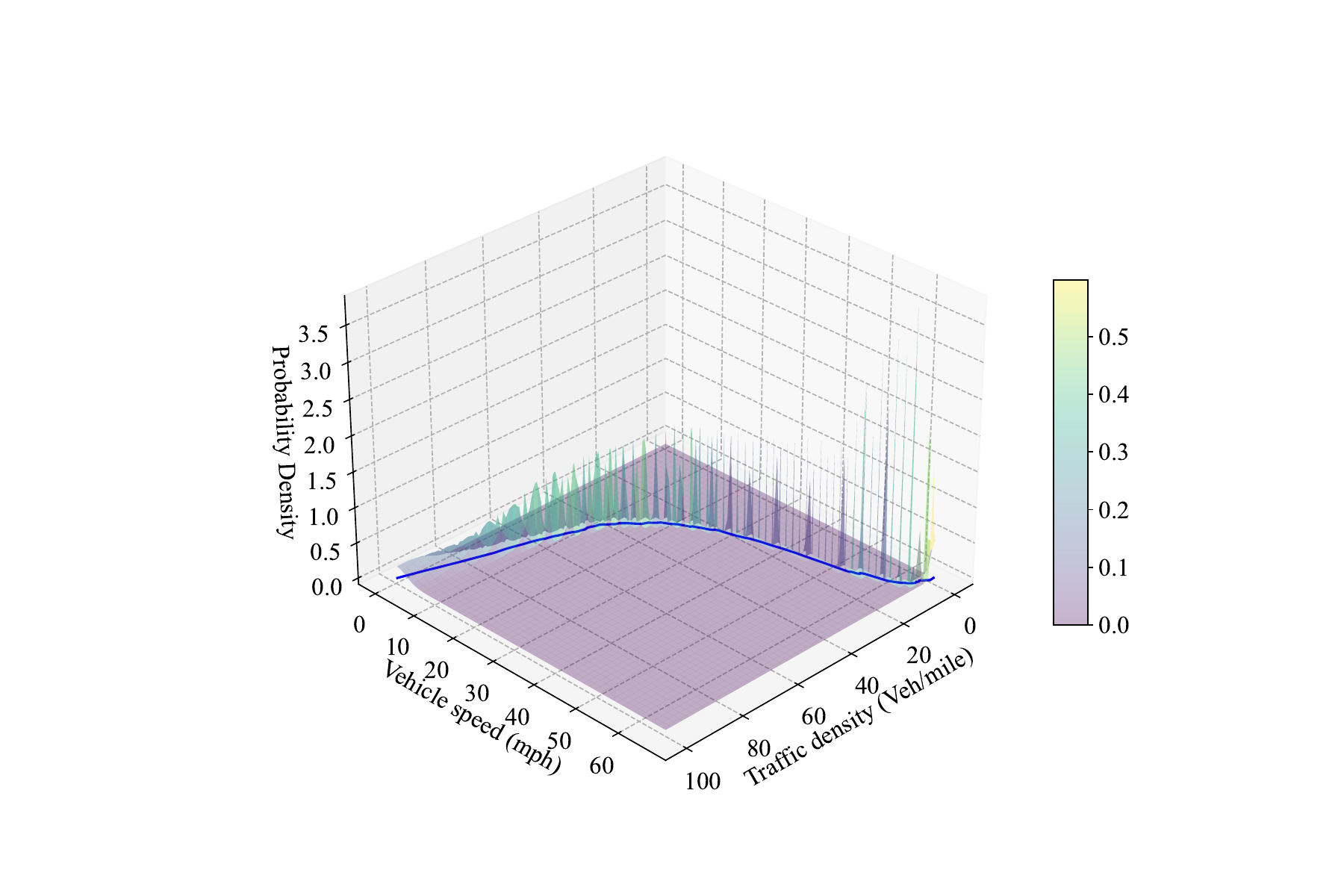}
        \caption{Newell prior}
        \label{figure:35}
    \end{subfigure}
    \begin{subfigure}{.27\textwidth}
        \centering
        \includegraphics[trim=5cm 1cm 4cm 3cm, clip, width=1.0\textwidth]{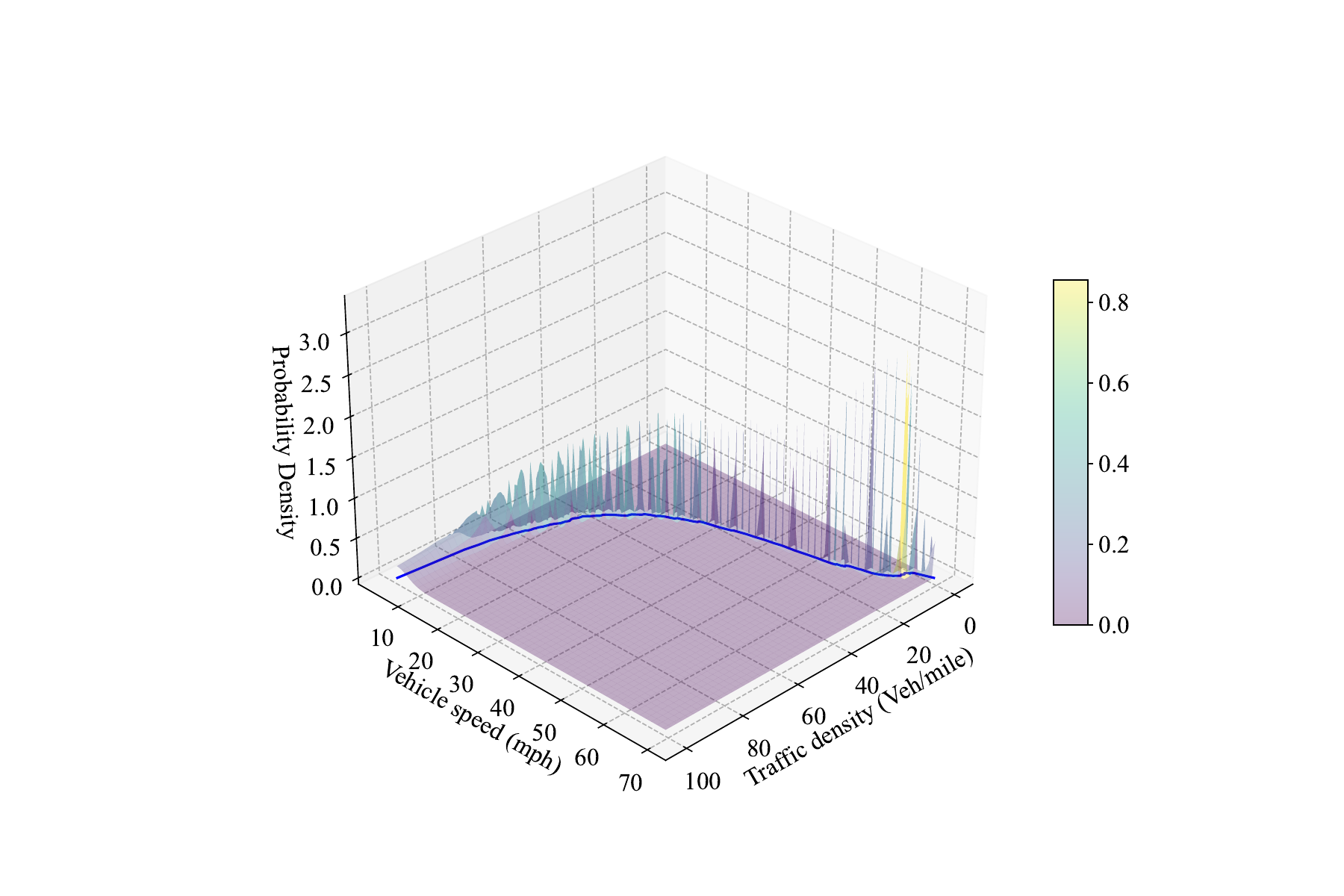}
        \caption{Pipes prior}
        \label{figure:36}
    \end{subfigure}
    \begin{subfigure}{.27\textwidth}
        \centering
        \includegraphics[trim=5cm 1cm 4cm 3cm, clip, width=1.0\textwidth]{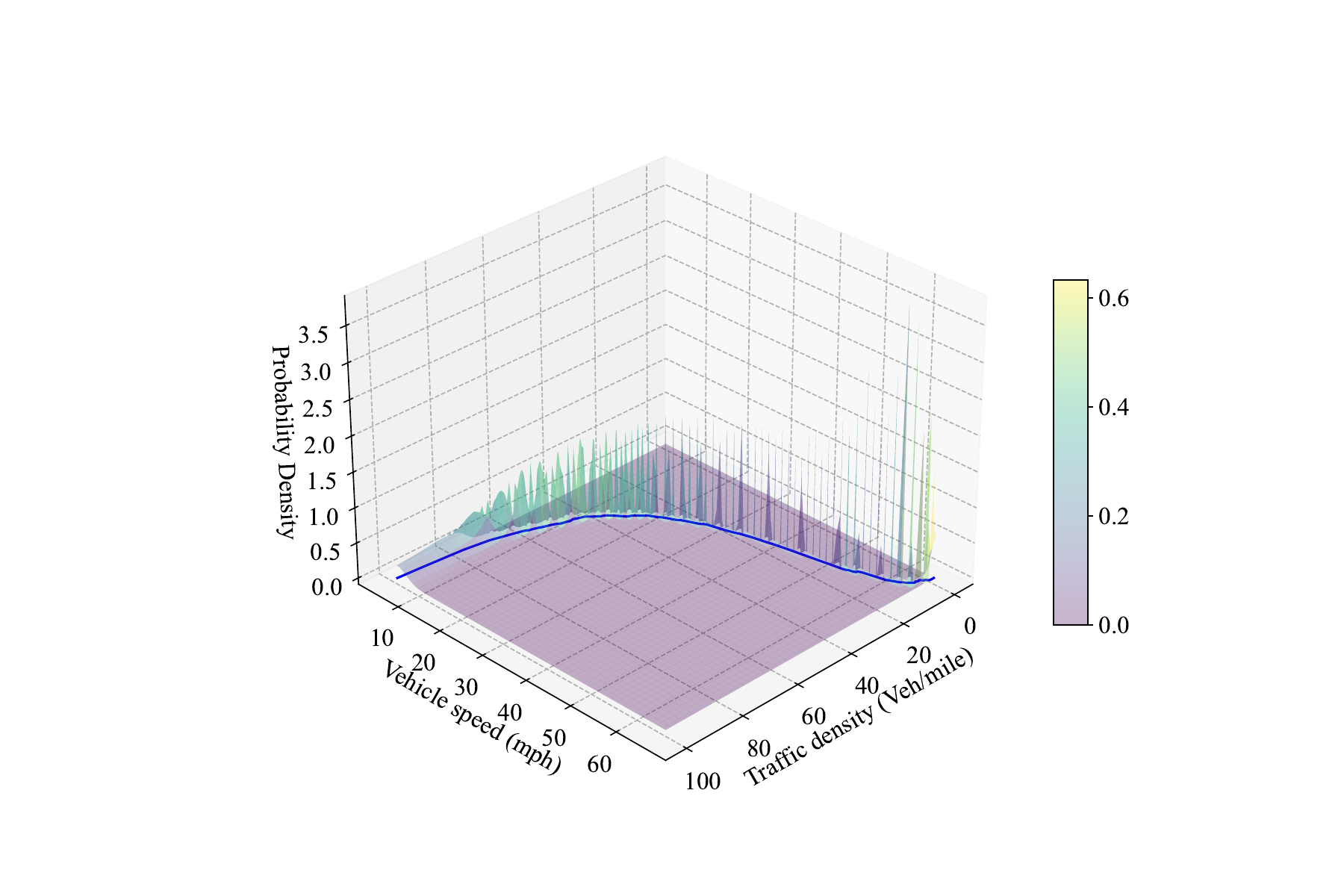}
        \caption{Papageorgiou prior}
        \label{figure:37}
    \end{subfigure}

    \begin{subfigure}{.27\textwidth}
        \centering
        \includegraphics[trim=5cm 1cm 4cm 3cm, clip, width=1.0\textwidth]{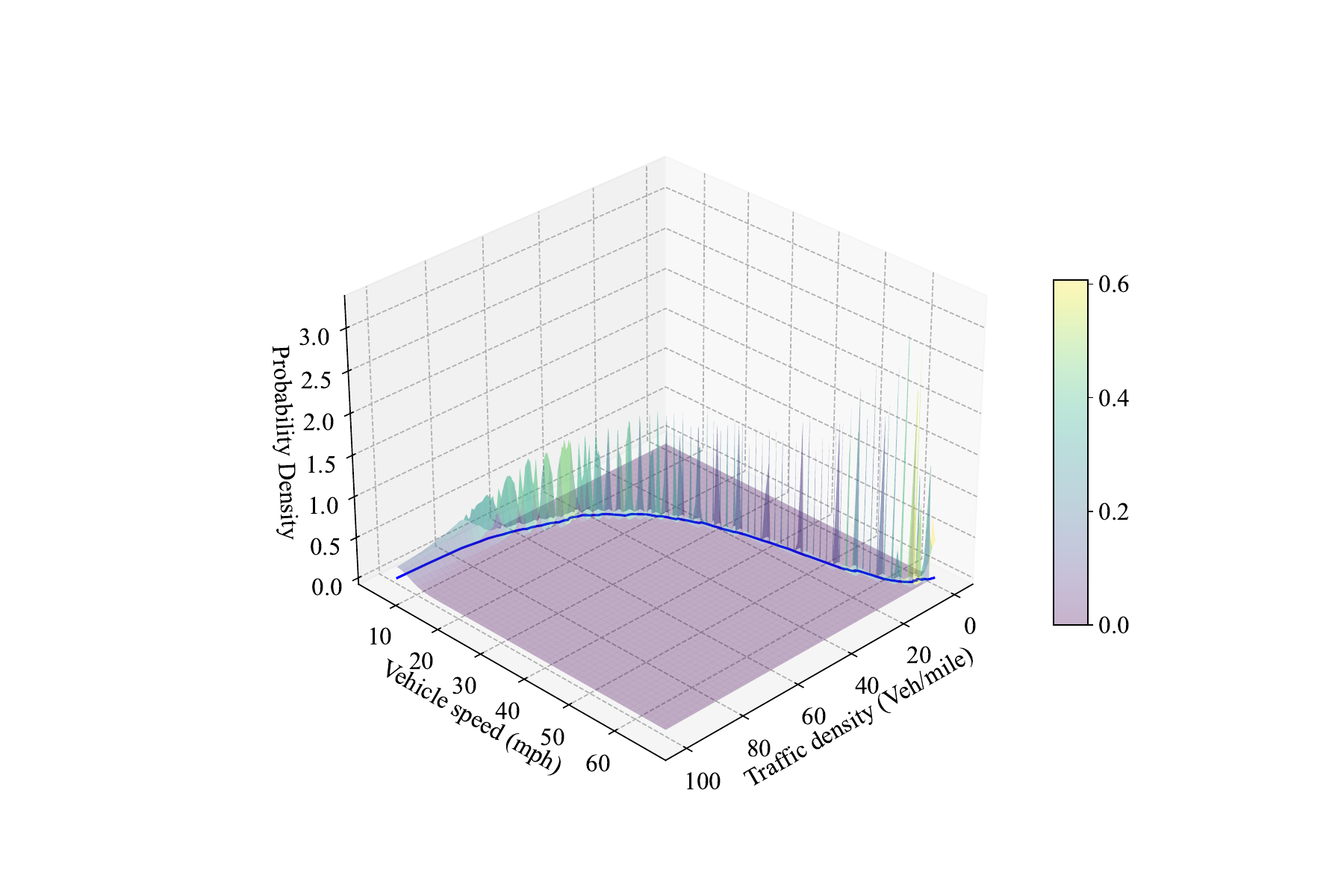}
        \caption{Kerner-Konhauser prior}
        \label{figure:38}
    \end{subfigure}
    \begin{subfigure}{.27\textwidth}
        \centering
        \includegraphics[trim=5cm 1cm 4cm 3cm, clip, width=1.0\textwidth]{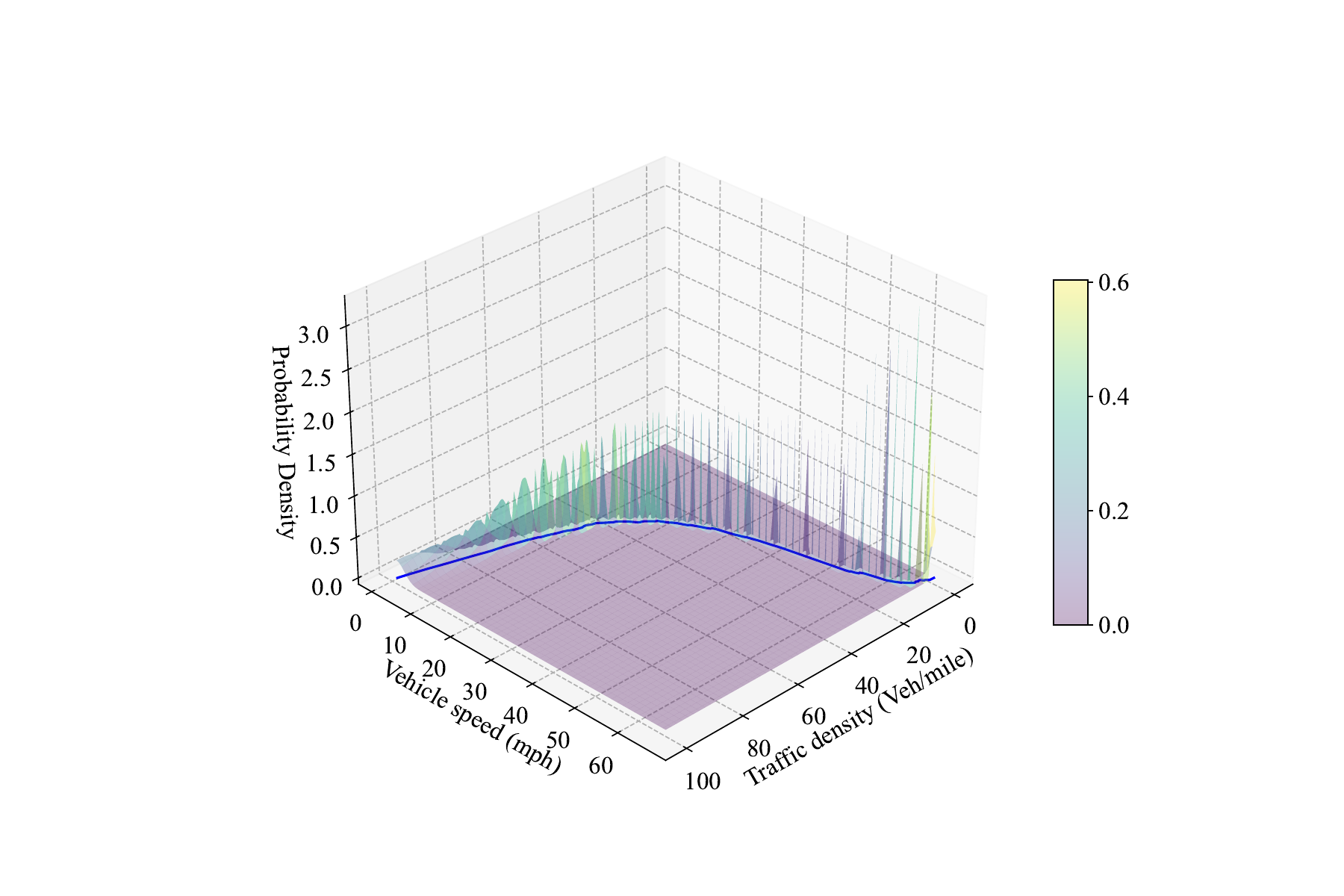}
    \caption{Del-Castillo-Benitez prior}
    \label{figure:39}
    \end{subfigure}
    \begin{subfigure}{.27\textwidth}
        \centering
        \includegraphics[trim=5cm 1cm 4cm 3cm, clip, width=1.0\textwidth]{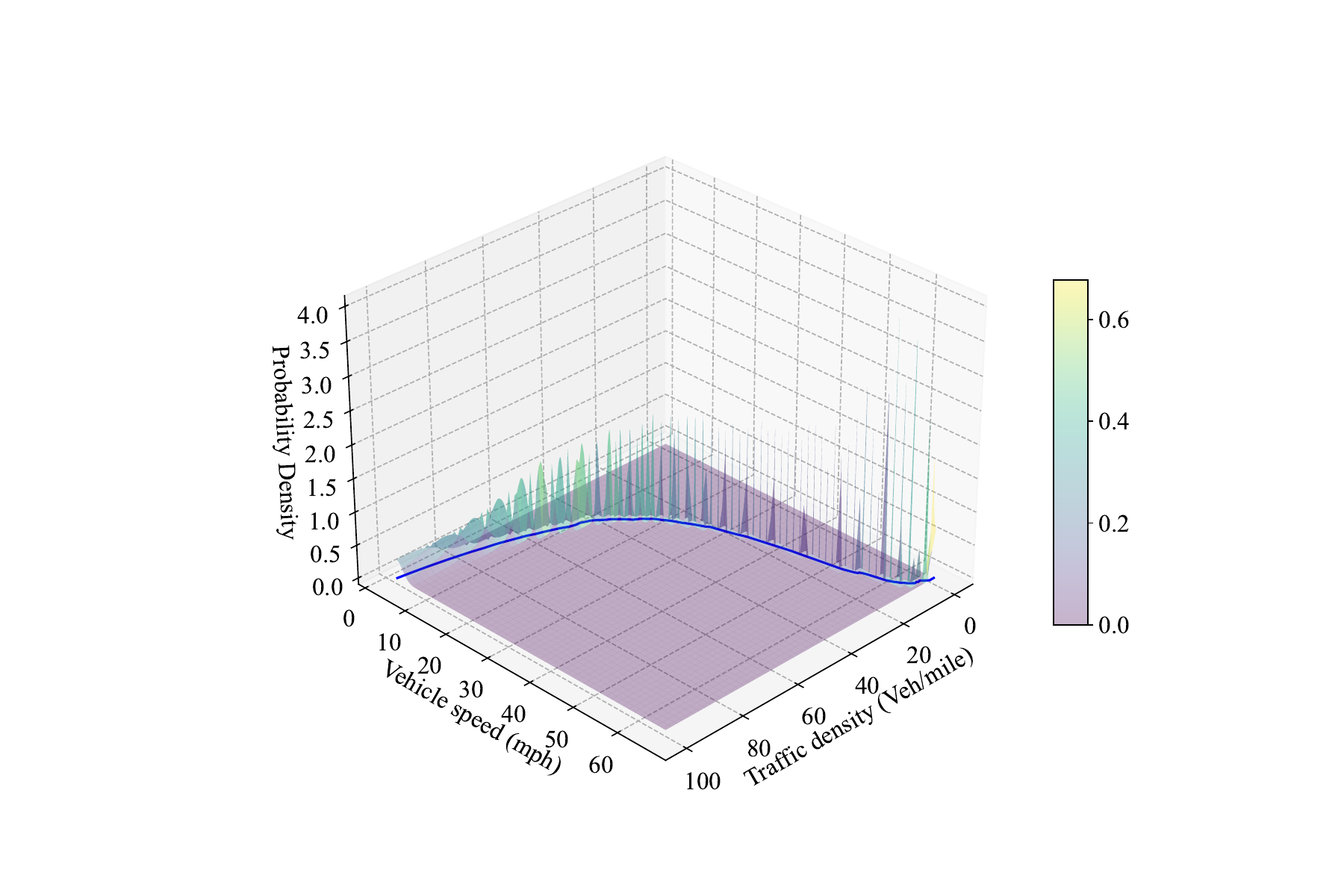}
        \caption{Cheng prior}
        \label{figure:40}
    \end{subfigure}
    
    \begin{subfigure}{.27\textwidth}
        \centering
        \includegraphics[trim=5cm 1cm 4cm 3cm, clip, width=1.0\textwidth]{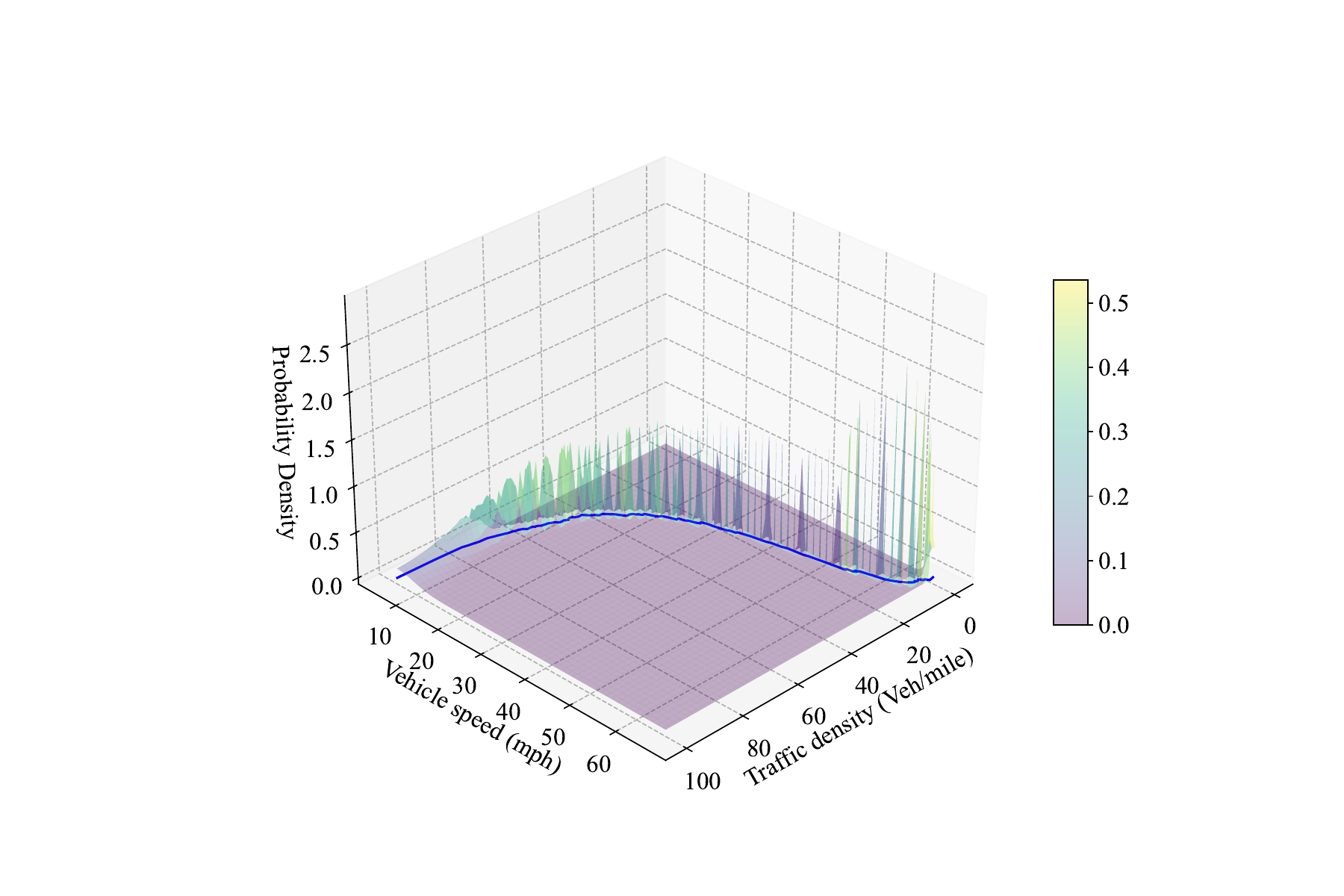}
        \caption{Mac-Nicholas prior}
        \label{figure:41}
    \end{subfigure}
    \begin{subfigure}{.27\textwidth}
        \centering
        \includegraphics[trim=5cm 1cm 4cm 3cm, clip, width=1.0\textwidth]{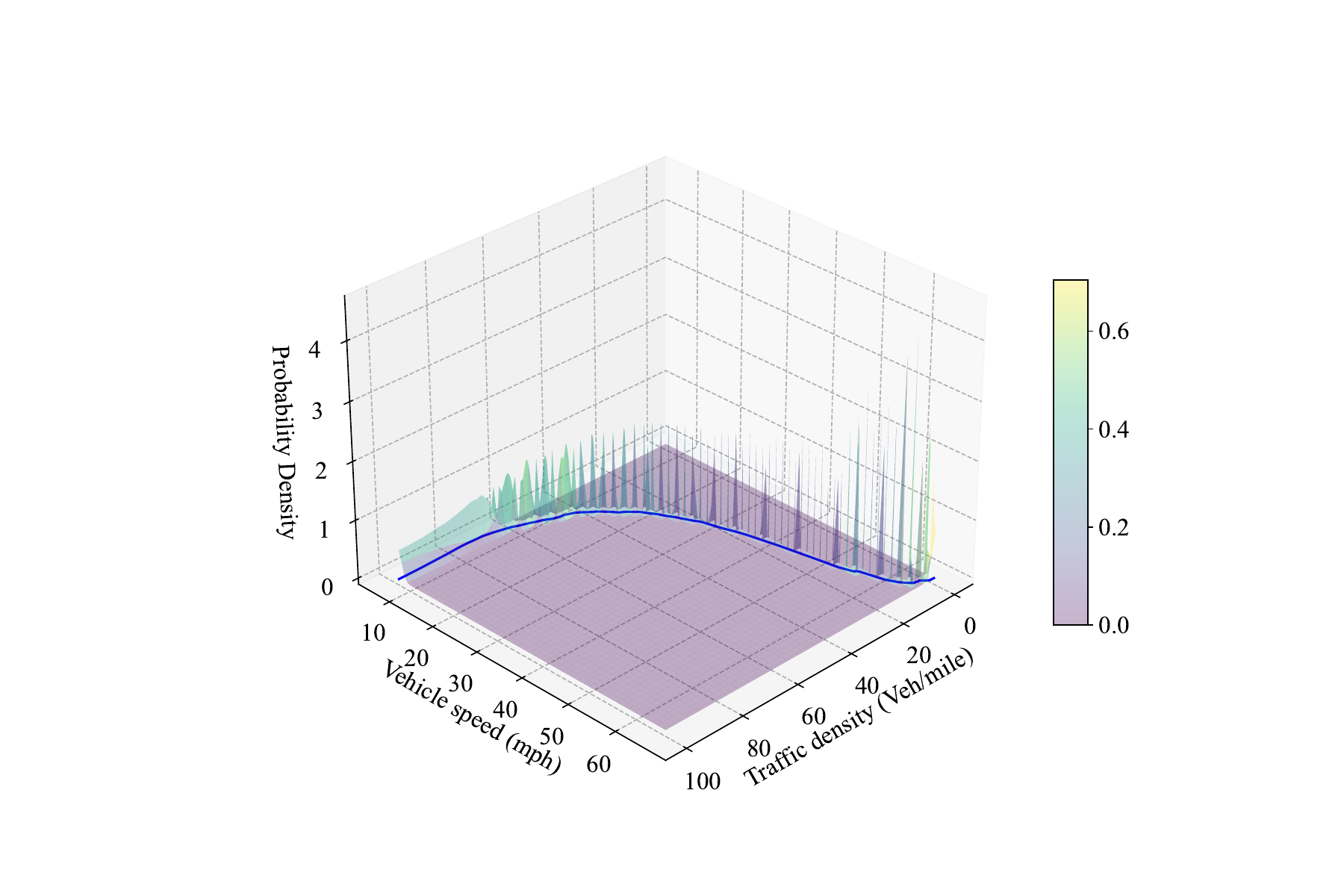}
    \caption{Wang model prior}
    \label{figure:42}
    \end{subfigure}
    \begin{subfigure}{.27\textwidth}
        \centering
        \includegraphics[trim=5cm 1cm 4cm 3cm, clip, width=1.0\textwidth]{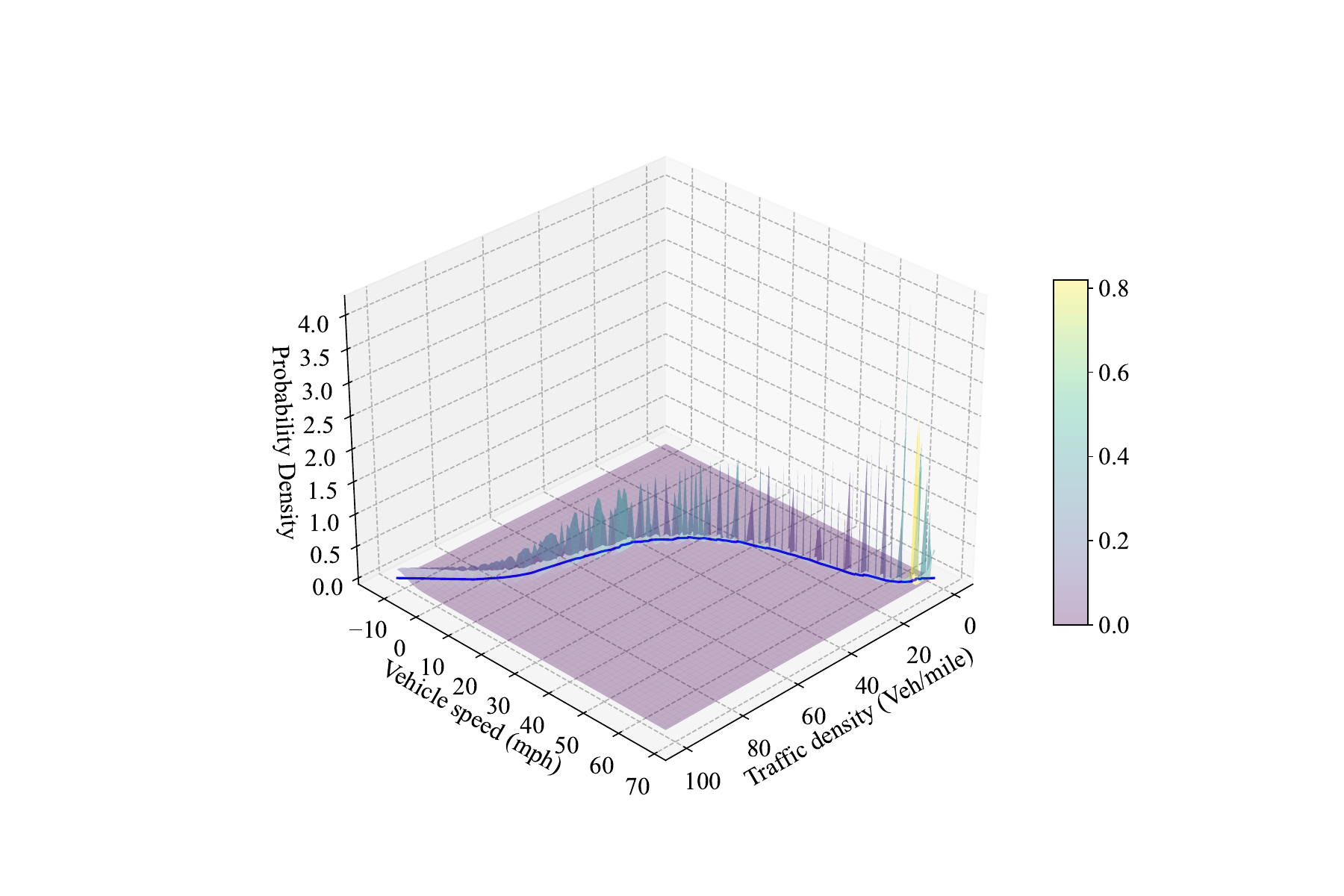}
        \caption{Jayakrishnan prior}
        \label{figure:43}
    \end{subfigure}
    
    \begin{subfigure}{.27\textwidth}
        \centering
        \includegraphics[trim=5cm 1cm 4cm 3cm, clip, width=1.0\textwidth]{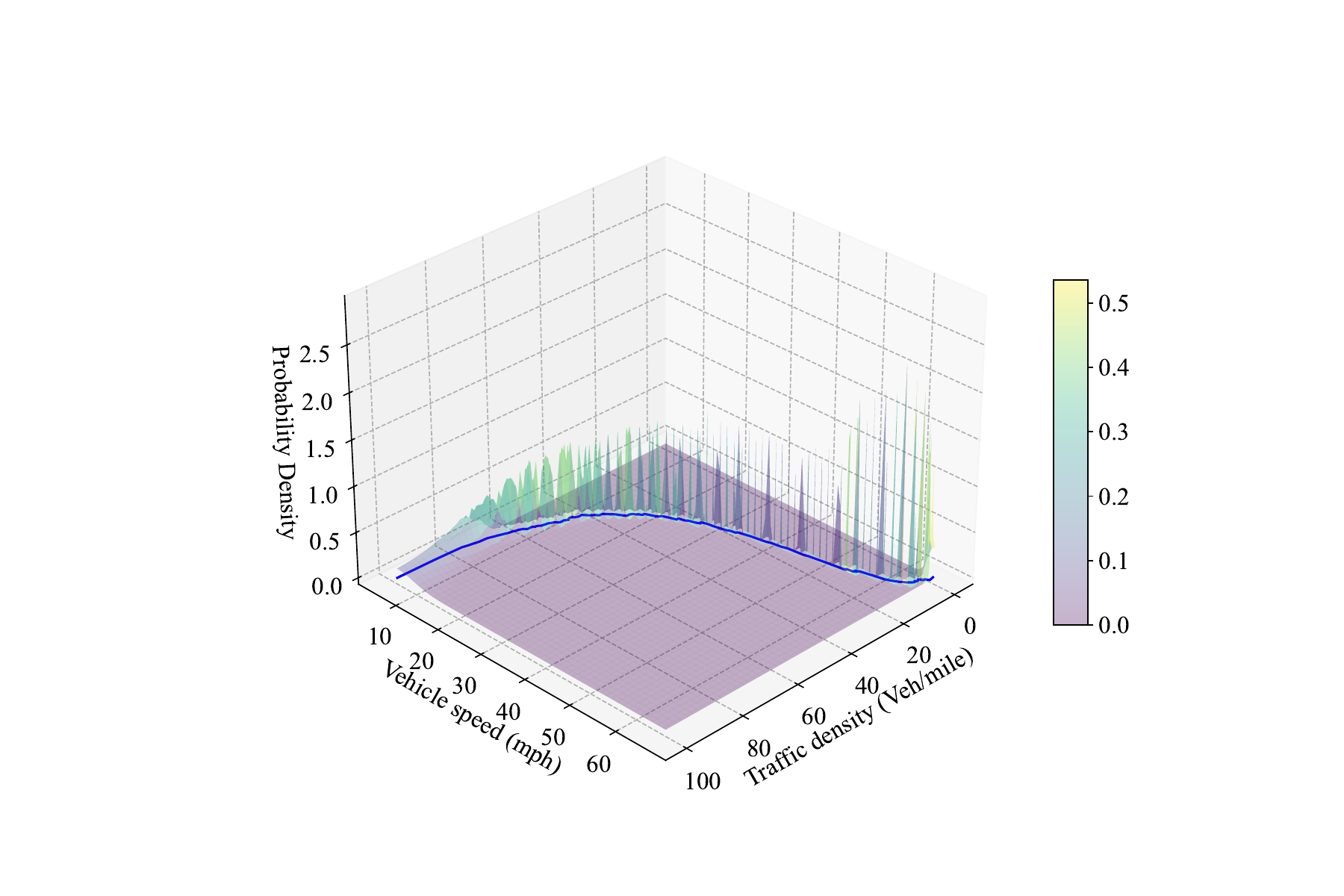}
        \caption{Pure GP regression}
        \label{figure:44}
    \end{subfigure}
    \hfill
    \caption{3D Visualization of speed-density relationships based on different priors}
    \label{figure:45}
\end{figure}
\subsection{Model validation}
\par \par Figures \ref{figure:18} through \ref{figure:30} illustrate the speed-density relationships in a two-dimensional depiction based on prior assumptions. The estimation curve represents the posterior mean function curve, which embodies the speed-density relationship with the maximum likelihood. This means that for any given density value, the speed value situated along this curve has the highest likelihood of occurrence in empirical scenarios. 
\par The stochastic fundamental diagram model proposed in this study is essentially a Gaussian process. Its domain of definition includes an infinite one-dimensional Gaussian distribution looking towards the y-axis direction. The confidence intervals of 90\%, 95\%, and 99\% have been specified, and within the 95\% confidence interval, for any given density, the stochastic model believes the corresponding speed has a 95\% likelihood of falling within this specified range. 
\par Furthermore, Figures \ref{figure:32} through \ref{figure:44} show the corresponding three-dimensional depiction of the speed-density relationships, informed by various prior assumptions, referring to Figure \ref{figure:4}. These figures illustrate the different probability densities of traffic speed given any traffic density. Therefore in this paper, the distribution of traffic speed, denoted as $w(v)$, is considered as a function of traffic density, denoted as $f(\rho)$. This is different from the conventional deterministic models, where the traffic speed is solely a function of traffic density (i.e., $v = f(\rho)$). Here, we take into account the fact that the distribution of traffic speed is a function of traffic density (i.e., $w(v) = f(\rho)$). $w(v) = f(\rho)$ depicts the overall traffic pattern of the density-speed relationship under all situations, and the generated posterior mean function curve $v = m(\rho)$ can be regarded as a deterministic representation of the equilibrium density-speed relationship under a steady state.
\par Two metrics, Root Mean Squared Error (RMSE) and Mean Absolute Percentage Error (MAPE), are employed to evaluate the accuracy of fitting models to empirical data. The formulas for these metrics are as follows:
\begin{equation}
    \mathbf{RMSE} = \sqrt{\frac{1}{n}\sum_{i = 1}^{n}(v^{o}_{i} - v^{e}_{i})^{2}} \label{eq:37}
\end{equation}
\begin{equation}
    \mathbf{MAPE} = \frac{1}{n}\sum_{i = 1}^{n}\left|\frac{v^{o}_{i} - v^{e}_{i}}{v^{o}_{i}}\right| \times 100\% \label{eq:38}
\end{equation}
where \(v^{o}_{i}\) denotes the observed speed corresponding to the observation density \(\rho^{o}_{i}\), and \(v^{e}_{i}\) represents the estimated speed for the same observation density. These metrics provide insights into the comparative performance of different models. 
\par Table \ref{Table:4} shows the comparison of the results of different empirical models. Table \ref{Table:5} compares the estimation accuracy of various empirical prior Gaussian regression models with the pure Gaussian process regression model across different sampling techniques. Table \ref{Table:6} shows the percentage of the empirical data points that fall within the 95\% confidence interval. Here, $\mathcal{RS}$, $\mathcal{SS}$, $\mathcal{CS}$, and $\mathcal{WRS}$ denote simple random sampling, systematic sampling, cluster sampling, and weighted random sampling, respectively. We also used the I-80 detector dataset from NGSIM for testing, and all results are shown in the appendix from table \ref{Table:7} to \ref{Table:9}.
\par  First, it is observed that among empirical models, the most precise one is Cheng's model (\cite{cheng2021s}). However, the Full Gaussian Process (GP) Model surpasses these empirical models in terms of accuracy. The Pure GP Model's superior performance is achieved with only 288 data pairs as opposed to the 44,787 utilized by the empirical models; also, from Table \ref{Table:6}, all models are well-fitted because they all capture the distribution of the empirical data under sparse approximation. Also, we should notice that model performances are nearly the same under all sampling methods for all models, which indicates that all sampling methods in our paper are suitable choices for the SGPR modeling.
\par Secondly, to figure out whether the sparse regression successfully captures the stochastic of the density-speed relationship. The percentage 
 density-speed points that fall within the 95\% confidence interval (PWCI) have been tested. The percentage of points that fall within the 95\% confidence interval (PWCI) of all models is higher than 94\%. The superior performance of sparse Gaussian process regression methods in accurately modeling the stochastic density-speed relationship is convincingly illustrated by the above facts. 
 \par Thirdly, the results shown in Table \ref{Table:5} unequivocally indicate that the Pure GP Model is more accurate compared to other models across all sampling methods employed. This leads to the inference that an accurate prior mean function may not sufficiently compensate for the limitations inherent in sparse approximations when seeking to enhance estimation accuracy. A purely data-driven approach may yield more accurate results in situations with a relatively small amount of quality data. Notably,  the same results are also shown in the I-80 dataset, shown in table \ref{Table:7} to \ref{Table:9}. Such findings imply that even a training dataset with a modest size can contain sufficient information for effective model training. Consequently, these insights prompt further questions: \textbf{\textit{"What is the minimal size of a dataset that remains effective? How much data is necessary to counterbalance the advantages of employing an empirical prior?"}}
 \par Therefore, a series of experiments have been conducted to find the fitting accuracy of both empirical prior Sparse Gaussian process regression model (EPGPR) and pure Sparse Gaussian process regression model under increasing training dataset size. As shown in Figure \ref{figure:46} to \ref{figure:56} and \ref{figure:58} to \ref{figure:68}, a comparison has been made to discover both RSME and MAPE between the pure GPR model and different EPGPR models with the increase of the training dataset size. 
 Our findings indicated no significant improvement in accuracy after the training dataset size exceeded 300 for all the models. It is worth noting that all results shown in Figure \ref{figure:46} to \ref{figure:68} are based on weighted random sampling. Experimental results based on other sampling methods are shown in the supplementary materials. The warm color area indicates the ESME or MPAE of the relevant proposed EPGPR model exceeds that of the pure SGPR model. Conversely, the cold color area indicates the ESME or MPAE of the relevant proposed EPGPR model is lower than the pure SGPR model.  Drawing insights from the above pictures,  only when well-fitted empirical models are utilized as the mean function, and notably, this improvement is largely confined to scenarios involving small training datasets. Despite the potential benefits offered by well-calibrated priors, such as those proposed in Wang's model (\cite{wang2011logistic}), Cheng's model (\cite{cheng2021s}), Papageorgiou's model (\cite{papageorgiou1989macroscopic}), Kerner's model \cite{kerner1994structure}, and Del-Castillo's model (\cite{del1995functional-part-1}), the empirical advantage diminishes as the size of the training dataset expands. Considering RSME and MAPE simultaneously, only the prior based on Cheng's and Kerner's models can benefit the EPGPR under a very small training dataset. Employing a judicious sampling strategy for choosing inducing variables, a training dataset comprising merely 30 density-speed pairs suffices for a purely data-driven model to outperform a model incorporating a well-fitted empirical prior for the GA 400 dataset. The findings of this study suggest that the speed-density relationship may not exhibit a complex nonlinear nature. Applying highly specific prior may lead to overfitting of the entire model. Consequently, selecting a suitable sampling method, coupled with the advantage of relatively clean and ample datasets, indicates that data-driven modeling could emerge as a more effective approach for developing stochastic fundamental diagram models.
\begin{table}[htbp]
  \centering
  \caption{Comparison of the result of different empirical models}
  \label{Table:4}
  \begin{tabular}{>{\centering\arraybackslash}p{0.5\linewidth} >{\centering\arraybackslash}p{0.25\linewidth} >{\centering\arraybackslash}p{0.25\linewidth}}
    \thicktoprule
    \textbf{Model} & \textbf{Speed RMSE (mph)} & \textbf{Speed MAPE}\\
    \midrule
       Original Greenshields' Model & 15.43 & 26.80\% \\
       Original Greenberg's Model & 68.59 & 116.52\% \\
       Original Underwood's Model & 15.58 & 34.82\% \\
       Original Newell's Model & 14.23 & 29.33\% \\
       Original Pipes' Model & 4.25 & 7.41\% \\
       Original Papageorgiou's Model & 5.52 & 8.77\% \\
       Original Kerner-Konhause's Model & 6.73 & 12.49\% \\
       Original Del Castillo-Benitez's Model & 12.14 & 18.97\% \\
       Original Jayakrishnan's Model & 15.43 & 26.80\% \\
       Original Cheng's Model & 3.81 & 5.64\% \\
       Original Wang's Model & 3.39 & 4.48\% \\
    \thickbottomrule
  \end{tabular}
\end{table}
\begin{table}[htbp]
\centering
\caption{Comparison of the result of different EPNR models with Exponential kernel function (m = 288)}
\label{Table:5}
\begin{tabularx}{\textwidth}{Z*{8}{Y}}
\thicktoprule
\textbf{EPNR Model} & \multicolumn{4}{c}{\textbf{Speed RMSE (mph)}} & \multicolumn{4}{c}{\textbf{Speed MAPE}} \\ \cline{2-9}
& $\mathcal{RS}$ & $\mathcal{SS}$ & $\mathcal{CS}$ & $\mathcal{WRS}$ & $\mathcal{RS}$ & $\mathcal{SS}$ & $\mathcal{CS}$ & $\mathcal{WRS}$ \\ \midrule
Pure GP  & 3.31  & 3.32  & 3.31 & 3.31 & 4.43\% & 4.43\% & 4.43\% & 4.43\% \\ 
Cheng's  & 3.32  & 3.32  & 3.32 & 3.32 & 4.43\% & 4.43\% & 4.43\% & 4.43\% \\ 
Wang's & 3.32  & 3.32  & 3.32 & 3.32 & 4.44\% & 4.44\% & 4.44\% & 4.43\% \\
Greenshields's  & 3.32  & 3.32  & 3.32 & 3.32 & 4.44\% & 4.44\% & 4.43\% & 4.43\% \\ 
Greenberg's  & 3.31  & 3.32  & 3.32 & 3.32 & 4.43\% & 4.43\% & 4.43\% & 4.43\% \\ 
Underwood's  & 3.32  & 3.32  & 3.32 & 3.32 & 4.43\% & 4.43\% & 4.43\% & 4.43\% \\ 
Newell's  & 3.32  & 3.32  & 3.32 & 3.32 & 4.43\% & 4.43\% & 4.43\% & 4.43\% \\ 
Pipes's & 3.32  & 3.32  & 3.32 & 3.32 & 4.43\% & 4.43\% & 4.43\% & 4.43\% \\ 
Papageorgiou's & 3.32  & 3.32  & 3.32 & 3.32 & 4.43\% & 4.43\% & 4.43\% & 4.43\% \\ 
Kerner's & 3.32  & 3.32  & 3.32 & 3.32 & 4.43\% & 4.43\% & 4.43\% & 4.43\% \\ 
Del-Castillo's  & 3.32  & 3.32  & 3.32 & 3.32 & 4.44\% & 4.44\% & 4.43\% & 4.43\% \\ 
Jayakrishnan's & 3.32  & 3.32  & 3.32 & 3.32 & 4.44\% & 4.44\% & 4.43\% & 4.43\% \\ 
\thickbottomrule
\end{tabularx}
\end{table}
\begin{table}[htbp]
\centering
\caption{Comparison of the result of percentages of points falling within the 95\% confidence interval}
\label{Table:6}
\begin{tabularx}{\textwidth}{Z*{4}{Y}}
\thicktoprule
\textbf{EPNR Model} & \multicolumn{4}{c}{\textbf{PWCI}} \\ \cline{2-5}
& $\mathcal{RS}$ & $\mathcal{SS}$ & $\mathcal{CS}$ & $\mathcal{WRS}$ \\ \midrule
Pure GP & 94.27\% & 94.27\% & 94.29\% & 94.26\% \\ 
Cheng's & 94.22\% & 94.21\% & 94.23\% & 94.21\% \\ 
Wang's & 94.22\% & 94.22\% & 94.22\% & 94.22\% \\
Greenshields's & 94.26\% & 94.25\% & 94.25\% & 94.26\% \\ 
Greenberg's& 94.26\% & 94.26\% & 94.25\% & 94.26\% \\ 
Underwood's & 94.25\% & 94.26\% & 94.25\% & 94.26\% \\ 
Newell's & 94.23\% & 94.23\% & 94.22\% & 94.23\% \\ 
Pipes's  & 94.24\% & 94.23\% & 94.23\% & 94.24\% \\ 
Papageorgiou's  & 94.21\% & 94.21\% & 94.22\% & 94.23\% \\ 
Kerner's  & 94.23\% & 94.23\% & 94.24\% & 94.24\% \\ 
Del-Castillo's  & 94.22\% & 94.23\% & 94.23\% & 94.22\% \\ 
Jayakrishnan's  & 94.25\% & 94.25\% & 94.25\% & 94.26\% \\ 
\thickbottomrule
\end{tabularx}
\end{table}

\begin{figure}[htbp]
    \centering
    \begin{subfigure}{.27\textwidth}
        \centering
        \includegraphics[width=1.0\textwidth]{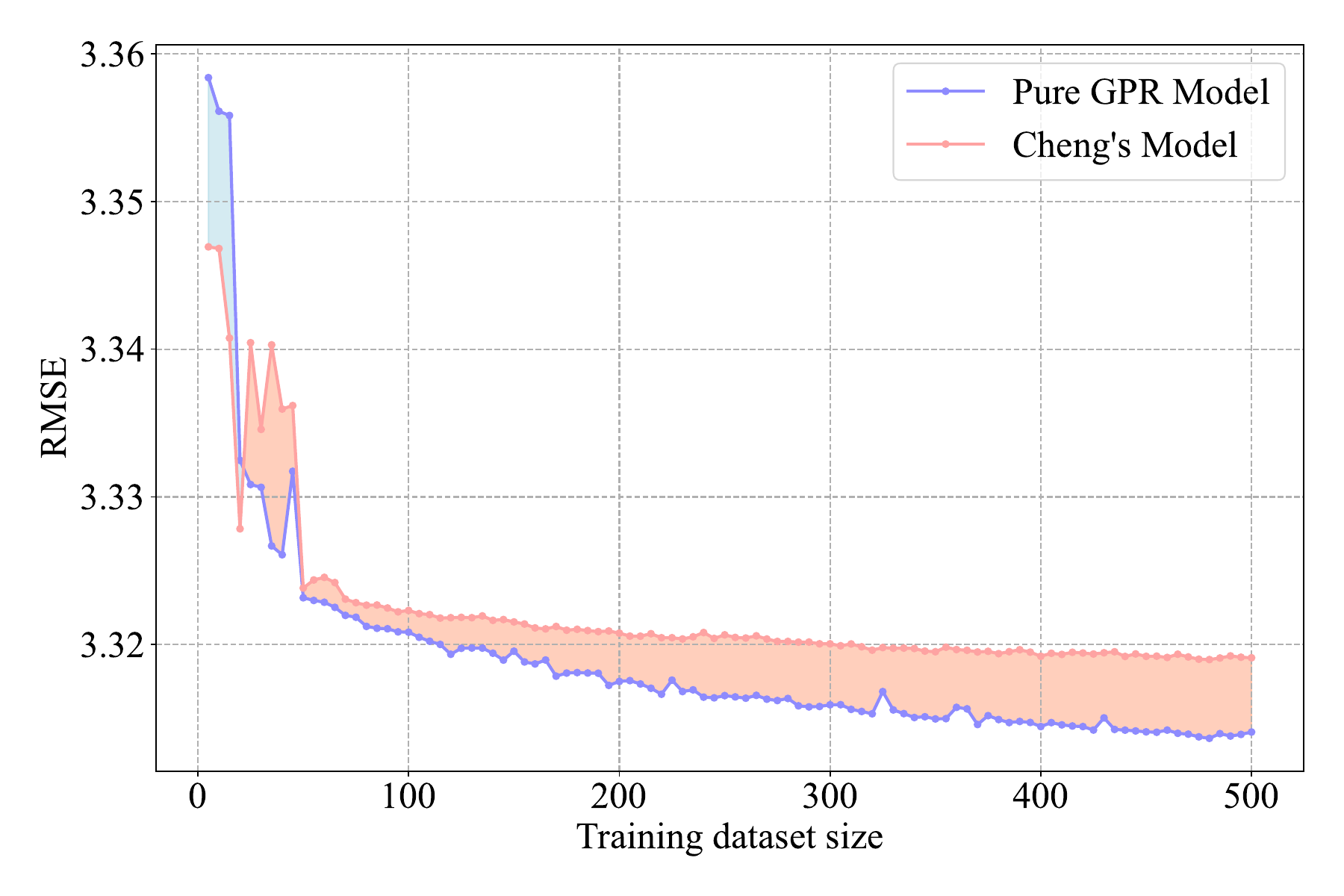}
        \caption{}
        \label{figure:46}
    \end{subfigure}
    \begin{subfigure}{.27\textwidth}
        \centering
        \includegraphics[width=1.0\textwidth]{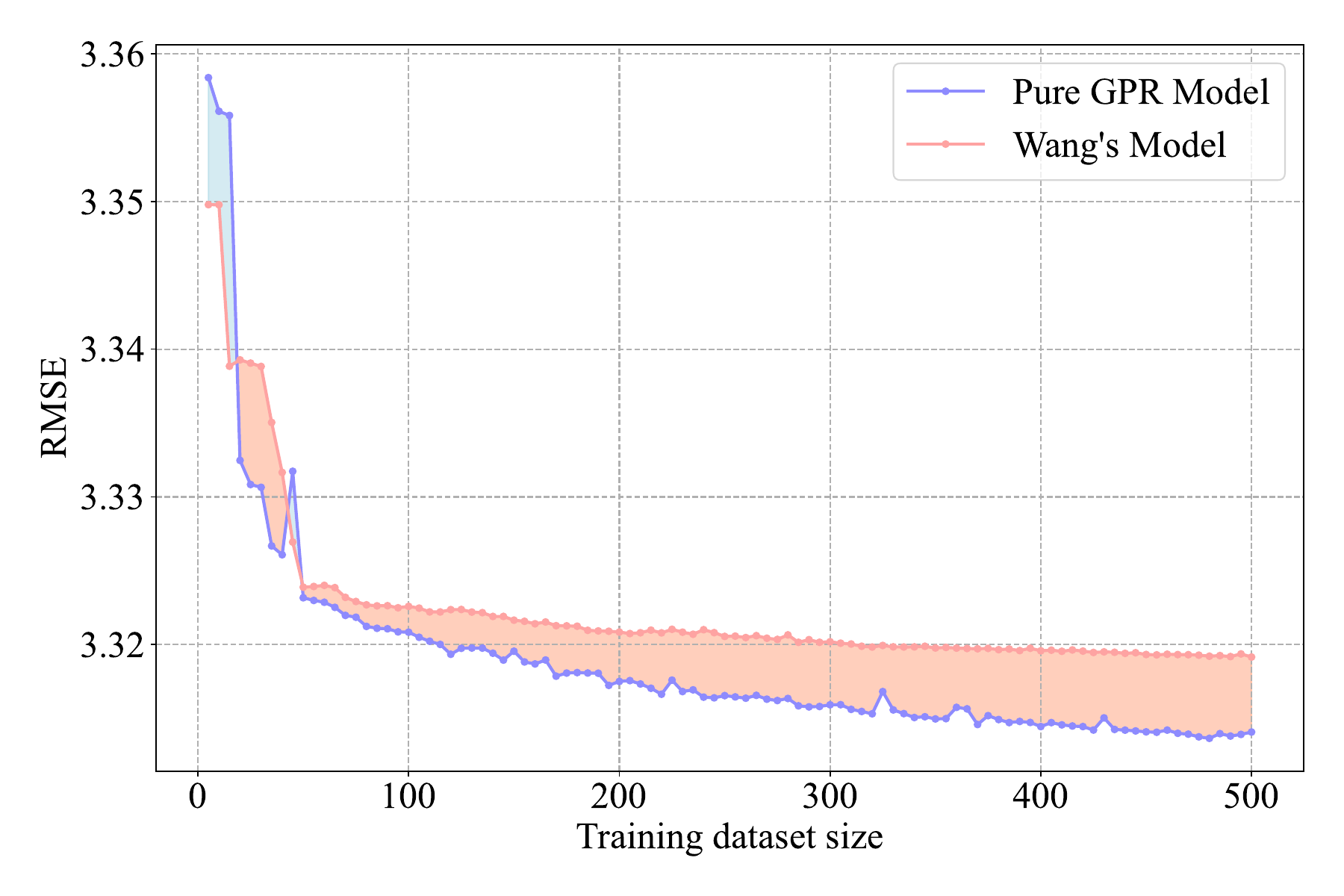}
        \caption{}
        \label{figure:47}
    \end{subfigure}
    \begin{subfigure}{.27\textwidth}
        \centering
        \includegraphics[width=1.0\textwidth]{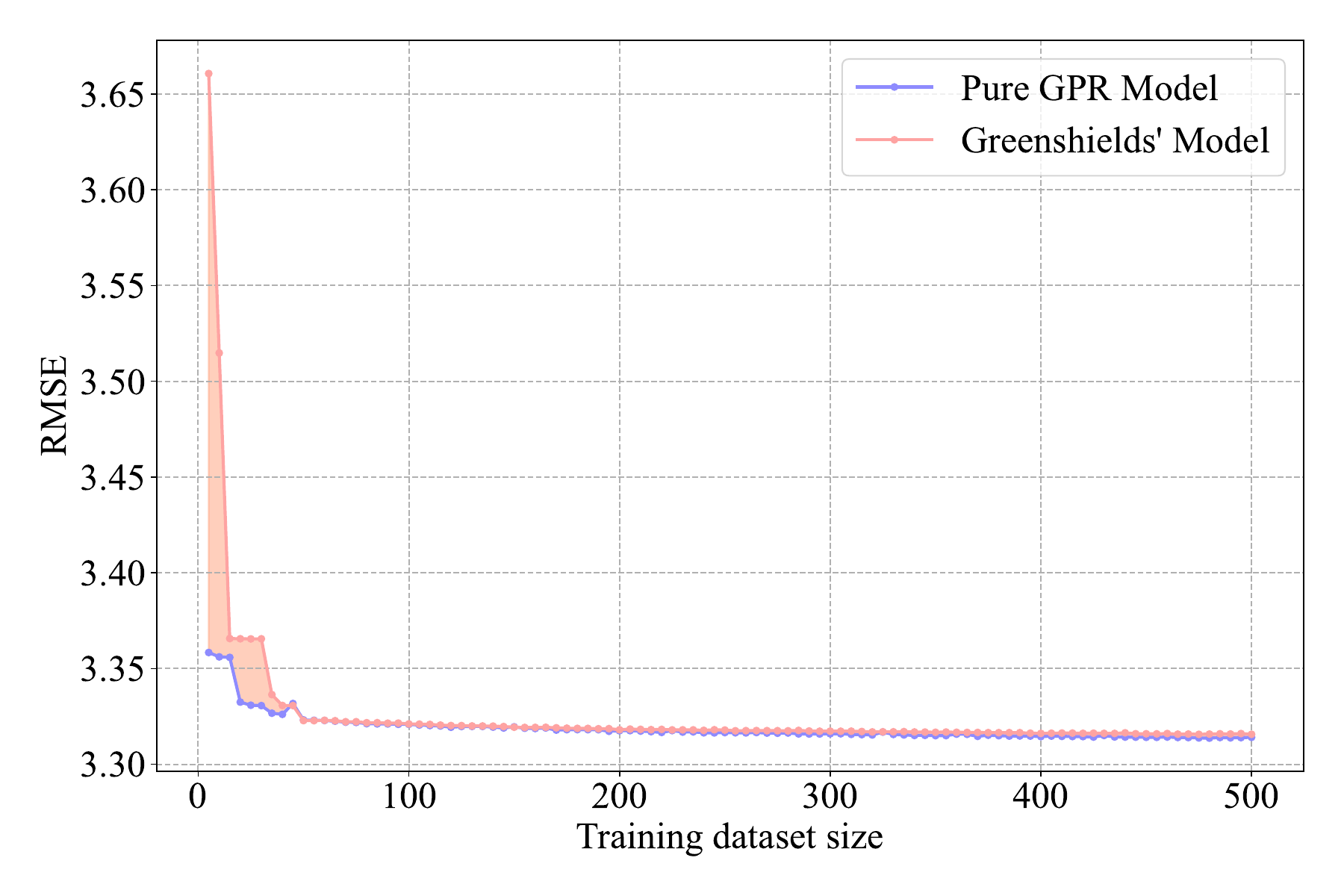}
         \caption{}
        \label{figure:48}
        \label{fig:newell}
    \end{subfigure}

    \begin{subfigure}{.27\textwidth}
        \centering
        \includegraphics[width=1.0\textwidth]{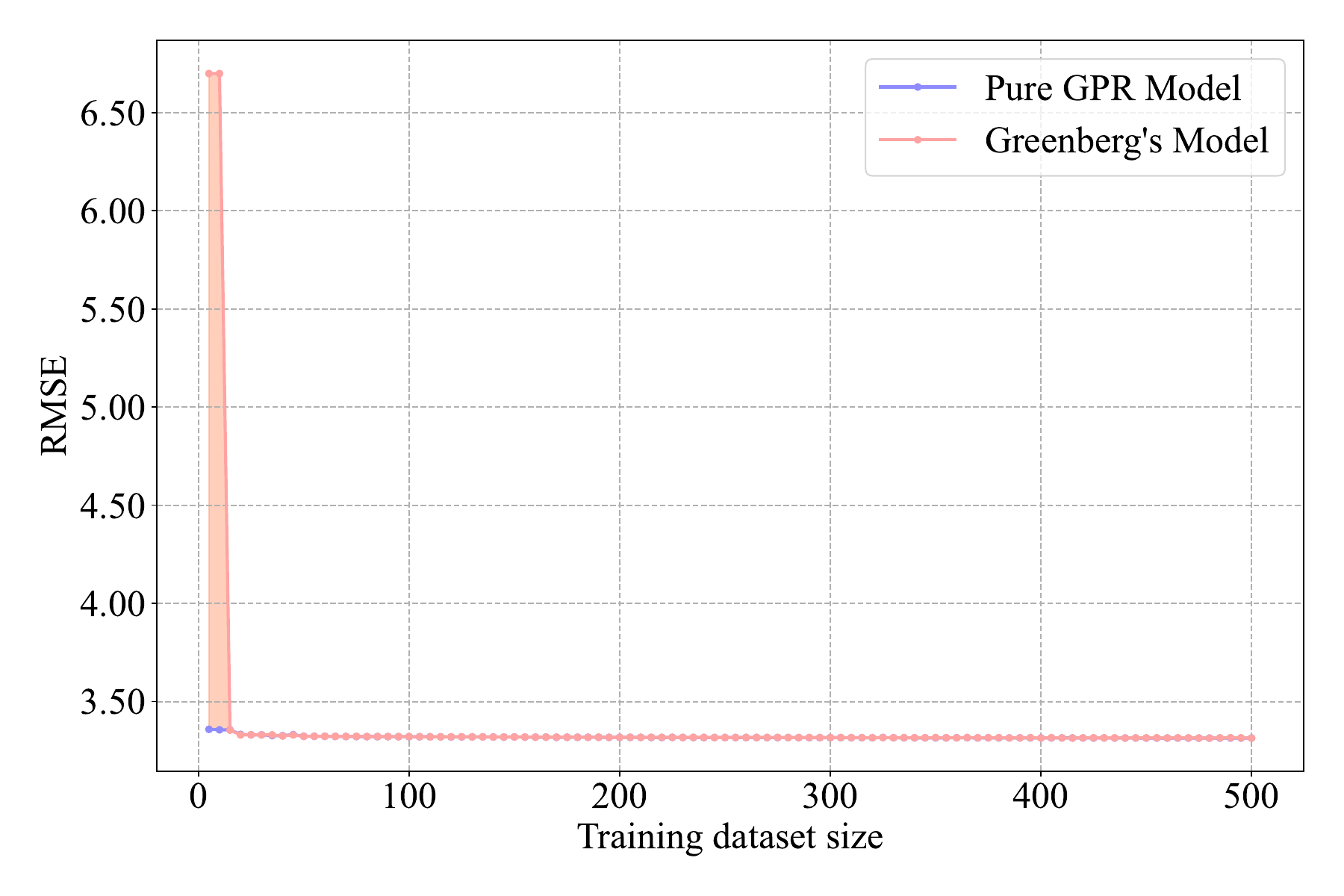}
        \caption{}
        \label{figure:49}
    \end{subfigure}
    \begin{subfigure}{.27\textwidth}
        \centering
        \includegraphics[width=1.0\textwidth]{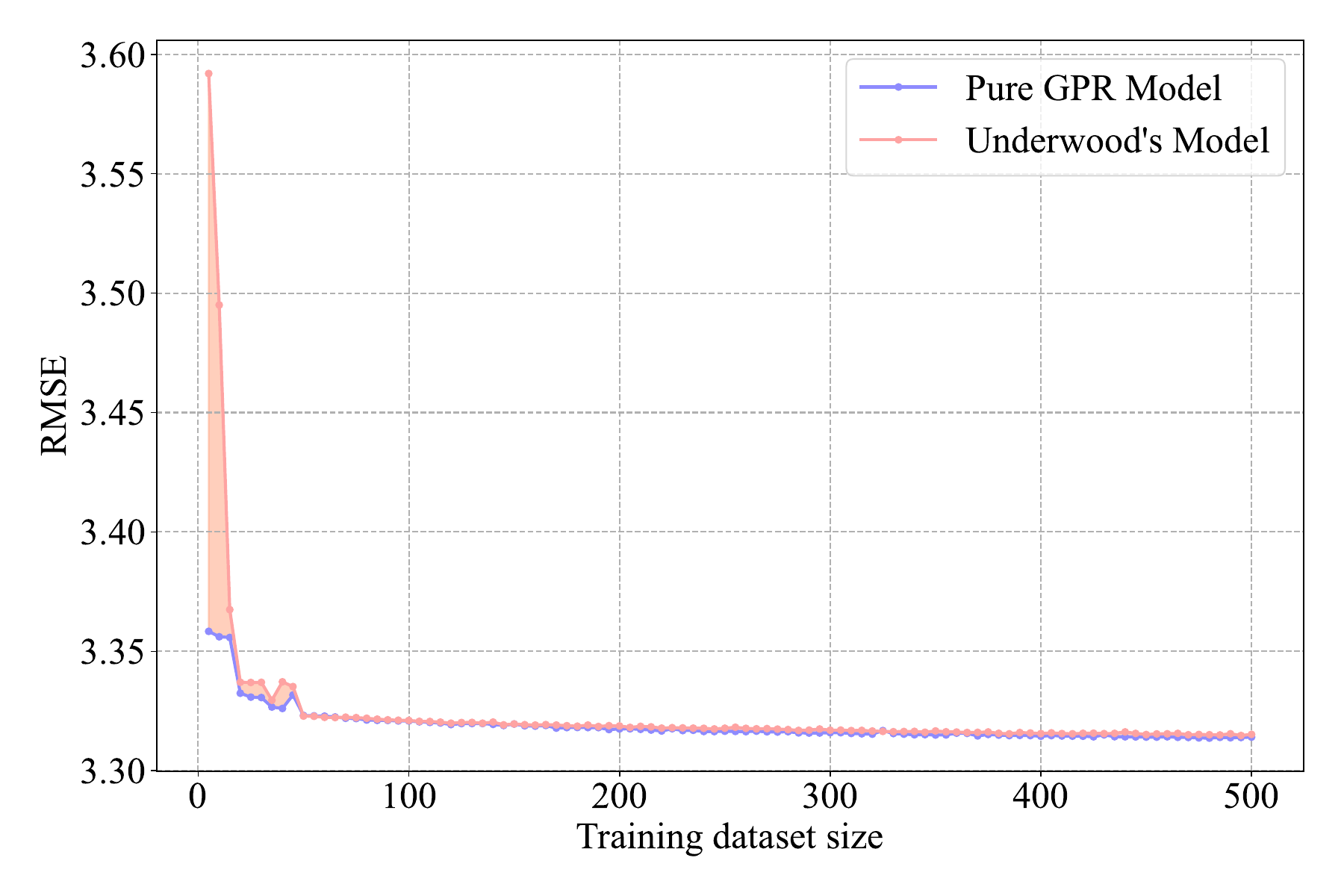}
        \caption{}
        \label{figure:50}
    \end{subfigure}
    \begin{subfigure}{.27\textwidth}
        \centering
        \includegraphics[width=1.0\textwidth]{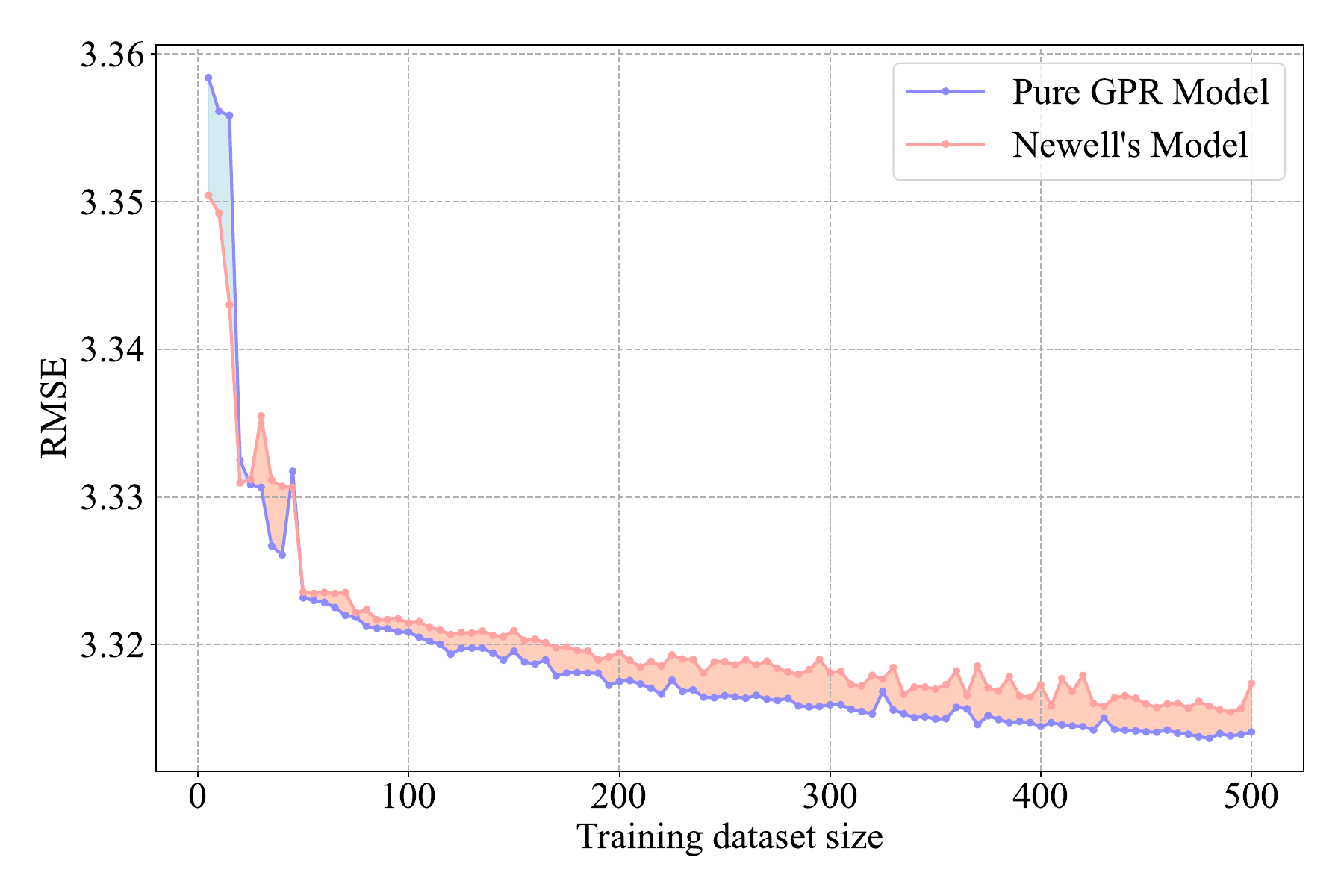}
        \caption{}
        \label{figure:51}
    \end{subfigure}

    \begin{subfigure}{.27\textwidth}
        \centering
        \includegraphics[width=1.0\textwidth]{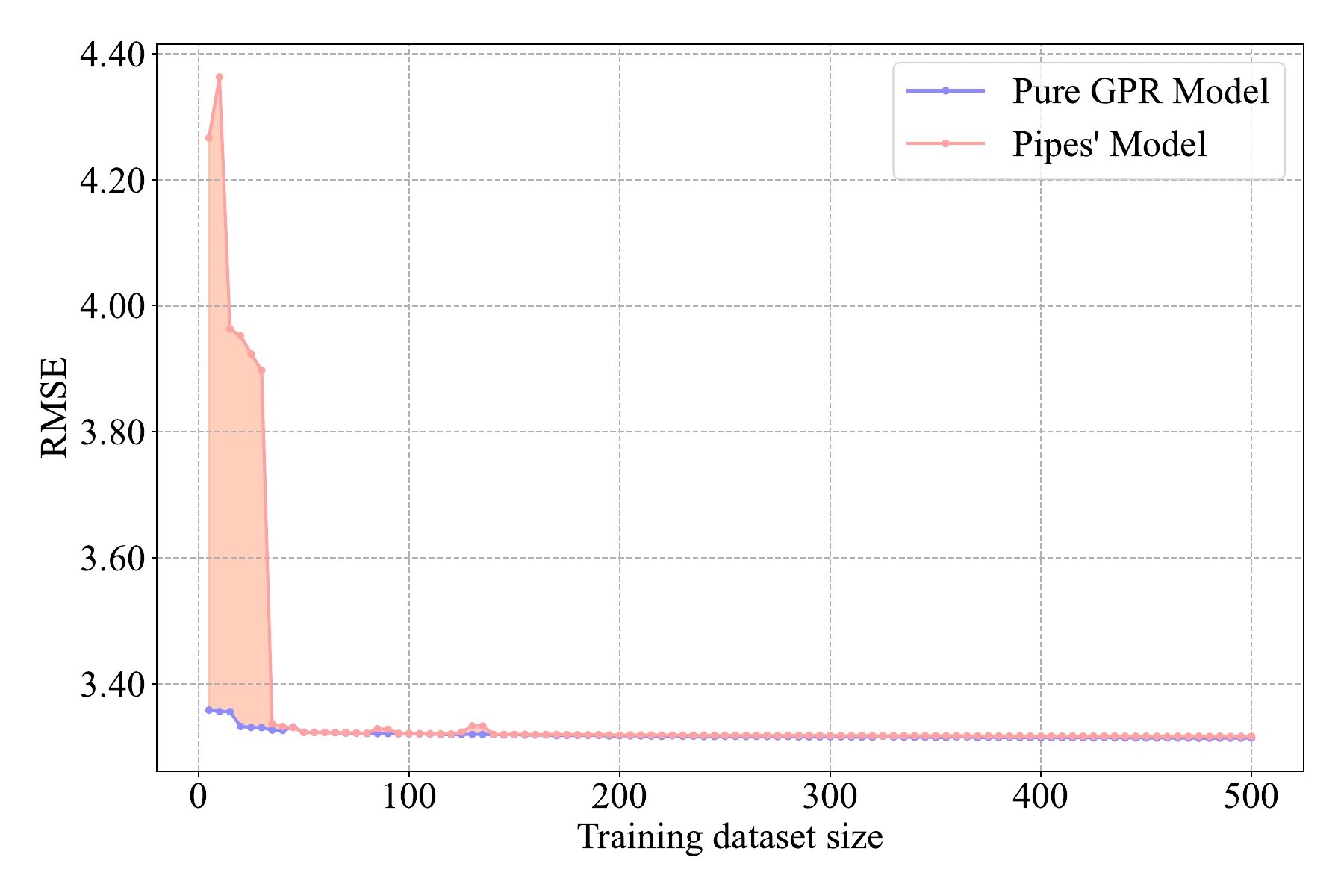}
        \caption{}
        \label{figure:52}
    \end{subfigure}
    \begin{subfigure}{.27\textwidth}
        \centering
        \includegraphics[width=1.0\textwidth]{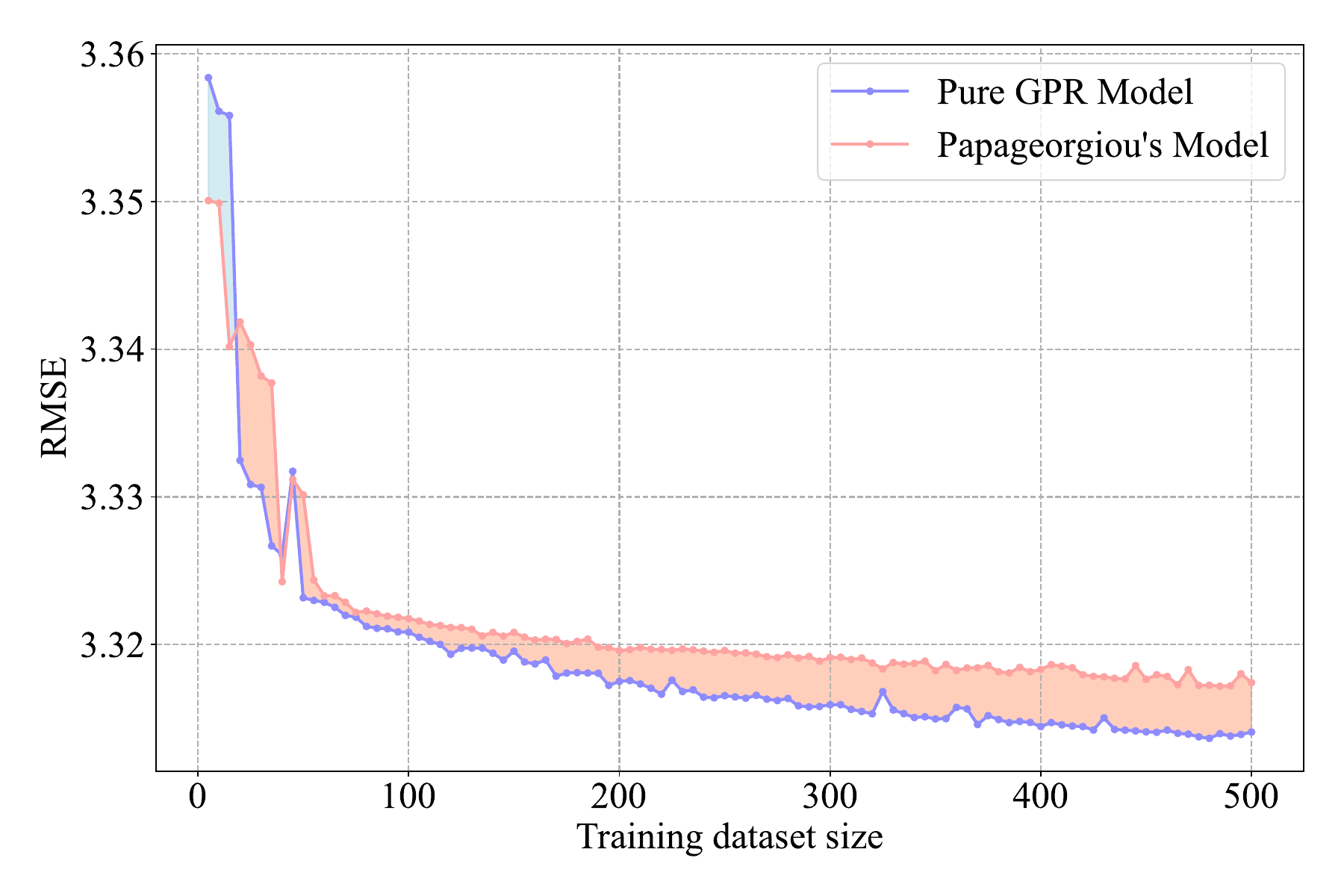}
    \caption{}
    \label{figure:53}
    \end{subfigure}
    \begin{subfigure}{.27\textwidth}
        \centering
        \includegraphics[width=1.0\textwidth]{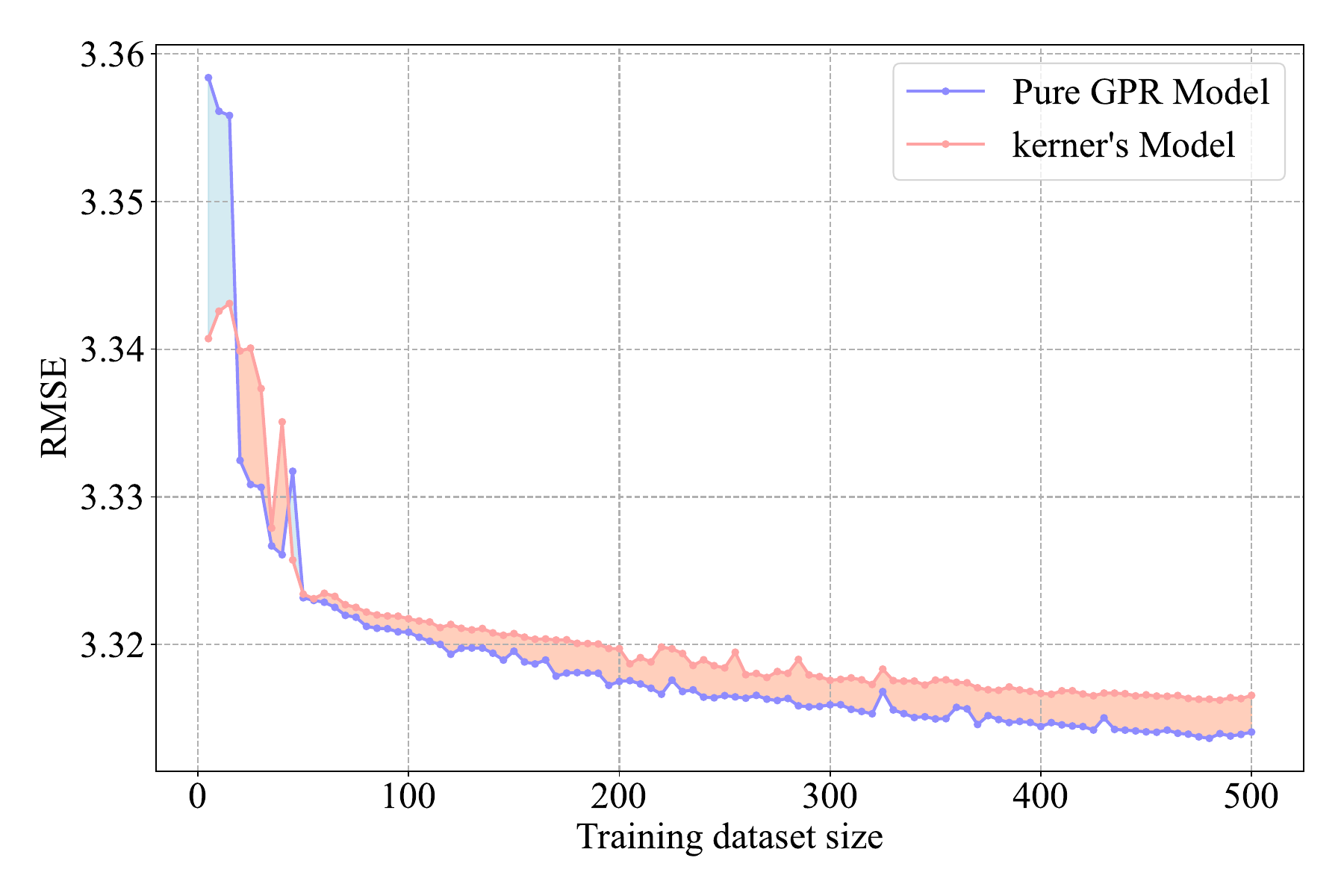}
        \caption{}
        \label{figure:54}
    \end{subfigure}
    
    \begin{subfigure}{.27\textwidth}
        \centering
        \includegraphics[width=1.0\textwidth]{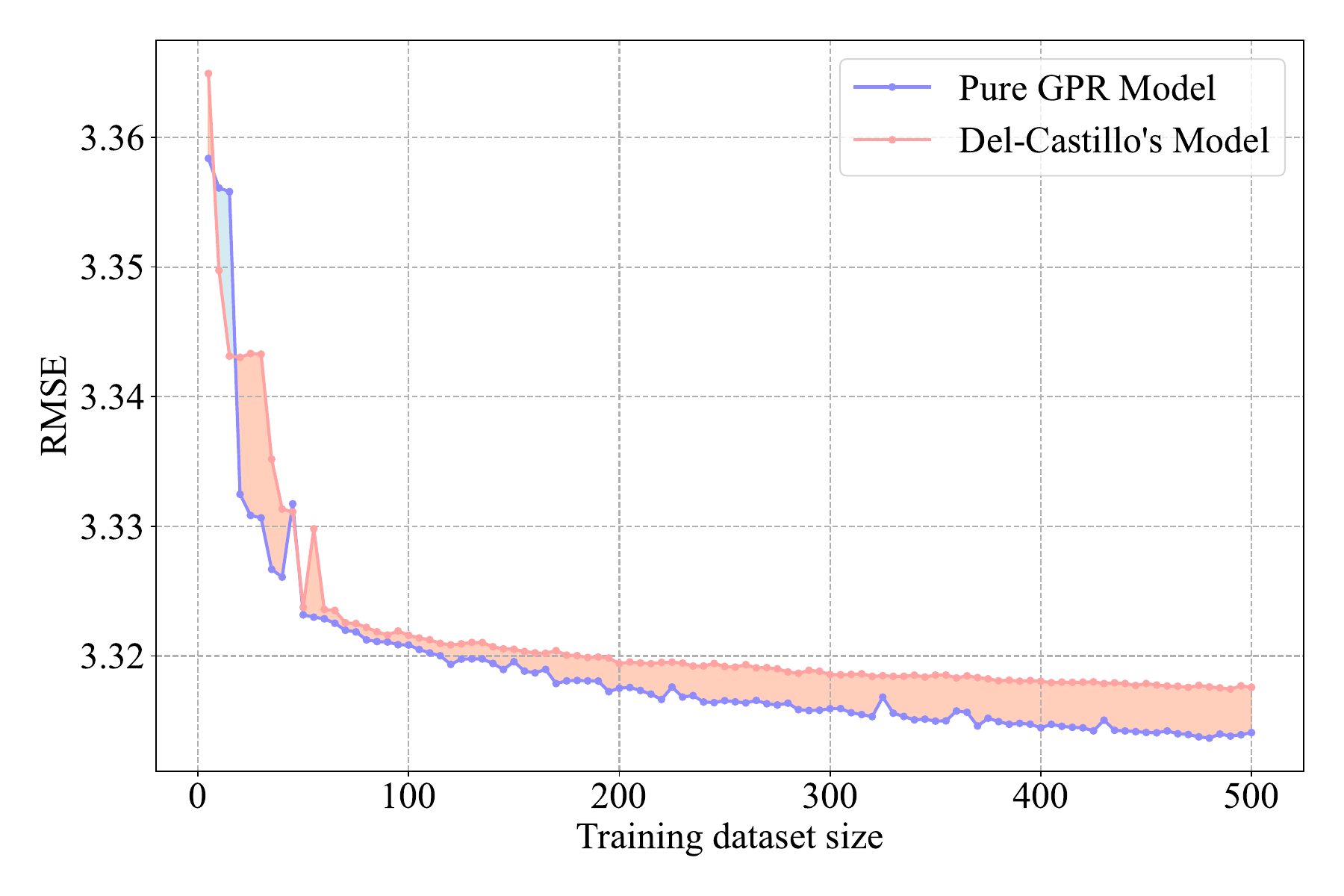}
        \caption{}
        \label{figure:55}
    \end{subfigure}
    \begin{subfigure}{.27\textwidth}
        \centering
        \includegraphics[width=1.0\textwidth]{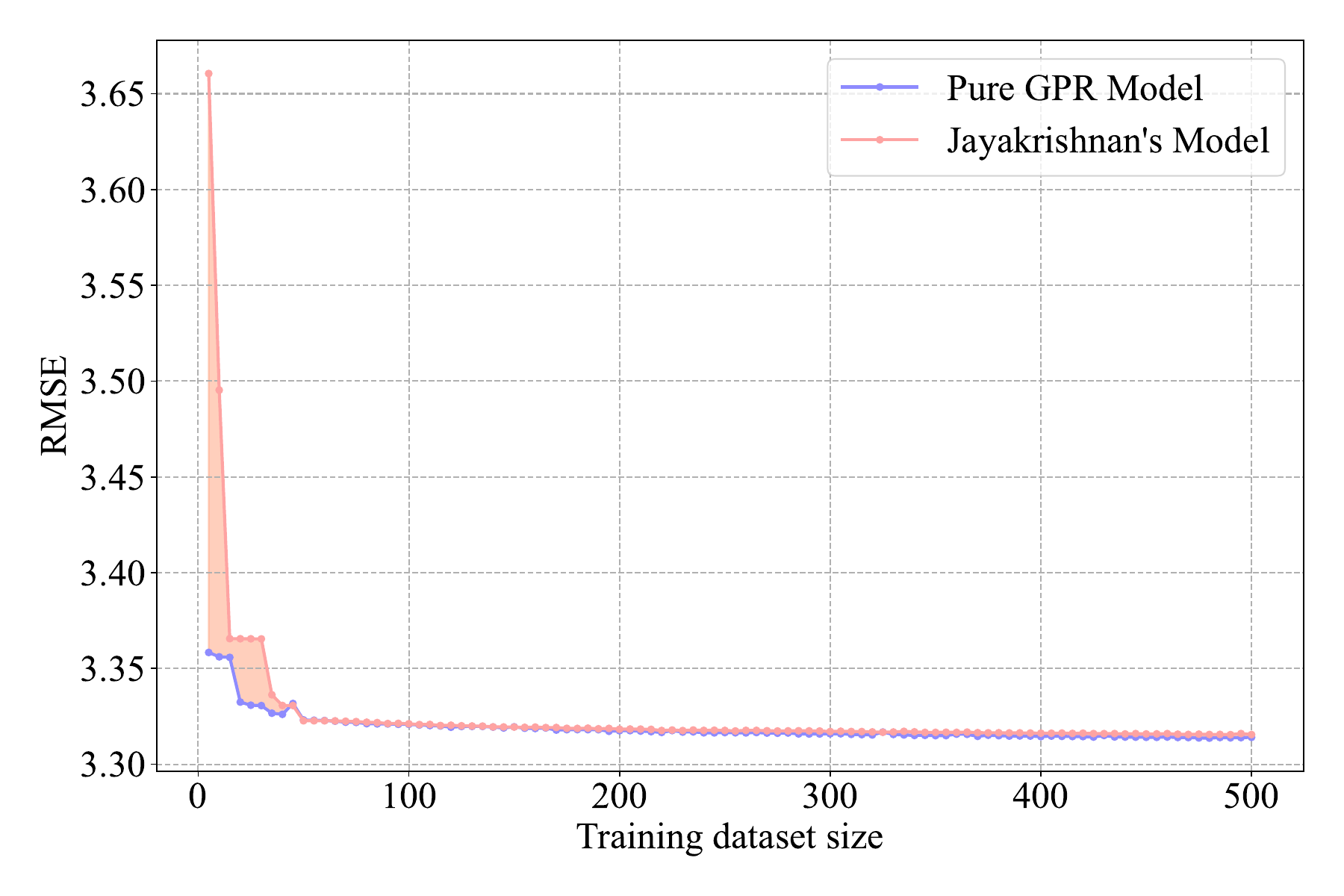}
    \caption{}
    \label{figure:56}
    \end{subfigure}
    
    \hfill
    \caption{RSME comparison between pure GPR and EPGPR model under different training dataset size}
    \label{figure:57}
\end{figure}

\begin{figure}[htbp]
    \centering
    \begin{subfigure}{.27\textwidth}
        \centering
        \includegraphics[width=1.0\textwidth]{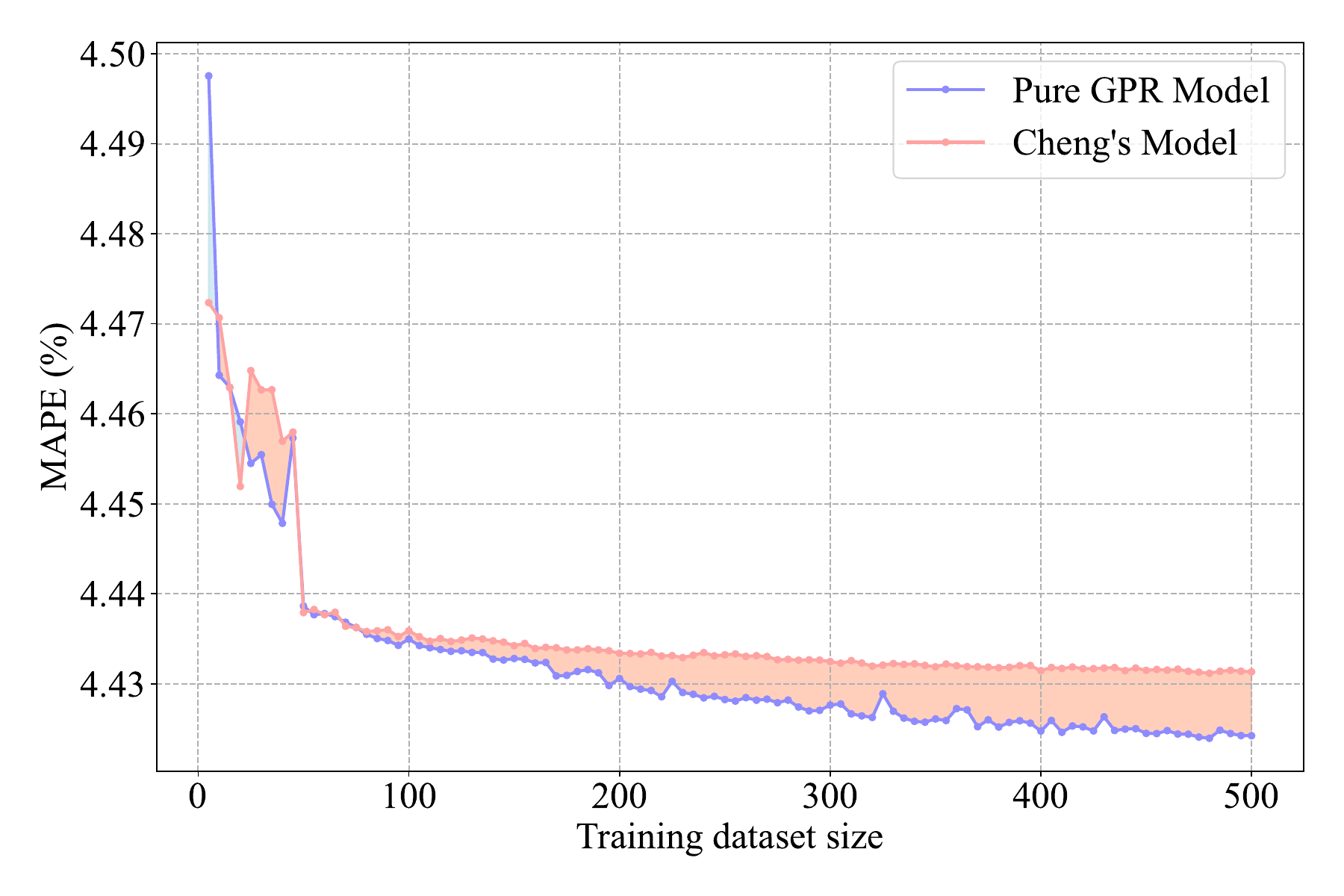}
        \caption{}
        \label{figure:58}
    \end{subfigure}
    \begin{subfigure}{.27\textwidth}
        \centering
        \includegraphics[width=1.0\textwidth]{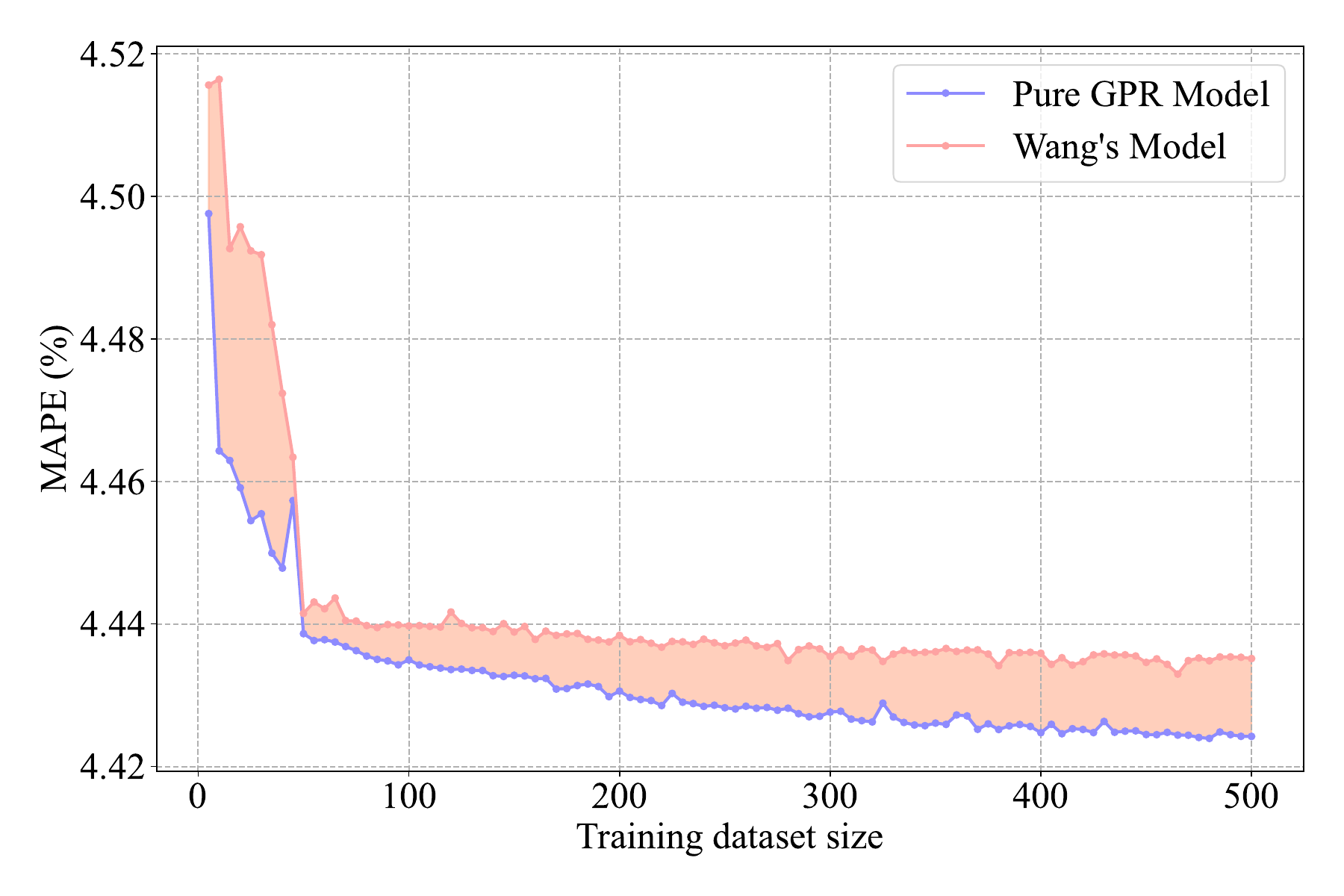}
        \caption{}
        \label{figure:59}
    \end{subfigure}
    \begin{subfigure}{.27\textwidth}
        \centering
        \includegraphics[width=1.0\textwidth]{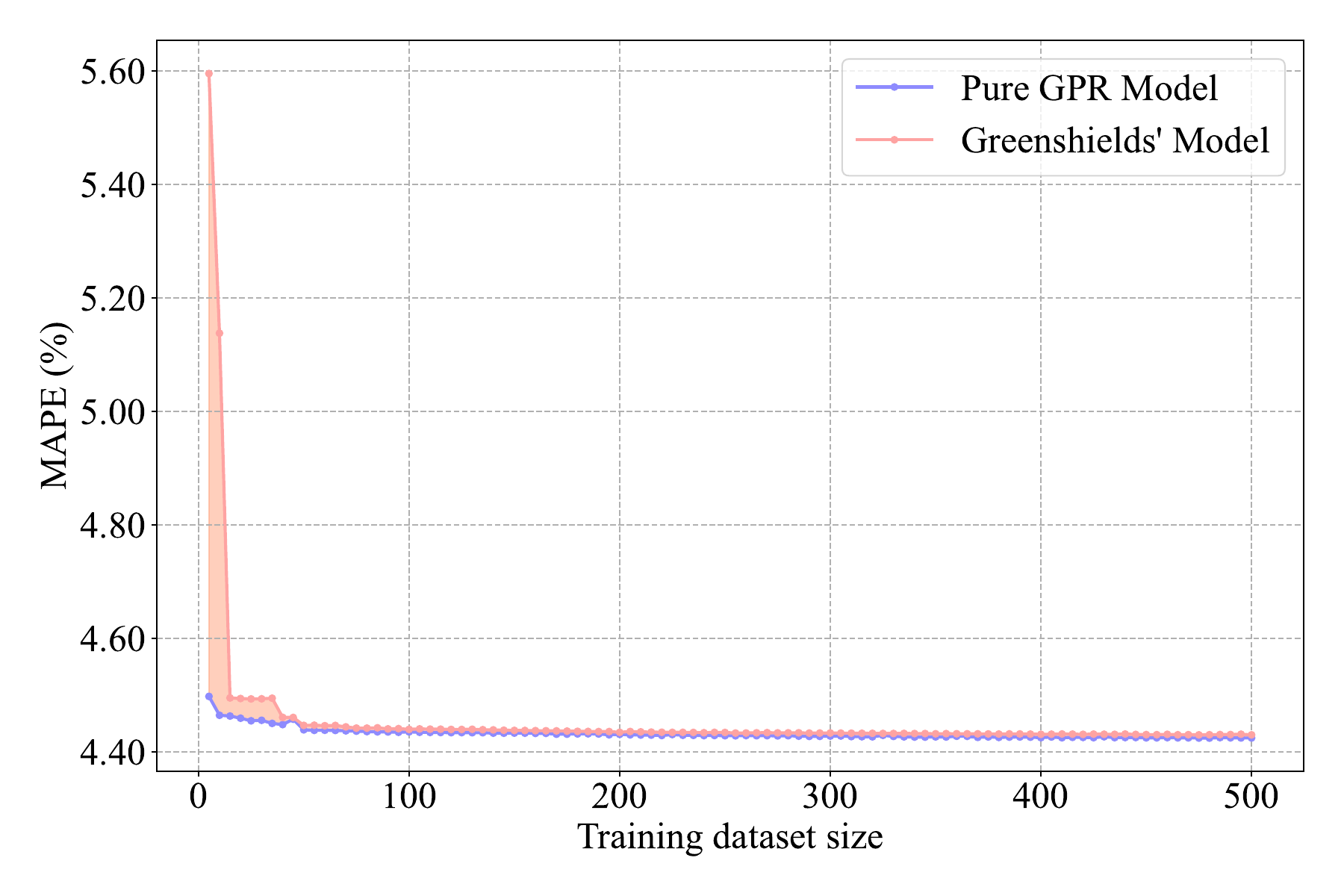}
         \caption{}
        \label{figure:60}
        \label{fig:newell}
    \end{subfigure}

    \begin{subfigure}{.27\textwidth}
        \centering
        \includegraphics[width=1.0\textwidth]{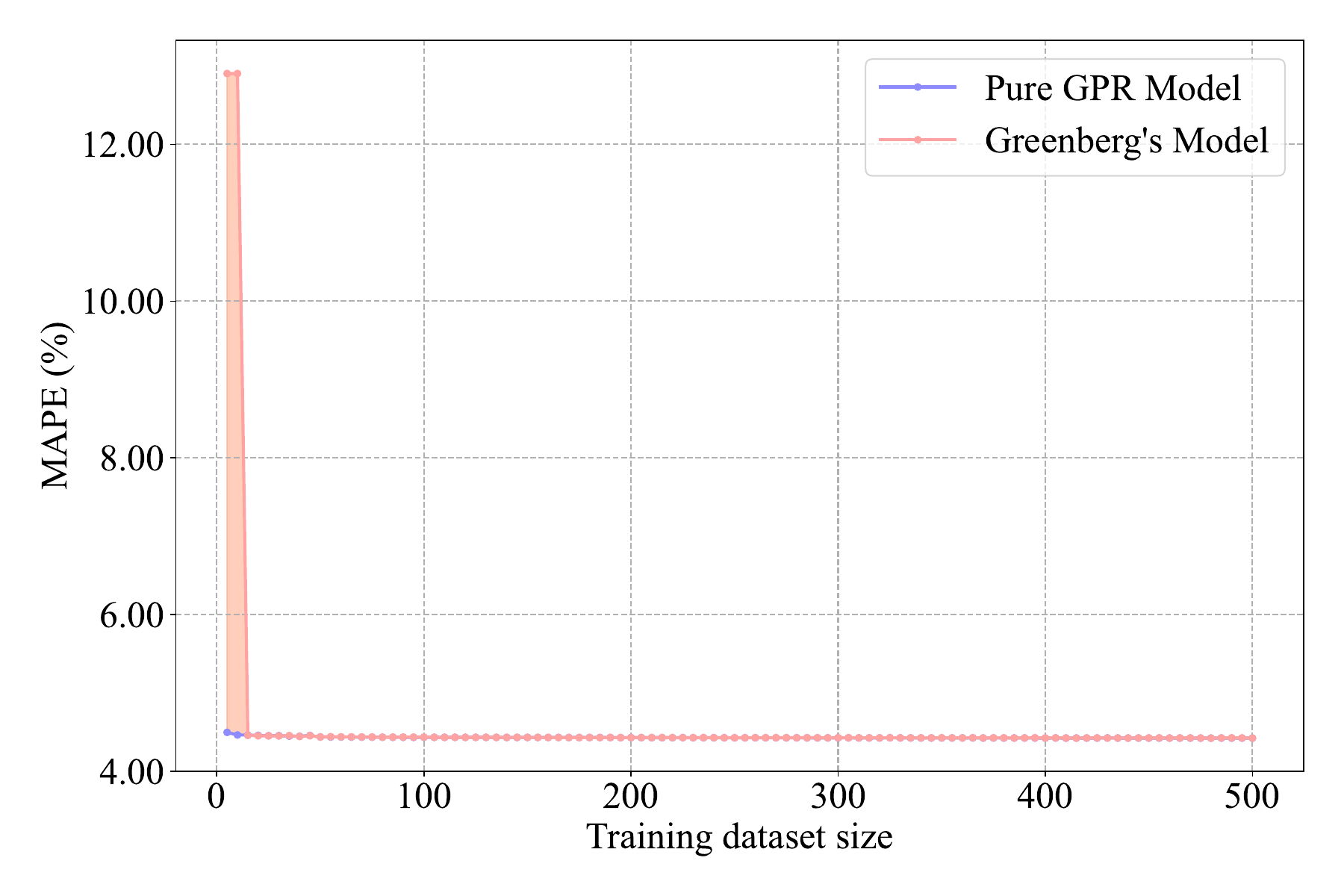}
        \caption{}
        \label{figure:61}
    \end{subfigure}
    \begin{subfigure}{.27\textwidth}
        \centering
        \includegraphics[width=1.0\textwidth]{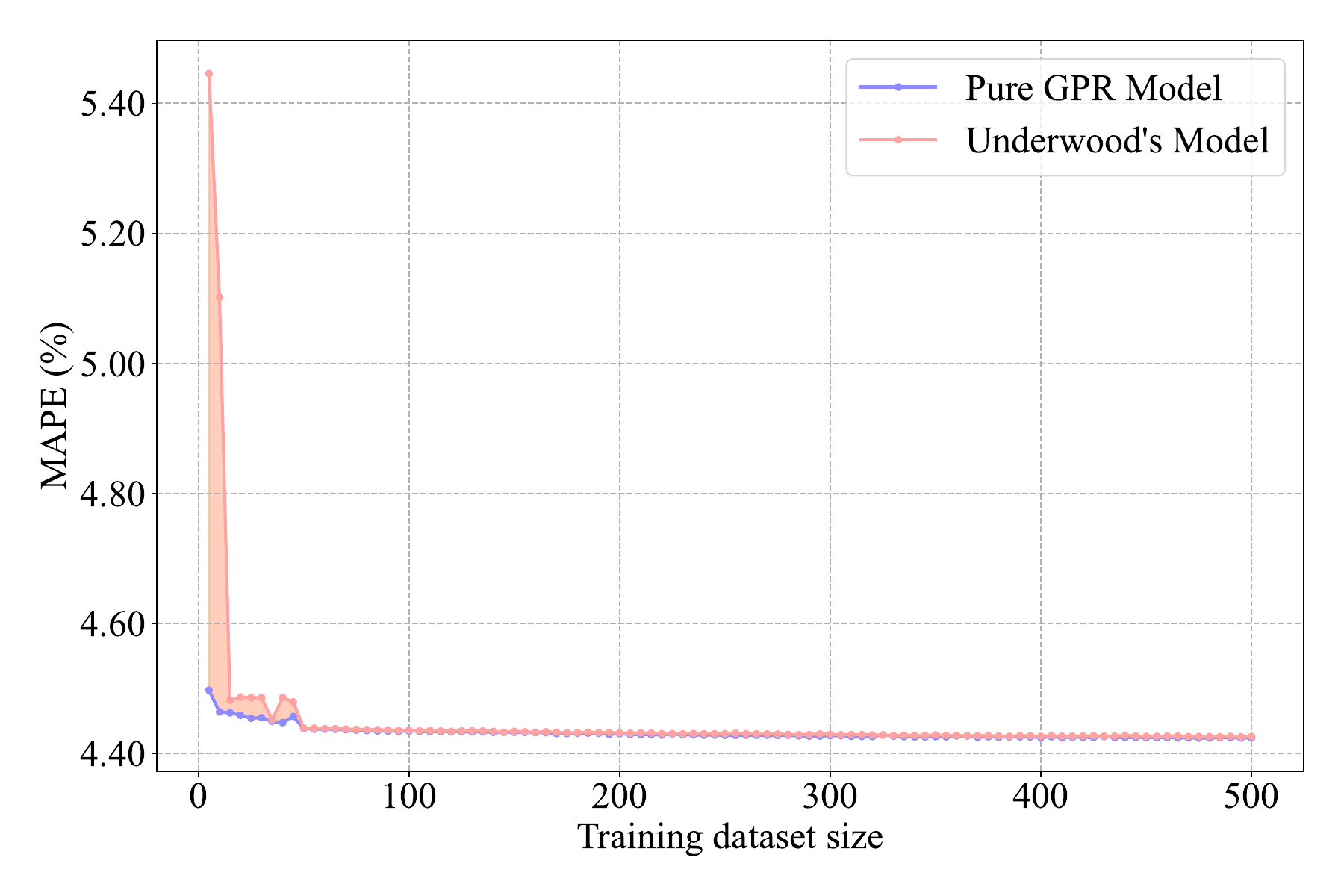}
        \caption{}
        \label{figure:62}
    \end{subfigure}
    \begin{subfigure}{.27\textwidth}
        \centering
        \includegraphics[width=1.0\textwidth]{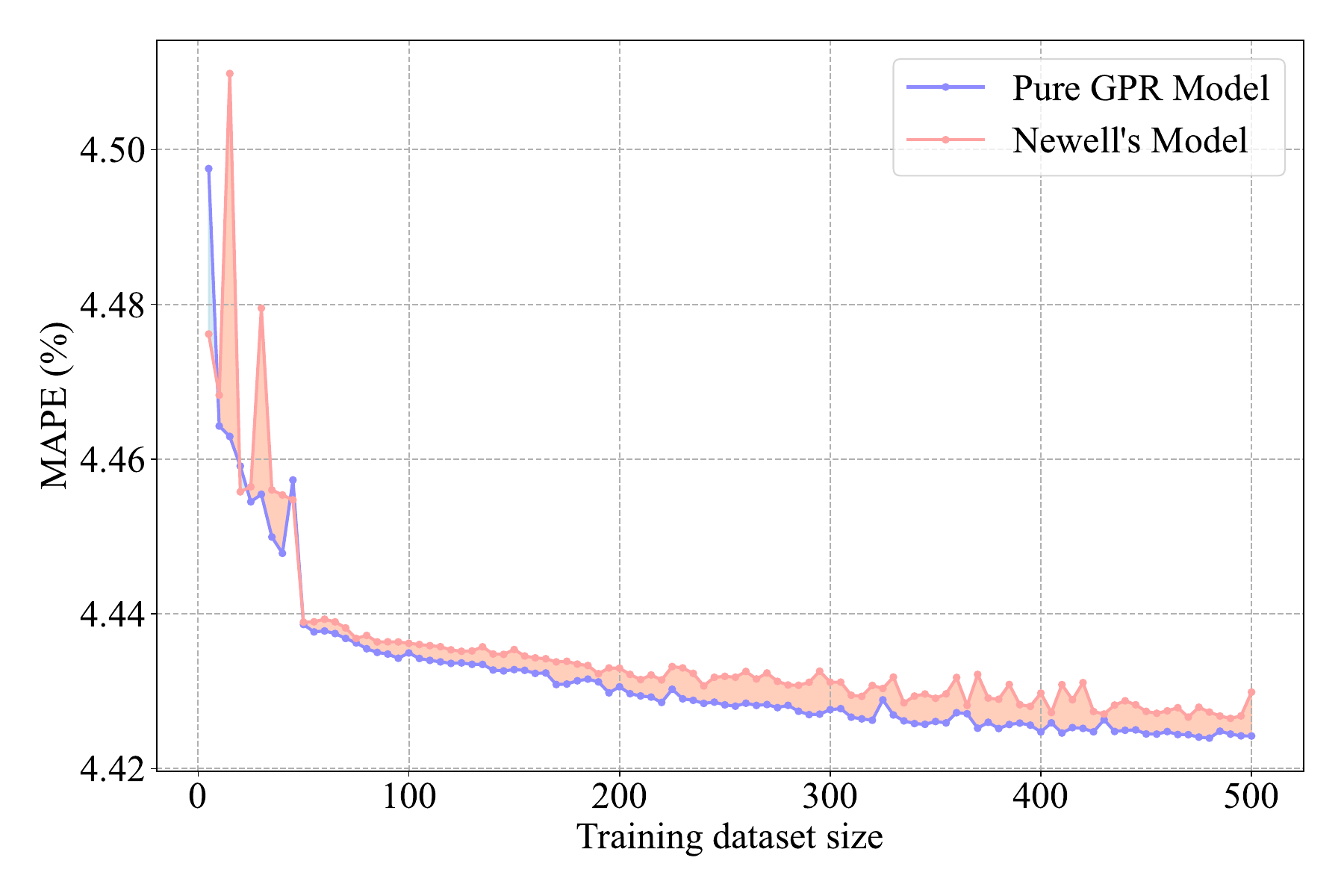}
        \caption{}
        \label{figure:63}
    \end{subfigure}

    \begin{subfigure}{.27\textwidth}
        \centering
        \includegraphics[width=1.0\textwidth]{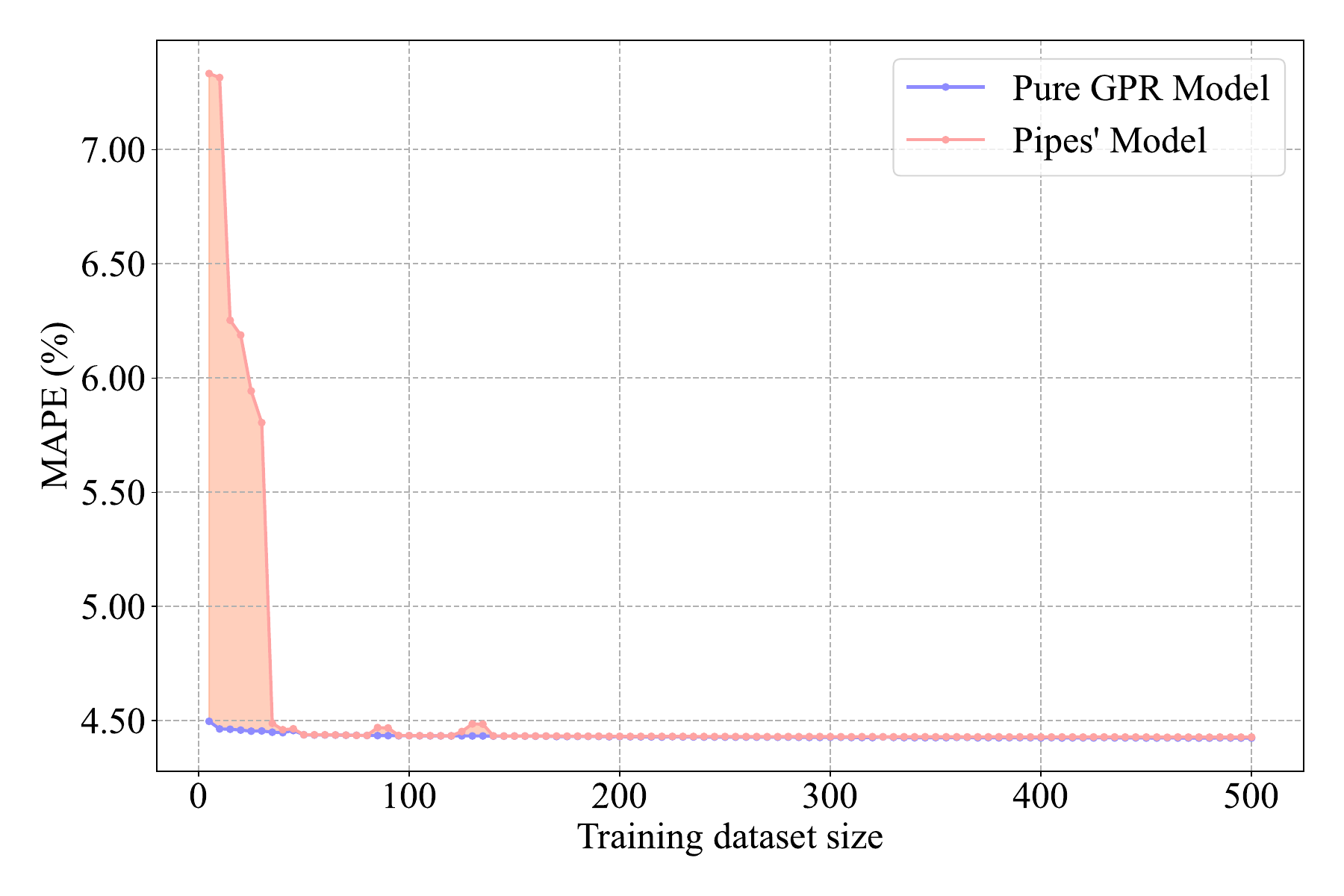}
        \caption{}
        \label{figure:64}
    \end{subfigure}
    \begin{subfigure}{.27\textwidth}
        \centering
        \includegraphics[width=1.0\textwidth]{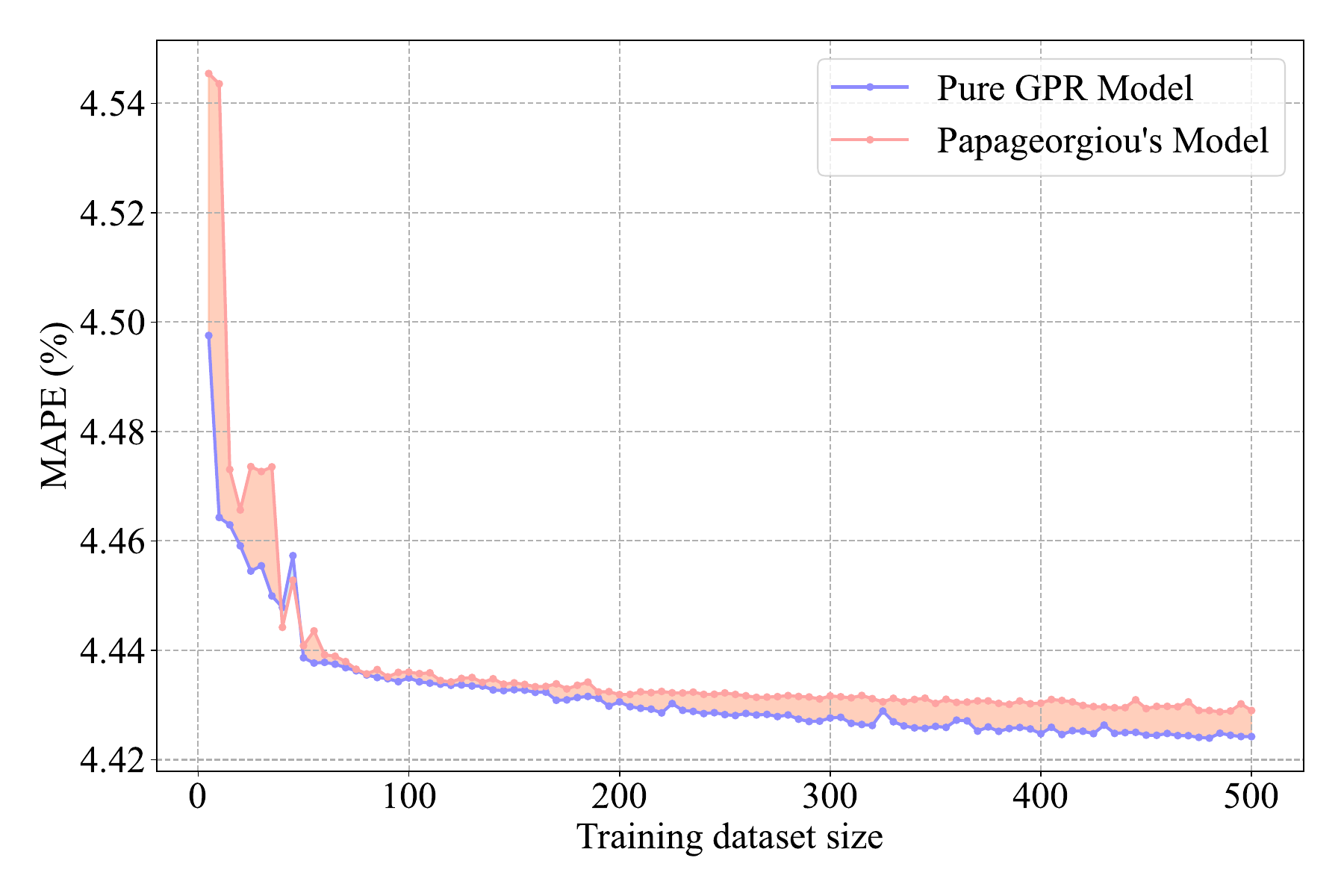}
    \caption{}
    \label{figure:65}
    \end{subfigure}
    \begin{subfigure}{.27\textwidth}
        \centering
        \includegraphics[width=1.0\textwidth]{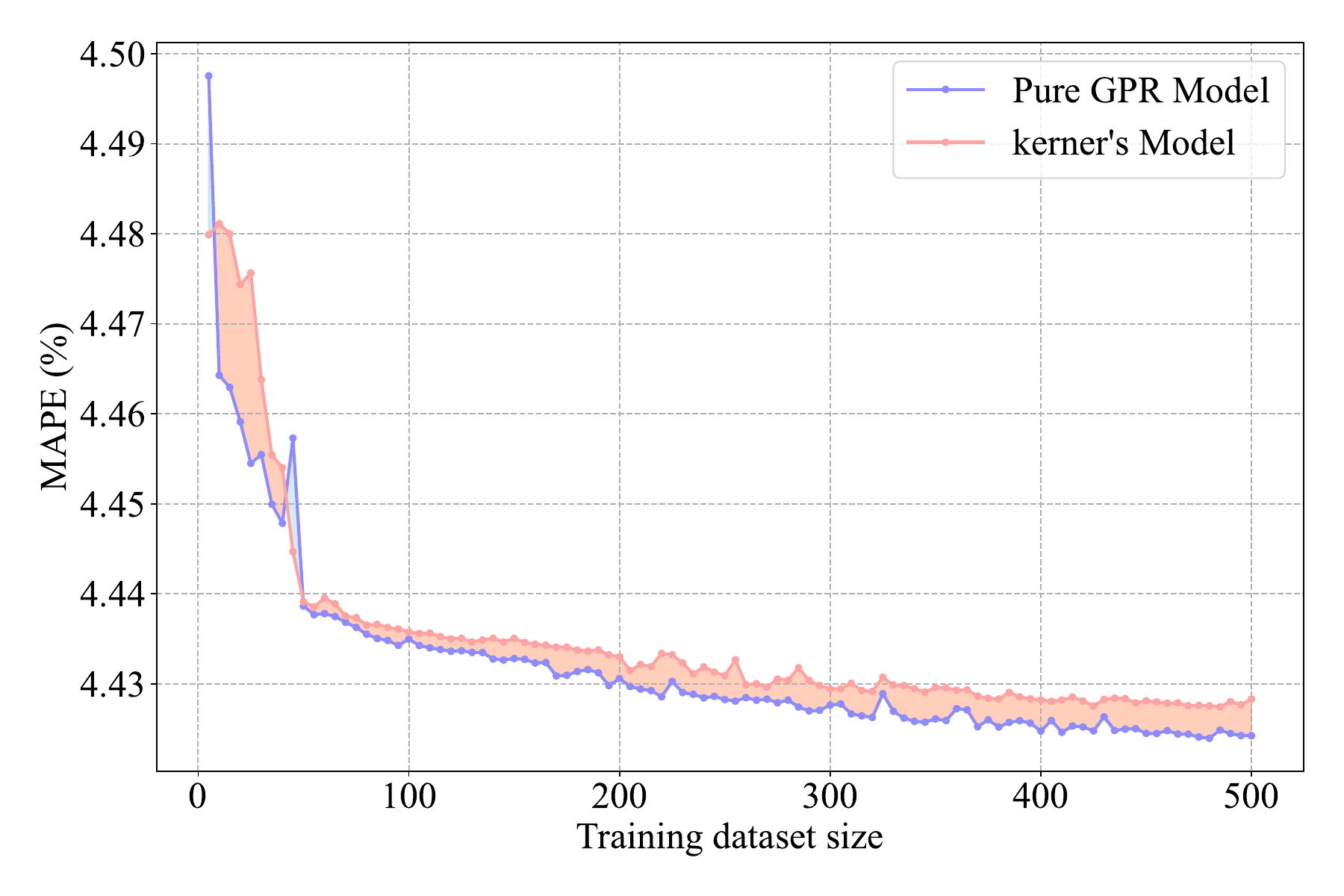}
        \caption{}
        \label{figure:66}
    \end{subfigure}
    
    \begin{subfigure}{.27\textwidth}
        \centering
        \includegraphics[width=1.0\textwidth]{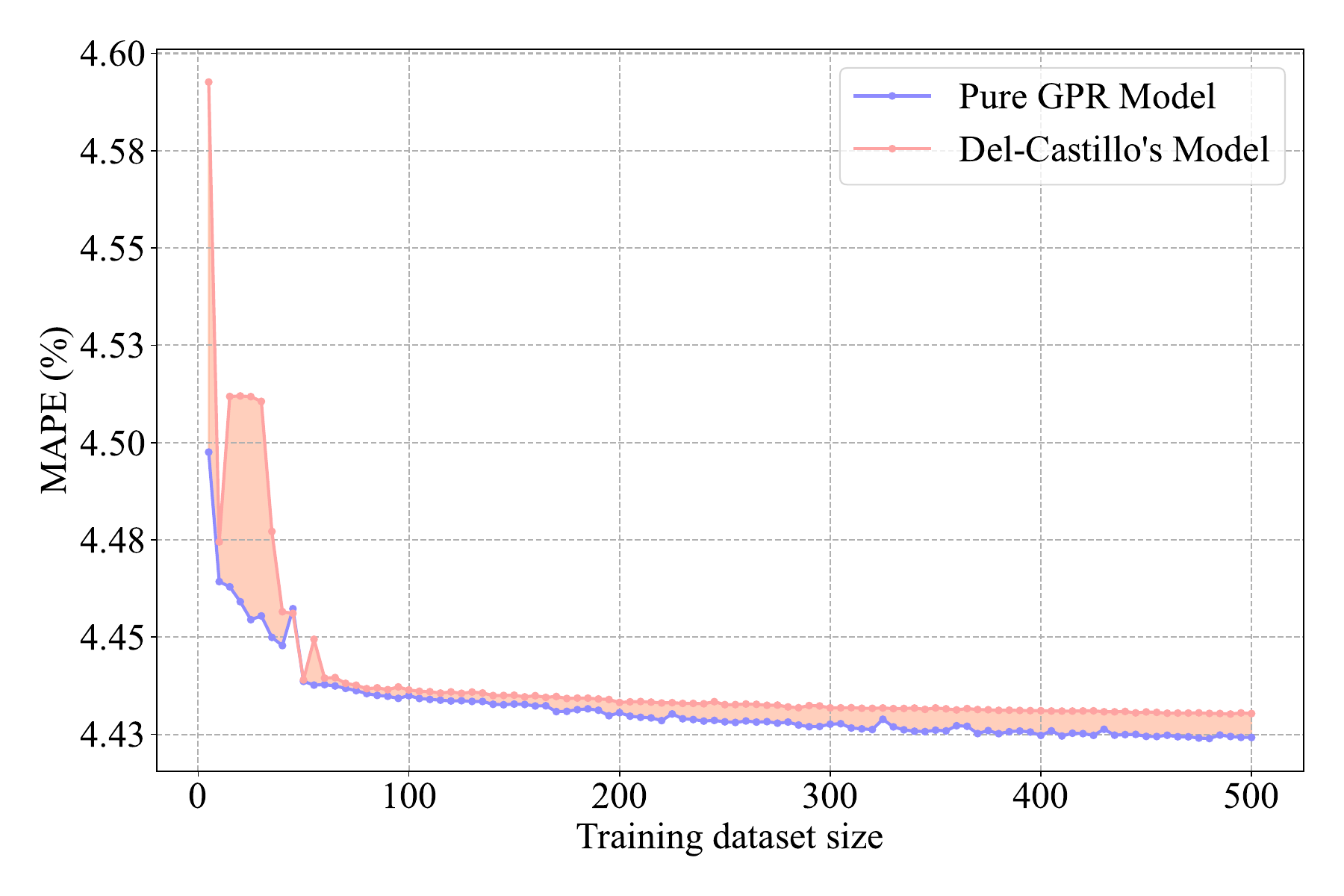}
        \caption{}
        \label{figure:67}
    \end{subfigure}
    \begin{subfigure}{.27\textwidth}
        \centering
        \includegraphics[width=1.0\textwidth]{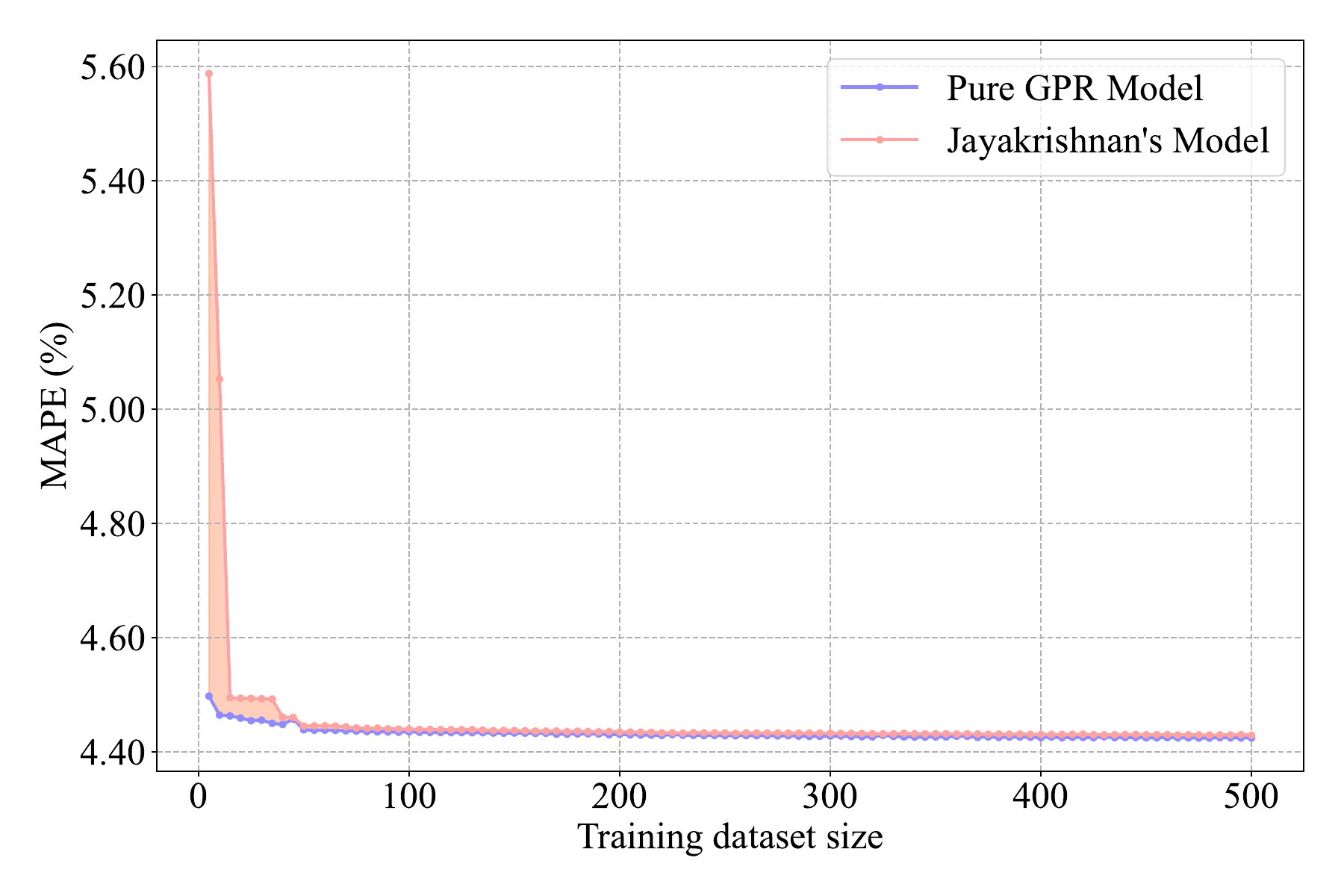}
    \caption{}
    \label{figure:68}
    \end{subfigure}
    
    \hfill
    \caption{MAPE comparison between pure GPR and EPGPR model under different training dataset size}
    \label{figure:69}
\end{figure}
\section{Conclusion} \label{5}
\par This research employs a sparse Gaussian process regression framework to analyze the stochastic fundamental diagram model. SGPR, a non-parametric regression technique, counters its intrinsic high computational requirements with a sparse approximation strategy, thereby enhancing fitting accuracy beyond that of all empirical models discussed herein. The Gaussian process distinguishes itself from other non-parametric regressions by its inherent stochastic nature, which is well-suited for modeling the speed distribution at a given traffic density.
\par The integration of the empirical model into the prior demonstrates that empirical models, derived and calibrated from comprehensive datasets, improve fitting performance over fully data-driven non-parametric regression models only when applied to smaller training datasets. It is noteworthy that the precision of an empirical model does not invariably translate to superior performance in fitting empirical data. The findings suggest that a moderately sized training dataset of relatively clean data suffices for the pure GPR model to represent the stochastic relationship between speed and density accurately.
\par Looking ahead, future investigations could incorporate methodologies from information theory and statistics to enhance the selection of the input-inducing variables that convey greater informational content about the data. Instead of solely encoding empirical knowledge into the mean function, as explored in this study, it may prove beneficial to embed such knowledge within the kernel function, potentially improving the precision and robustness of non-parametric regression models.
\section{CRediT} \label{6}
\textbf{Yuan-Zheng Lei}: Conceptualization, Methodology, Writing - original draft. \textbf{Yaobang Gong}: Conceptualization, Writing - original draft. \textbf{Xianfeng Terry Yang}: Conceptualization, Methodology and Supervision.

\section{Model Source Code} \label{7}
The source code of the proposed model is available for readers to download via the following link: \url{https://github.com/YuanzhengLei/Unraveling-stochastic-fundamental-diagrams-considering-empirical-knowledge}.

\section{Acknowledgement} \label{8}
This research is supported by the award "CAREER: Physics Regularized Machine Learning Theory: Modeling Stochastic Traffic Flow Patterns for Smart Mobility Systems (\# 2234289)" which is funded by the National Science Foundation. All authors would like to thank Professor Daiheng Ni of the University of Massachusetts Amherst for sharing the dataset kindly. 

\newpage
\section{Appendix} \label{9}
\begin{table}[hthp]
  \centering
  \caption{Parameters of single-regime traffic flow model after calibration\textsuperscript{*}}
  \label{Table:7}
  \begin{tabular}{>{\centering\arraybackslash}p{0.35\linewidth} >{\centering\arraybackslash}p{0.65\linewidth}}
    \thicktoprule
    \textbf{Model} &  \textbf{Parameters}\\
    \midrule
    \cite{greenshields1935study} & \sout{$v_{f} = 18.22\ \rho_{j} = -111232.99$} \\
    \cite{greenberg1959analysis} & \sout{$v_{critical} = -29357.73\ \rho_{j} = -29083.32$} \\
    \cite{newell1961nonlinear} & $v_{f} = 44.66\ \rho_{j} = 107.05\ \lambda = 3063.25$\\
            \cite{underwood1961speed} & \sout{$v_{f} = 49.35\ \rho_{critical} = -29083.32$} \\
        \cite{drake1965statistical} & $v_{f} = 49.35\ \rho_{critical} = 214.11$ \\
        \cite{papageorgiou1989macroscopic} & $v_{f} = 47.92\ \rho_{j} = 102.13\ \alpha = 1.09$\\
        \cite{kerner1994structure} & $v_{f} = 41.34\ \rho_{critical} = 391.92$ \\
       \cite{del1995functional-part-1}   & $v_{f} = 44.66\ \rho_{j} = 485.62\ v_{j} = 6.31$ \\
       \cite{jayakrishnan1995dynamic}  & $v_{f} = 33.72\ v_{min} = 37.35\ \rho_{j} = 27.30$ \\
       \cite{ardekani2008modified} & $v_{critical} = 46.64\ \rho_{j} = 52.24\ \rho_{min} = 0.38$ \\
       \cite{macnicholas2011simple} & \sout{$v_{f} = 45.22\ \rho_{j} = 10544.64\ n = 1.95\ m = 12412.20$} \\
       \cite{pipes1966car} & $v_{f} = 46.89\ \rho_{j} = 132.87\ n = 0.69$ \\
       \cite{wang2011logistic} & $v_{f} = 42.72\ v_{critical} = 6.64\ \rho_{critical} = 40.97\ \theta_{1} = 3.50\ \theta_{2} = 0.08$ \\
       \cite{cheng2021s} & $v_{f} = 44.42\ \rho_{critical} = 80.09 \ m = 2.21$ \\
    \thickbottomrule
  \end{tabular}
  \footnotesize\textsuperscript{*} Parameters are calibrated based on I-80 detector data (Station id:1).
\end{table}
\begin{table}[htbp]
\centering
\caption{Comparison of the result of different EPNR models with Exponential kernel function (I-80)}
\label{Table:8}
\begin{tabularx}{\textwidth}{Z*{8}{Y}}
\thicktoprule
\textbf{EPNR Model} & \multicolumn{4}{c}{\textbf{Speed RMSE (mph)}} & \multicolumn{4}{c}{\textbf{Speed MAPE}} \\ \cline{2-9}
& $\mathcal{RS}$ & $\mathcal{SS}$ & $\mathcal{CS}$ & $\mathcal{WRS}$ & $\mathcal{RS}$ & $\mathcal{SS}$ & $\mathcal{CS}$ & $\mathcal{WRS}$ \\ \midrule
Pure GP  & 5.27  & 5.27  & 5.26 & 5.27 & 12.74\% & 12.74\% & 12.73\% & 12.74\% \\ 
Cheng's  & 5.27  & 5.27  & 5.27 & 5.27 & 12.76\% & 12.76\% & 12.75\% & 12.75\% \\ 
Wang's & 5.27  & 5.27  & 5.27 & 5.27 & 12.75\% & 12.75\% & 12.75\% & 12.75\% \\
Newell's  & 5.27  & 5.27  & 5.27 & 5.27 & 12.76\% & 12.76\% & 12.74\% & 12.75\% \\ 
Papageorgiou's & 5.27  & 5.26  & 5.27 & 5.27 & 12.76\% & 12.76\% & 12.75\% & 12.75\% \\ 
Kerner's & 5.27  & 5.27  & 5.27 & 5.27 & 12.75\% & 12.75\% & 12.75\% & 12.75\% \\ 
Del-Castillo's  & 5.27  & 5.27  & 5.27 & 5.27 & 12.76\% & 12.75\% & 12.75\% & 12.75\% \\ 
Jayakrishnan's & 5.27  & 5.27  & 5.27 & 5.27 & 12.76\% & 12.76\% & 12.75\% & 12.75\% \\ 
\thickbottomrule
\end{tabularx}
\end{table}
\begin{table}[htbp]
\centering
\caption{Comparison of the result of percentages of points falling within the 95\% confidence interval}
\label{Table:9}
\begin{tabularx}{\textwidth}{Z*{4}{Y}}
\thicktoprule
\textbf{EPNR Model} & \multicolumn{4}{c}{\textbf{PWCI}} \\ \cline{2-5}
& $\mathcal{RS}$ & $\mathcal{SS}$ & $\mathcal{CS}$ & $\mathcal{WRS}$ \\ \midrule
Pure GP & 95.79\% & 95.78\% & 95.77\% & 95.80\% \\ 
Cheng's & 95.80\% & 95.80\% & 95.78\% & 95.80\% \\ 
Wang's & 95.79\% & 95.79\% & 95.79\% & 95.79\% \\
Newell's & 95.79\% & 95.79\% & 95.78\% & 95.80\% \\ 
Papageorgiou's  & 95.78\% & 86.55\% & 95.78\% & 95.79\% \\ 
Kerner's  & 95.79\% & 95.79\% & 95.80\% & 95.79\% \\ 
Del-Castillo's  & 95.79\% & 95.80\% & 95.79\% & 95.80\% \\ 
Jayakrishnan's  & 95.78\% & 95.78\% & 95.77\% & 95.78\% \\ 
\thickbottomrule
\end{tabularx}
\end{table}

\begin{lemma}
   \begin{equation}
        q(\mathbf{f}_{*}) \approx p(\mathbf{f}_{*}|\mathbf{y}), q(\mathbf{f}_{*}) =  \int p(\mathbf{f}_{*}|\mathbf{f}_{m})\phi(\mathbf{f}_{m})d\mathbf{f}_{m} \label{eq:39}
   \end{equation} \label{lemma:1}
\end{lemma}
\begin{proof}
\begin{align*}
    p(\mathbf{f}_{*}|\mathbf{y}) &= \int p(\mathbf{f}_{*}|\mathbf{f})p(\mathbf{f}|\mathbf{y})d\mathbf{f} \\
    &=  \int p(\mathbf{f}_{*}|\mathbf{f}_{m},\mathbf{f})p(\mathbf{f}|\mathbf{f}_{m},\mathbf{y})p(\mathbf{f}_{m}|\mathbf{y})d\mathbf{f}d\mathbf{f}_{m} 
\end{align*}
Since $\mathbf{f}_{m}$ is 
 independent from $\mathbf{f}$ and $\mathbf{y} = \mathbf{f} + \sigma$,so $\mathbf{f}_{m}$ is also
 independent from $\mathbf{y}$,thus, $p(\mathbf{f}_{*}|\mathbf{f}_{m}) = p(\mathbf{f}_{*}|\mathbf{f}_{m},\mathbf{f}),p(\mathbf{f}|\mathbf{f}_{m}) = p(\mathbf{f}|\mathbf{f}_{m},\mathbf{y})$, thus, we can have:
\begin{align*}
     p(\mathbf{f}_{*}|\mathbf{y}) &= \int p(\mathbf{f}_{*}|\mathbf{f}_{m})p(\mathbf{f}|\mathbf{f}_{m})p(\mathbf{f}_{m}|\mathbf{y})d\mathbf{f}d\mathbf{f}_{m}
\end{align*}
And since $\int p(\mathbf{f}|\mathbf{f}_{m}) d\mathbf{f} = 1, \phi(\mathbf{f}_{m}) = p(\mathbf{f}_{m}|\mathbf{y})$, thus:
\begin{align*}
     p(\mathbf{f}_{*}|\mathbf{y}) &= \int p(\mathbf{f}_{*}|\mathbf{f}_{m})p(\mathbf{f}_{m}|\mathbf{y})d\mathbf{f}_{m} \\
     &= \int p(\mathbf{f}_{*}|\mathbf{f}_{m})\phi(\mathbf{f}_{m})d\mathbf{f}_{m} \\
     &= q(\mathbf{f}_{*})
\end{align*}
\end{proof}
\begin{lemma}
Minimize $KL(q(\mathbf{f},\mathbf{f}_{m})\|p(\mathbf{f},\mathbf{f}_{m}|\mathbf{y}))$ is equivalent to maximize:
   \begin{equation}
   F_{V}(\mathbf{x}_{m},\phi) = \int p(\mathbf{f}|\mathbf{f}_{m})\phi(\mathbf{f}_{m})\log\frac{p(\mathbf{f}|\mathbf{y})p(\mathbf{f}_{m})}{\phi(\mathbf{f}_{m})}d\mathbf{f}d\mathbf{f}_{m}  \label{eq:40}
   \end{equation} \label{lemma:2}
\end{lemma}
\begin{proof}
\begin{align*}
    KL(q(\mathbf{f},\mathbf{f}_{m})\|p(\mathbf{f},\mathbf{f}_{m}|\mathbf{y}))) &= \mathbb{E}_{q_{\phi(\mathbf{f},\mathbf{f}_{m})}}[\log\frac{q_{\phi(\mathbf{f},\mathbf{f}_{m})}}{p(\mathbf{f},\mathbf{f}_{m}|\mathbf{y})}] \\
    &= \mathbb{E}_{q_{\phi(\mathbf{f},\mathbf{f}_{m})}}[\log\frac{q_{\phi(\mathbf{f},\mathbf{f}_{m})}p(\mathbf{y})}{p(\mathbf{f},\mathbf{f}_{m}|\mathbf{y})p(\mathbf{y}}] \\
    &= \mathbb{E}_{q_{\phi(\mathbf{f},\mathbf{f}_{m})}}[\log\frac{q_{\phi(\mathbf{f},\mathbf{f}_{m})}p(\mathbf{y})}{p(\mathbf{f},\mathbf{f}_{m}, \mathbf{y})}] \\
    &= \log p(\mathbf{y}) + \mathbb{E}_{q_{\phi(\mathbf{f},\mathbf{f}_{m})}}[\log\frac{q_{\phi(\mathbf{f},\mathbf{f}_{m})}}{p(\mathbf{f},\mathbf{f}_{m}, \mathbf{y})}] \\
    &= \log p(\mathbf{y}) - \mathbb{E}_{q_{\phi(\mathbf{f},\mathbf{f}_{m})}}[\log\frac{p(\mathbf{f},\mathbf{f}_{m}, \mathbf{y})}{q_{\phi(\mathbf{f},\mathbf{f}_{m})}}] \\
    &= \log p(\mathbf{y}) -\int q(\mathbf{f}, \mathbf{f}_{m}) \log \left( \frac{p(\mathbf{f}, \mathbf{f}_{m}, \mathbf{y})}{q(\mathbf{f}, \mathbf{f}_{m})} \right) d\mathbf{f} \, d\mathbf{f}_{m} \\
    &= \log p(\mathbf{y}) -\int p(\mathbf{f}| \mathbf{f}_{m})p(\mathbf{f}_{m}) \log \left( \frac{p(\mathbf{y}|\mathbf{f})p(\mathbf{f}|\mathbf{f}_{m})p(\mathbf{f}_{m})}{q(\mathbf{f}, \mathbf{f}_{m})} \right) d\mathbf{f} \, d\mathbf{f}_{m}
\end{align*}
From equation \ref{eq:21} and \ref{eq:22}, $q(\mathbf{f}_{*})$ is an approximation to $p(\mathbf{f}_{m}|\mathbf{y})$, i.e. $q(\mathbf{f}_{*}) = p(\mathbf{f}_{*}|\mathbf{y})$, also, $\phi(\mathbf{f}_{m}) = p(\mathbf{f}_{m}|\mathbf{y})$, therefore, we know that $q(\mathbf{f}, \mathbf{f}_{m}) = p(\mathbf{f}, \mathbf{f}_{m}|\mathbf{y}) = p(\mathbf{f}| \mathbf{f}_{m})p(\mathbf{f}_{m}|\mathbf{y}) =  p(\mathbf{f}| \mathbf{f}_{m})\phi(\mathbf{f}_{m})$, therefore:
\begin{align*}
    KL(q(\mathbf{f},\mathbf{f}_{m})\|p(\mathbf{f},\mathbf{f}_{m}|\mathbf{y}))) 
    &= \log p(\mathbf{y}) -\int p(\mathbf{f}| \mathbf{f}_{m})\phi(\mathbf{f}_{m}) \log \left( \frac{p(\mathbf{y}|\mathbf{f})p(\mathbf{f}|\mathbf{f}_{m})p(\mathbf{f}_{m})}{p(\mathbf{f}| \mathbf{f}_{m})\phi(\mathbf{f}_{m})} \right) d\mathbf{f} \, d\mathbf{f}_{m} \\
    &= \log p(\mathbf{y}) - \int p(\mathbf{f}|\mathbf{f}_{m})\phi(\mathbf{f}_{m})\log\frac{p(\mathbf{f}|\mathbf{y})p(\mathbf{f}_{m})}{\phi(\mathbf{f}_{m})}d\mathbf{f}d\mathbf{f}_{m}
\end{align*}
And since KL divergence $\geq 0$, the lemma proved.
\end{proof}
\begin{lemma}
The variation lower bound after taking the optimal choice of the variation distribution $\phi_{*}$ is:
   \begin{equation}
   F_{V}(\mathbf{x}_{m},\phi_{*}) = \log [\mathcal{N}(\mathbf{y}|\boldsymbol{0},\sigma^{2}\mathbf{I}+\mathbf{Q}_{nn})] -\frac{1}{2\sigma^{2}}Tr(\mathbf{K}_{nn} - \mathbf{Q}_{nn}) \label{eq:41}
   \end{equation} \label{lemma:3}
\end{lemma}
where $\mathbf{Q}_{nn} = \mathbf{K}_{nm}\mathbf{K}_{mm}^{-1}\mathbf{K}_{mn}$.
\begin{proof}
\begin{align*}
F_{V}(\mathbf{x}_{m},\phi) &= \int p(\mathbf{f}|\mathbf{f}_{m})\phi(\mathbf{f}_{m})\log\frac{p(\mathbf{f}|\mathbf{y})p(\mathbf{f}_{m})}{\phi(\mathbf{f}_{m})}d\mathbf{f}d\mathbf{f}_{m} \\
&= \int \phi(\mathbf{f}_{m})\left\{\int p(\mathbf{f}|\mathbf{f}_{m})\log p(\mathbf{y}|\mathbf{f})d\mathbf{f} + \log \frac{p(\mathbf{f}_{m})}{\phi(\mathbf{f}_{m})} \right\}d\mathbf{f}_{m}
\end{align*}
The integral in the brace has two parts, only the first part involving $\mathbf{f}$, which can be computed as:
\begin{equation}
    \log G(\mathbf{f}_{m},\mathbf{y}) = \int p(\mathbf{f}|\mathbf{f}_{m})\log p(\mathbf{y}|\mathbf{f})d\mathbf{f} \label{eq:42}
\end{equation}
Since $\mathbf{y} = \mathbf{f} + \boldsymbol{\epsilon}$, where $\boldsymbol{\epsilon} \thicksim \mathcal{N}(\boldsymbol{0},\sigma^{2}\mathbf{I})$, therefore,$p(\mathbf{y}|\mathbf{f}) \thicksim \mathcal{N}(\mathbf{f}, \sigma^{2}\mathbf{I})$. $\Rightarrow$
\begin{align*}  
   \log G(\mathbf{f}_{m},\mathbf{y}) &= \int p(\mathbf{f}|\mathbf{f}_{m})\left\{-\frac{n}{2}\log(2\pi \sigma^{2}) - \frac{1}{2\sigma^{2}}(\mathbf{y} - \mathbf{f})^{T}(\mathbf{y} - \mathbf{f})\right\}d\mathbf{f}  \\
   &= \int p(\mathbf{f}|\mathbf{f}_{m})\left\{-\frac{n}{2}\log(2\pi \sigma^{2}) - \frac{1}{2\sigma^{2}}(\mathbf{y}^{T}\mathbf{y} - 2\mathbf{y}^{T}\mathbf{f} + \mathbf{f}^{T}\mathbf{f})\right\}d\mathbf{f}  \\
   &= \int p(\mathbf{f}|\mathbf{f}_{m})\left\{-\frac{n}{2}\log(2\pi \sigma^{2}) - \frac{1}{2\sigma^{2}}Tr(\mathbf{y}\mathbf{y}^{T} - 2\mathbf{y}\mathbf{f}^{T} + \mathbf{f}\mathbf{f}^{T})\right\} d\mathbf{f} 
\end{align*}
Since $\int p(\mathbf{f}|\mathbf{f}_{m})d\mathbf{f} = 1$, and $\mathbf{y}\mathbf{y}^{T}$ doesn't involve $\mathbf{f}$, therefore, let us focus on the last two integrals.
\begin{equation}
    -\int p(\mathbf{f}|\mathbf{f}_{m})2\mathbf{y}\mathbf{f}^{T}d\mathbf{f} = -\mathbb{E}[2\mathbf{y}\mathbf{f}|\mathbf{f}_{m}] = -2\mathbf{y}\mathbb{E}[\mathbf{f}|\mathbf{f}_{m}] \label{eq:43}
\end{equation}
\begin{equation}
    \int p(\mathbf{f}|\mathbf{f}_{m})\mathbf{f}\mathbf{f}^{T}d\mathbf{f} = \mathbb{E}[\mathbf{f}\mathbf{f}^{T}|\mathbf{f}_{m}] = Cov(\mathbf{f}|\mathbf{f}_{m}) + \mathbb{E}[\mathbf{f}|\mathbf{f}_{m}]\mathbb{E}[\mathbf{f}|\mathbf{f}_{m}]^{T} \label{eq:44}
\end{equation}
From \cite{williams2006gaussian}, we know that, given:
\begin{equation}
\begin{bmatrix}
\mathbf{X} \\
\mathbf{Y}
\end{bmatrix} \sim \mathcal{N} \left(\begin{bmatrix}
\mu_{\mathbf{X}}\\
\mu_{\mathbf{Y}}
\end{bmatrix}, \begin{bmatrix}
\mathbf{\Sigma}_{\mathbf{XX}} & \mathbf{\Sigma}_{\mathbf{XY}} \\
\mathbf{\Sigma}_{\mathbf{YX}} & \mathbf{\Sigma}_{\mathbf{YY}}
\end{bmatrix} \right) \label{eq:45}
\end{equation}
The conditional expectation and covariance are calculated by:
\begin{equation}
    \mathbb{E}[\mathbf{X}|\mathbf{Y}] = \mu_{\mathbf{X}} + \Sigma_{\mathbf{XY}}\Sigma_{\mathbf{YY}}^{-1}(\mathbf{Y} - \mu_{\mathbf{Y}}) \label{eq:46}
\end{equation}
\begin{equation}
    Cov[\mathbf{X}|\mathbf{Y}] = \Sigma_{\mathbf{XX}} - \Sigma_{\mathbf{XY}}\Sigma_{\mathbf{YY}}^{-1}\Sigma_{\mathbf{YX}} \label{eq:47}
\end{equation}
In our case:
\begin{equation}
\begin{bmatrix}
\mathbf{f} \\
\mathbf{f}_{m}
\end{bmatrix} \sim \mathcal{N} \left(\begin{bmatrix}
\boldsymbol{0}\\
\boldsymbol{0}
\end{bmatrix}, \begin{bmatrix}
\mathbf{K}_{nn} & \mathbf{K}_{nm} \\
\mathbf{K}_{mn} & \mathbf{K}_{mm}
\end{bmatrix} \right) \label{eq:48}
\end{equation}
Therefore:
\begin{equation}
    -2\mathbf{y}\mathbb{E}[\mathbf{f}|\mathbf{f}_{m}] = -2\mathbf{y}\mathbf{K}_{nm}\mathbf{K}_{mm}^{-1}\mathbf{f}_{m} \label{eq:49}
\end{equation}
\begin{equation}
 \mathbb{E}[\mathbf{f}\mathbf{f}^{T}|\mathbf{f}_{m}] = Cov(\mathbf{f}|\mathbf{f}_{m}) + \mathbb{E}[\mathbf{f}|\mathbf{f}_{m}]\mathbb{E}[\mathbf{f}|\mathbf{f}_{m}]^{T} = \mathbf{K}_{nn} - \mathbf{K}_{nm}\mathbf{K}_{mm}^{-1}\mathbf{K}_{mn} + \mathbf{K}_{nm}\mathbf{K}_{mm}^{-1}\mathbf{f}_{m}\left(\mathbf{K}_{nm}\mathbf{K}_{mm}^{-1}\mathbf{f}_{m}\right)^{T} \label{eq:50}
\end{equation}
Therefore, $\log G(\mathbf{f}_{m},\mathbf{y})$ can be written as:
\begin{align*}
    \log G(\mathbf{f}_{m},\mathbf{y}) &= -\frac{n}{2}\log(2\pi \sigma^{2}) - \frac{1}{2\sigma^{2}}Tr(\mathbf{y}\mathbf{y}^{T} - 2\mathbf{y}\boldsymbol{\alpha}^{T} + \boldsymbol{\alpha}\boldsymbol{\alpha}^{T} + \mathbf{K}_{nn} - \mathbf{Q}_{nn}) \\
    &= \log[\mathcal{N}(\mathbf{y}|\boldsymbol{\alpha}, \sigma^{2}\mathbf{I})] -\frac{1}{2\sigma^{2}}Tr(\mathbf{K}_{nn} - \mathbf{Q}_{nn})
\end{align*}
where $\boldsymbol{\alpha} = \mathbf{K}_{nm}\mathbf{K}_{mm}^{-1}\mathbf{f}_{m}$,$\mathbf{Q}_{nn} = \mathbf{K}_{nm}\mathbf{K}_{mm}^{-1}\mathbf{K}_{mn}$. Therefore, $F_{V}(\mathbf{x}_{m},\phi)$ can be written as:
\begin{equation}
    F_{V}(\mathbf{x}_{m},\phi) = \int \phi(\mathbf{f}_{m})\log\frac{\mathcal{N}(\mathbf{y}|\boldsymbol{\alpha}, \sigma^{2}\mathbf{I})}{\phi(\mathbf{f}_{m})}d\mathbf{f}_{m} -\frac{1}{2\sigma^{2}}Tr(\mathbf{K}_{nn} - \mathbf{Q}_{nn}) \label{eq:51}
\end{equation}
Based on Jensen's inequality (\cite{jensen1906fonctions}), $\varphi(\mathbb{E}[\mathbf{X}]) = \mathbb{E}[\varphi(\mathbf{X})]$($\varphi$ is a convex function), taking the log function outside the above integral, we can have:
\begin{align*}
    F_{V}(\mathbf{x}_{m},\phi) &= \log\int \phi(\mathbf{f}_{m})\frac{\mathcal{N}(\mathbf{y}|\boldsymbol{\alpha}, \sigma^{2}\mathbf{I})}{\phi(\mathbf{f}_{m})}d\mathbf{f}_{m} -\frac{1}{2\sigma^{2}}Tr(\mathbf{K}_{nn} - \mathbf{Q}_{nn}) \\
    &= \log\int \mathcal{N}(\mathbf{y}|\boldsymbol{\alpha}, \sigma^{2}\mathbf{I})d\mathbf{f}_{m} -\frac{1}{2\sigma^{2}}Tr(\mathbf{K}_{nn} - \mathbf{Q}_{nn}) \\
    &= \log \mathcal{N}(\mathbf{y}|\boldsymbol{0}, \sigma^{2}\mathbf{I} +\mathbf{Q}_{nn}) -\frac{1}{2\sigma^{2}}Tr(\mathbf{K}_{nn} - \mathbf{Q}_{nn})
\end{align*}
\end{proof}
\begin{lemma}
\begin{equation}
    \phi(\mathbf{f}_{m}) = \mathcal{N}(\mathbf{f}_{m}|\boldsymbol{\mu}_{m},\mathbf{A}_{m}) \label{eq:52}
\end{equation} \label{lemma:4}
\end{lemma}
where:
\begin{equation}
  \boldsymbol{\mu}_{m} = \frac{1}{\sigma^{2}}\mathbf{K}_{mm}\mathbf{\Sigma}\mathbf{K}_{mn}\mathbf{y}    \label{eq:53}
\end{equation}
\begin{equation}
  \mathbf{A}_{m} = \mathbf{K}_{mm}\mathbf{\Sigma}\mathbf{K}_{mm}   \label{eq:54}
\end{equation}
\begin{equation}
\mathbf{\Sigma} = (\mathbf{K}_{mm} + \sigma^{-2}\mathbf{K}_{mn}\mathbf{K}_{nm})^{-1}  \label{eq:55}
\end{equation}
\begin{proof}
Based on lemma \ref{lemma:2} and \cite{titsias2009variational-1}, achieving the optimal bound occurs when equality holds in Jensen’s inequality. i.e. :
\begin{align*}
    \phi^{*}(\mathbf{f}_{m}) &\propto \mathcal{N}(\mathbf{y}|\boldsymbol{\alpha},\sigma^{2})p(\mathbf{f}_{m}) \\
     &\propto p(\mathbf{f}|\mathbf{f}_{m})p(\mathbf{f}_{m})\\
     &\propto \exp(\log p(\mathbf{f}|\mathbf{f}_{m}) + \log p(\mathbf{f}_{m}))
\end{align*}
Since $p(\mathbf{f}|\mathbf{f}_{m}) = \mathcal{N}(\mathbf{y}|\boldsymbol{\alpha},\sigma^{2})$,$p(\mathbf{f}_{m}) = \mathcal{N}(\mathbf{y}|\boldsymbol{0},\mathbf{K}_{mm})$, $\phi^{*}(\mathbf{f}_{m})$ can be further written as:
\begin{align*}
      \phi^{*}(\mathbf{f}_{m}) &= \exp \left\{-\log\left((2\pi)^{\frac{n}{2}}|\sigma^{2}|^{\frac{1}{2}} \right) - \frac{1}{2}(\mathbf{y} - \boldsymbol{\alpha})^{T}\frac{1}{\sigma^{2}}(\mathbf{y} - \boldsymbol{\alpha}) -\log\left((2\pi)^{\frac{n}{2}}|\mathbf{K}_{mm}|^{\frac{1}{2}} \right) - \frac{1}{2}\mathbf{f}_{m}^{T}\mathbf{K}_{mm}^{-1}\mathbf{f}_{m} \right\}  \\
      &=  c\exp \left\{ - \frac{1}{2\sigma^{2}}(\mathbf{y} - \mathbf{K}_{nm}\mathbf{K}_{mm}^{-1}\mathbf{f}_{m})^{T}(\mathbf{y} - \mathbf{K}_{nm}\mathbf{K}_{mm}^{-1}\mathbf{f}_{m}) - \frac{1}{2}\mathbf{f}_{m}^{T}\mathbf{K}_{mm}^{-1}\mathbf{f}_{m} \right\} \\
    &= c\exp \left\{ -\frac{1}{2\sigma^{2}}\mathbf{y}^{T}\mathbf{y}-\frac{1}{2}\mathbf{f}_{m}^{T}(\mathbf{K}_{mm}^{-1}+\frac{1}{\sigma^{2}}\mathbf{K}_{mm}^{-1}\mathbf{K}_{mn}\mathbf{K}_{nm}\mathbf{K}_{mm}^{-1})\mathbf{f}_{m} + \frac{1}{\sigma^{2}}\mathbf{y}^{T}\mathbf{K}_{nm}\mathbf{K}_{mm}^{-1}\mathbf{f}_{m}\right\} \\
    &= \mathcal{N}(\mathbf{f}_{m}|\left(\mathbf{K}_{mm}^{-1} + \frac{1}{\sigma^{2}}\mathbf{K}_{mm}^{-1}\mathbf{K}_{mn}\mathbf{K}_{nm}\mathbf{K}_{mm}^{-1}\right)^{-1}\frac{1}{\sigma^{2}}\mathbf{K}_{mn}\mathbf{y},\left(\mathbf{K}_{mm}^{-1} + \frac{1}{\sigma^{2}}\mathbf{K}_{mm}^{-1}\mathbf{K}_{mn}\mathbf{K}_{nm}\mathbf{K}_{mm}^{-1}\right)^{-1})
\end{align*}
Based on the Sherman–Morrison formula (\cite{sherman1950adjustment}), we can have:
\begin{equation}
    \left(\mathbf{K}_{mm}^{-1} + \frac{1}{\sigma^{2}}\mathbf{K}_{mm}^{-1}\mathbf{K}_{mn}\mathbf{K}_{nm}\mathbf{K}_{mm}^{-1}\right)^{-1} = \mathbf{K}_{mm}\left(\mathbf{K}_{mm} + \frac{1}{\sigma^{2}}\mathbf{K}_{mn}\mathbf{K}_{nm} \right)^{-1}\mathbf{K}_{mm} \label{eq:56}
\end{equation}
Therefore:
\begin{align*}
   \phi^{*}(\mathbf{f}_{m}) &= \mathcal{N}(\mathbf{f}_{m}|\sigma^{-2}\mathbf{K}_{mm}\mathbf{\Sigma}\mathbf{K}_{mn}\mathbf{y}, \mathbf{K}_{mm}\mathbf{\Sigma}\mathbf{K}_{mm})
\end{align*}
where $\mathbf{\Sigma} = (\mathbf{K}_{mm} + \sigma^{-2}\mathbf{K}_{mn}\mathbf{K}_{nm})^{-1}$,$c$ is a constant.
\end{proof}
\begin{lemma}
When the prior of $\mathbf{f}$ and $\mathbf{f}_{m}$ is not zero, then:
\begin{equation}
        \phi(\mathbf{f}_{m}) = \mathcal{N}(\mathbf{f}_{m}|\boldsymbol{\mu}_{*},\boldsymbol{\Sigma}_{*}) \label{eq:57}
\end{equation}
\begin{equation}
        \boldsymbol{\Sigma}_{*} = (\mathbf{K}_{mm}^{-1} + \frac{1}{2\sigma^{2}}\mathbf{K}_{mm}^{-1}\mathbf{K}_{mn}\mathbf{K}_{nm}\mathbf{K}_{mm}^{-1})^{-1}  \label{eq:58}
\end{equation}
\begin{equation}
        \boldsymbol{\mu}_{*} = \mathbf{\Sigma}_{*}\left(\frac{1}{\sigma^{2}}\mathbf{K}_{mm}^{-1}\mathbf{K}_{nm}(\mathbf{y} - m(\mathbf{x}))\right) \label{eq:59}
\end{equation} \label{lemma:5}
\end{lemma}
\begin{proof}
Since the prior the prior of $\mathbf{f}$ and $\mathbf{f}_{m}$ is not zero, then the joint distribution changed from:
\begin{equation}
\begin{bmatrix}
\mathbf{f} \\
\mathbf{f}_{m}
\end{bmatrix} \sim \mathcal{N} \left(\begin{bmatrix}
\boldsymbol{0}\\
\boldsymbol{0}
\end{bmatrix}, \begin{bmatrix}
\mathbf{K}_{nn} & \mathbf{K}_{nm} \\
\mathbf{K}_{mn} & \mathbf{K}_{mm}
\end{bmatrix} \right) \label{eq:60}
\end{equation}
to:
\begin{equation}
\begin{bmatrix}
\mathbf{f} \\
\mathbf{f}_{m}
\end{bmatrix} \sim \mathcal{N} \left(\begin{bmatrix}
m(\mathbf{x})\\
m(\mathbf{x}_{m})
\end{bmatrix}, \begin{bmatrix}
\mathbf{K}_{nn} & \mathbf{K}_{nm} \\
\mathbf{K}_{mn} & \mathbf{K}_{mm}
\end{bmatrix} \right) \label{eq:61}
\end{equation}
Then equations \ref{eq:50} and \ref{eq:51} will be written as:
\begin{equation}
    -2\mathbf{y}\mathbb{E}[\mathbf{f}|\mathbf{f}_{m}] = -2\mathbf{y}\left\{m(\mathbf{x}) + \mathbf{K}_{nm}\mathbf{K}_{mm}^{-1}(\mathbf{f}_{m} - m(\mathbf{x}_{m})) \right\} \label{eq:62}
\end{equation}
\begin{equation}
 \mathbb{E}[\mathbf{f}\mathbf{f}^{T}|\mathbf{f}_{m}] =  \mathbf{K}_{nn} - \mathbf{K}_{nm}\mathbf{K}_{mm}^{-1}\mathbf{K}_{mn} + \left\{m(\mathbf{x}) + \mathbf{K}_{nm}\mathbf{K}_{mm}^{-1}(\mathbf{f}_{m} - m(\mathbf{x}_{m})) \right\}\left\{m(\mathbf{x}) + \mathbf{K}_{nm}\mathbf{K}_{mm}^{-1}(\mathbf{f}_{m} - m(\mathbf{x}_{m})) \right\}^{T} \label{eq:63}
\end{equation}
Therefore, $\log G(\mathbf{f}_{m},\mathbf{y})$ can be written as:
\begin{align*}
    \log G(\mathbf{f}_{m},\mathbf{y}) &= -\frac{n}{2}\log(2\pi \sigma^{2}) - \frac{1}{2\sigma^{2}}Tr(\mathbf{y}\mathbf{y}^{T} - 2\mathbf{y}\boldsymbol{\alpha}^{T} + \boldsymbol{\alpha}\boldsymbol{\alpha}^{T} + \mathbf{K}_{nn} - \mathbf{Q}_{nn}) \\
    &= \log[\mathcal{N}(\mathbf{y}|\boldsymbol{\alpha}, \sigma^{2}\mathbf{I})] -\frac{1}{2\sigma^{2}}Tr(\mathbf{K}_{nn} - \mathbf{Q}_{nn})
\end{align*}
where $\boldsymbol{\alpha} = m(\mathbf{x}) + \mathbf{K}_{nm}\mathbf{K}_{mm}^{-1}(\mathbf{f}_{m} - m(\mathbf{x}_{m}))$,$\mathbf{Q}_{nn} = \mathbf{K}_{nm}\mathbf{K}_{mm}^{-1}\mathbf{K}_{mn}$. Therefore, $F_{V}(\mathbf{x}_{m},\phi)$ can be written as:
\begin{equation}
    F_{V}(\mathbf{x}_{m},\phi) = \int \phi(\mathbf{f}_{m})\log\frac{\mathcal{N}(\mathbf{y}|\boldsymbol{\alpha}, \sigma^{2}\mathbf{I})}{\phi(\mathbf{f}_{m})}d\mathbf{f}_{m} -\frac{1}{2\sigma^{2}}Tr(\mathbf{K}_{nn} - \mathbf{Q}_{nn}) \label{eq:64}
\end{equation}
Similarly, based on Jensen's inequality, taking the log function outside the above integral, we can have:
\begin{align*}
    F_{V}(\mathbf{x}_{m},\phi) &= \log\int \phi(\mathbf{f}_{m})\frac{\mathcal{N}(\mathbf{y}|\boldsymbol{\alpha}, \sigma^{2}\mathbf{I})}{\phi(\mathbf{f}_{m})}d\mathbf{f}_{m} -\frac{1}{2\sigma^{2}}Tr(\mathbf{K}_{nn} - \mathbf{Q}_{nn}) \\
    &= \log\int \mathcal{N}(\mathbf{y}|\boldsymbol{\alpha}, \sigma^{2}\mathbf{I})d\mathbf{f}_{m} -\frac{1}{2\sigma^{2}}Tr(\mathbf{K}_{nn} - \mathbf{Q}_{nn}) \\
    &= \log \mathcal{N}(\mathbf{y}|\boldsymbol{0}, \sigma^{2}\mathbf{I} +\mathbf{Q}_{nn}) -\frac{1}{2\sigma^{2}}Tr(\mathbf{K}_{nn} - \mathbf{Q}_{nn})
\end{align*}
which has the same structure as when the prior is zero, but $\boldsymbol{\alpha}$ is different. Also, based on lemma \ref{lemma:4}, we can recalculate $\phi(\mathbf{f}_{m})$ as:
\begin{align*}
      \phi^{*}(\mathbf{f}_{m}) &= \exp \left\{-\log\left((2\pi)^{\frac{n}{2}}|\sigma^{2}|^{\frac{1}{2}} \right) - \frac{1}{2}(\mathbf{y} - \boldsymbol{\alpha})^{T}\frac{1}{\sigma^{2}}(\mathbf{y} - \boldsymbol{\alpha}) -\log\left((2\pi)^{\frac{n}{2}}|\mathbf{K}_{mm}|^{\frac{1}{2}} \right) - \frac{1}{2}(\mathbf{f}_{m} - m(\mathbf{x}_{m}))^{T}\mathbf{K}_{mm}^{-1}(\mathbf{f}_{m} - m(\mathbf{x}_{m})) \right\}  \\
      &=  c\exp \left\{  -\frac{1}{2}(\mathbf{y} - \boldsymbol{\alpha})^{T}\frac{1}{\sigma^{2}}(\mathbf{y} - \boldsymbol{\alpha}) -  \frac{1}{2}(\mathbf{f}_{m} - m(\mathbf{x}_{m}))^{T}\mathbf{K}_{mm}^{-1}(\mathbf{f}_{m} - m(\mathbf{x}_{m}))\right\} \\
    &= c \exp \{ \frac{1}{2\sigma^{2}} (\mathbf{y} - m(\mathbf{x}))^{T} (\mathbf{y} - m(\mathbf{x})) - \frac{1}{\sigma^{2}} (\mathbf{y} - m(\mathbf{x}))^{T} \mathbf{K}_{nm} \mathbf{K}_{mm}^{-1} (\mathbf{f}_{m} - m(\mathbf{x}_{m})) \\
    &- \frac{1}{2} (\mathbf{f}_{m} - m(\mathbf{x}_{m}))^{T} \left( \mathbf{K}_{mm}^{-1} + \frac{1}{\sigma^{2}} \mathbf{K}_{mm}^{-1} \mathbf{K}_{mn} \mathbf{K}_{nm} \mathbf{K}_{mm}^{-1} \right) (\mathbf{f}_{m} - m(\mathbf{x}_{m})) \}
 \\
    &= \mathcal{N}(\mathbf{f}_{m}|\boldsymbol{\mu}_{*},\mathbf{\Sigma}_{*})
\end{align*}
where $\mathbf{\Sigma}_{*} = \left(\mathbf{K}_{mm}^{-1}+\frac{1}{\sigma^{2}}\mathbf{K}_{mm}^{-1}\mathbf{K}_{mn}\mathbf{K}_{nm}\mathbf{K}_{mm}\right)^{-1}$,$\boldsymbol{\mu}_{*} = m(\mathbf{x}_m) + \mathbf{\Sigma}_{*}\left(\frac{1}{\sigma^{2}}\mathbf{K}_{mm}^{-1}\mathbf{K}_{nm}(\mathbf{y} - m(\mathbf{x}))\right)$ ,$c$ is a constant.
\end{proof}
\begin{lemma}
When the prior of $\mathbf{f}_{*}$ and $\mathbf{f}_{m}$ is not zero, then:
\begin{equation}
    q(\mathbf{f}_{*}) = \mathcal{N}(m(\mathbf{x}_{*}) + \mathbf{K}_{*m}\mathbf{K}_{mm}^{-1}(\boldsymbol{\mu}_{*}-m(\mathbf{x}_{m}),\mathbf{K}_{**}-\mathbf{K}_{*m}\mathbf{K}_{mm}^{-1}\mathbf{K}_{m*} + \mathbf{K}_{*m}\mathbf{K}_{mm}^{-1}\boldsymbol{\mu}_{*}\mathbf{K}_{mm}^{-1}\mathbf{K}_{m*}  \label{eq:65}
\end{equation} \label{lemma:6}
\end{lemma}
\begin{proof}
Since the joint distribution is now written as:
\begin{equation}
\begin{bmatrix}
\mathbf{f}_{*} \\
\mathbf{f}_{m}
\end{bmatrix} \sim \mathcal{N} \left(\begin{bmatrix}
m(\mathbf{x}_{*})\\
m(\mathbf{x}_{m})
\end{bmatrix}, \begin{bmatrix}
\mathbf{K}_{**} & \mathbf{K}_{*m} \\
\mathbf{K}_{m*} & \mathbf{K}_{mm}
\end{bmatrix} \right) \label{eq:66}
\end{equation}
Based on conditional expectation and covariance formulation shown in equation \ref{eq:49} and \ref{eq:50}, and $q(\mathbf{f}_{*})  =  \int p(\mathbf{f}_{*}|\mathbf{f}_{m})\phi(\mathbf{f}_{m})d\mathbf{f}_{m}$
\begin{equation}
    \mathbf{f}_{*}|\mathbf{f}_{m} \sim \mathcal{N}(m(\mathbf{x}_{*}) + \mathbf{K}_{*m}\mathbf{K}_{mm}^{-1}(\boldsymbol{\mu}_{*}-m(\mathbf{x}_{m}),\mathbf{K}_{**}-\mathbf{K}_{*m}\mathbf{K}_{mm}^{-1}\mathbf{K}_{m*}) \label{eq:67}
\end{equation}
\begin{equation}
    q(\mathbf{f}_{*}) =  \mathcal{N}(m(\mathbf{x}_{*}) + \mathbf{K}_{*m}\mathbf{K}_{mm}^{-1}(\boldsymbol{\mu}_{*}-m(\mathbf{x}_{m}),\mathbf{K}_{**}-\mathbf{K}_{*m}\mathbf{K}_{mm}^{-1}\mathbf{K}_{m*} + \mathbf{K}_{*m}\mathbf{K}_{mm}^{-1}\boldsymbol{\mu}_{*}\mathbf{K}_{mm}^{-1}\mathbf{K}_{m*}) \label{eq:68}
\end{equation}
\end{proof}
\begin{definition}[Simple random sampling]
\par Simple random sampling is the basic sampling technique where we select a group of subjects for study from a larger group (\cite{easton1997statistics}). Each individual is chosen entirely by chance, and each member of the population has an equal chance of being included in the sample. Every possible sample of a given size has the same chance of selection. 
\par The simple random sampling without replacement that is adopted to do the comparison study is the famous reservoir sampling \cite{vitter1985random} shown as follows:
\begin{algorithm}[htbp]
    \caption{Reservoir sampling}
    \label{algorithm:1}
    \LinesNumbered
    \KwIn {Input: $\mathcal{P}$[1,2,..,$n$],$S$[1,2,...,$k$] }{
        \For{i $\in$ \textbf{Range} (1 \textbf{to} $k$+1)}{
            $\mathcal{S}[i] = \mathcal{P}[i]$
        }
        \For{i $\in$ \textbf{Range} ($k$+1 \textbf{to} $n$+1)}{
        $j$ = RandomInteger(1, $i$) \\
        \If{$j$ $\leq$ $k$}{
            $\mathcal{S}[j] = \mathcal{P}[i]$
        }
        }
    }
\end{algorithm}
\end{definition}
\begin{definition}[Systematic sampling]
\par Systematic sampling is a statistical method involving the selection of elements from an ordered sampling frame. The most common form of systematic sampling is an equiprobability method, which is shown as follows:
\begin{algorithm}[htbp]
    \caption{Systematic sampling}
    \label{algorithm:2}
    \LinesNumbered
    \KwIn {Input: $\mathcal{P}$[1,2,..,$N$],$\mathcal{S}$ = $\varnothing$,$n$ }{
        $k = \frac{N}{n}$ \\
        Randomly choose an index $I$, $I \in \lceil[1,k]\rceil$ \\
        $\mathcal{S}[1] = \mathcal{P}[I]$ \\
        \For{i $\in$ \textbf{Range} (2 \textbf{to} $k$+1)}{
            $\mathcal{R}[i] = \mathcal{S}[\lceil I + k * (i - 1) \rceil]$
        }
    }
\end{algorithm}
where $N$ is the population size, $n$ is the sample size. To avoid the risk that the last house chosen does not exist ($I + k * (n - 1) > N$), the initial index $I$ must be located within $[1,k]$ and should be rounded up to the next integer if it is not an integer.
\end{definition}
\begin{definition}[Cluster sampling]
\par Cluster sampling is a method used when mutually homogeneous yet internally heterogeneous groupings are evident in a statistical population. Based on k-mean algorithm \cite{macqueen1967some}, the cluster sampling can be expressed as:
\begin{algorithm}[htbp]
    \caption{Cluster sampling algorithm}
    \label{algorithm:3}
    \LinesNumbered
    \KwIn {Input: $\mathcal{P} = \{x_{1},x_{2},...,x_{n}\}$,$\mathcal{C}_{i} = \varnothing, i = 1, 2,..., k$,$\mathcal{S} = \varnothing$,\textbf{M} }{
        Randomly selecting $k$ random observations $m_{i}^{(1)}, i = 1,2,...,k$ from $\mathcal{P}$ \\
        \For{$t$ \textbf{Range}(1, \textbf{M}+1)}{
            \For{$x_{p}$ \textbf{in} $\mathcal{P}$}{
                \For{$i$ = \textbf{Range}(1,$k$+1)}{
                    \If{$\|x_{p} - m_{i}^{(t)}\|^{2} \leq \|x_{p} - m_{j}^{(t)}\|^{2}\ \forall 1 \leq j \leq k$}{
                        Add $x_{p}$ to $C_{i}^{(t)}$ and remove $x_{p}$ from $\mathcal{P}$
                    }
                }
            }
            \For{$i$ = \textbf{Range}(1,$k$+1)}{
                $m_{i}^{(t+1)} = \frac{1}{|\mathcal{C}_{i}^{t}|}\sum_{x \in \mathcal{C}_{i}^{t}}x$
            }   
        }
        \For{$i$ = \textbf{Range}(1,$k$+1)}{
            $\mathcal{S}[i] = x^{*}, \{x^{*}| \textbf{min}_{x^{*}}:\|x - m_{i}^{\textbf{M}}\|^{2}, \forall x \in \mathcal{C}_{i}^{\textbf{M}}\}$
        } 
    }
\end{algorithm}
\end{definition}
\begin{definition}[Weighted random sampling] \label{definition:5}
\par This sampling method is motivated by \cite{qu2015fundamental}, which aims to avoid selecting high percentage samples with a relatively low traffic density. The core idea is that first sort the $\mathcal{S} =\{(\rho_{1},v_{1}),(\rho_{2},v_{2}),...,(\rho_{n},v_{n})\}$ with a density ascending order. And then compute the weight $w$ for each density-speed pair based on \cite{qu2015fundamental}. Then for each density-speed pair $(\rho_{i},v_{i})$, it will has a probability $\frac{w_{i}}{\sum_{j = 1}^{n}w_{j}}$ to be selected.
\end{definition}
\bibliography{main}
\end{document}